%% file: gaussianScoresFull.tex
\documentclass[11pt]{article}
\usepackage[margin=1in,footskip=0.25in]{geometry}

\usepackage[utf8]{inputenc} % allow utf-8 input
\usepackage[T1]{fontenc}    % use 8-bit T1 fonts
\usepackage[hidelinks]{hyperref}       % hyperlinks
\usepackage{url}            % simple URL typesetting
\usepackage{booktabs}       % professional-quality tables
\usepackage{amsfonts}       % blackboard math symbols
\usepackage{nicefrac}       % compact symbols for 1/2, etc.
\usepackage{microtype}      % microtypography
\usepackage{amssymb,amsthm,amsmath}
\usepackage{color}
\usepackage[usenames,dvipsnames,svgnames,table]{xcolor}
\usepackage{graphicx}
\usepackage{float}
\usepackage{caption}
\usepackage{subcaption}

\usepackage{bbm}

\usepackage{enumitem}
\usepackage{algorithm}
\usepackage{algpseudocode}

\newif\ifdraft
\draftfalse
\newcommand{\eqdef}{\mathbin{\stackrel{\rm def}{=}}}
\makeatletter
\def\hlinewd#1{%
	\noalign{\ifnum0=`}\fi\hrule \@height #1 \futurelet
	\reserved@a\@xhline}
\makeatother

\newtheorem{theorem}{Theorem}

\newtheorem{corollary}[theorem]{Corollary}
\newtheorem{lemma}[theorem]{Lemma}
\newtheorem{fact}[theorem]{Fact}
\newtheorem{claim}[theorem]{Claim}

\newtheorem{definition}{Definition}
\newtheorem{problem}[definition]{Problem}

\newtheorem*{rep@theorem}{\rep@title}
\newcommand{\newreptheorem}[2]{%
	\newenvironment{rep#1}[1]{%
		\def\rep@title{#2 \ref{##1}}%
		\begin{rep@theorem}}%
		{\end{rep@theorem}}}
\makeatother
\newreptheorem{theorem}{Theorem}
\newreptheorem{claim}{Claim}

\newcommand{\R}{\mathbb{R}}
\newcommand{\C}{\mathbb{C}}

\newcommand{\bs}[1]{\boldsymbol{#1}}
\newcommand{\bv}[1]{\mathbf{#1}}

\newcommand{\norm}[1]{\|#1\|}
\newcommand{\opnorm}[1]{\|#1\|_\mathrm{op}}
\DeclareMathOperator{\supp}{\mathrm{supp}}
\DeclareMathOperator{\poly}{poly}

\DeclareMathOperator*{\argmin}{arg\,min}
\newcommand{\Kmu}{\mathcal{K}_{p,q}}

\newcommand{\Fmu}{\mathcal{F}_{p,q}}

\newcommand{\Fp}{\mathcal{F}_{p,q}}

\newcommand{\smu}{s_{p,q,\lambda}}
\newcommand{\tmu}{\tau_{p,q,\lambda}}
\newcommand{\ttmu}{\bar {\tau}_{p,q,\lambda}}

\newcommand{\E}{\mathbb{E}}

\DeclareMathOperator{\tr}{tr}

\input{bib_macros}

\title{Fourier Sparse Leverage Scores and Approximate Kernel Learning}

\author{
	Tam\'{a}s Erd\'{e}lyi \\ Texas A\&M University\\ \texttt{terdelyi@math.tamu.edu}
	\and 
	Cameron Musco\\ UMass Amherst\\ \texttt{cmusco@cs.umass.edu}
	\and
	Christopher Musco\\ New York University\\ \texttt{cmusco@nyu.edu}
}

\begin{document}
\maketitle

\begin{abstract}
We prove new explicit upper bounds on the leverage scores of Fourier sparse functions under both the Gaussian and Laplace measures. In particular, we study $s$-sparse functions of the form $f(x) = \sum_{j=1}^s a_j e^{i \lambda_j x}$ for coefficients $a_j \in \C$ and frequencies $\lambda_j \in \R$.
Bounding Fourier sparse leverage scores under various measures is of pure mathematical interest in approximation theory, and our work extends existing results for the uniform measure \cite{Erdelyi:2017, ChenPrice:2019a}. Practically, our bounds are motivated by two important applications in machine learning:

\smallskip

\noindent \textbf{1. Kernel Approximation.} They yield a new random Fourier features algorithm for approximating Gaussian  and Cauchy (rational quadratic) kernel matrices. For low-dimensional data, our method uses a near optimal number of features, and its runtime is polynomial in 
the \emph{statistical dimension} of the approximated kernel matrix. It is the first ``oblivious sketching method'' with this property for any kernel besides the polynomial kernel, resolving an open question of \cite{AvronKapralovMusco:2017, AhleKapralovKnudsen:2020}. 
	
	\smallskip

\noindent \textbf{2. Active Learning.} They can be used as non-uniform sampling distributions for robust active learning when data follows a Gaussian or Laplace distribution. Using the framework of \cite{AvronKapralovMusco:2019}, we provide essentially optimal results for bandlimited and multiband interpolation, and Gaussian process regression. These results generalize existing work that only applies to uniformly distributed data. 
\end{abstract}

%\newpage
\section{Introduction}

Statistical leverage scores have emerged as an important tool in machine learning and algorithms, with applications including randomized numerical linear algebra \cite{DrineasMahoneyMuthukrishnan:2006,Sarlos:2006}, efficient kernel methods \cite{AlaouiMahoney:2015,MuscoMusco:2017,AvronKapralovMusco:2017,li2018towards,shahrampour2019sampling,liu2019random,fanuel2019nystr,kammonen2020adaptive}, graph algorithms \cite{SpielmanSrivastava:2011,KyngSachdeva:2016}, active learning \cite{DerezinskiWarmuthHsu:2018,ChenVarmaSingh:2016,MaMahoneyYu:2015,AvronKapralovMusco:2019}, and faster constrained and unconstrained optimization \cite{LeeSidford:2015,AgarwalKakadeKidambi:2020}.

The purpose of these scores is to quantify how large the magnitude of a function in a particular class can be at a \emph{single location}, in comparison to the \emph{average} magnitude of the function. 
In other words, they measure how ``spiky'' a function can be.
The function class might consist of all vectors $\bv{y}\in\R^n$ which can be written as $\bv{Ax}$ for a fixed $\bv{A}\in \R^{n\times d}$, all degree $q$ polynomials, all functions with bounded norm in some kernel Hilbert space, or (as in this paper) all functions that are $s$-sparse in the Fourier basis. By quantifying \emph{where} and \emph{how much} such functions can spike to large magnitude, leverage scores help us approximate and reconstruct functions via sampling, leading to provably accurate algorithms for a variety of problems.

Formally, for any class $\mathcal{F}$ of functions mapping some domain $\mathcal{S}$ to the complex numbers $\C$, and any probability density $p$ over $\mathcal{S}$, the leverage score $\tau_{\mathcal{F},p}(x)$ for $x\in \mathcal{S}$ is:
\begin{align}
\label{eq:gen_lev_score}
\tau_{\mathcal{F},p}(x) = \sup_{f\in \mathcal{F}: \norm{f}_p^2 \neq 0} \frac{|f(x)|^2 \cdot p(x)}{\norm{f}_p^2}\text{ where } \norm{f}_p^2 = \int_{y\in \mathcal S} |f(y)|^2 \cdot p(y) \ dy.
\end{align}
Readers who have seen leverage scores in the context of machine learning and randomized algorithms \cite{SpielmanSrivastava:2011,MaMahoneyYu:2015,DrineasMahoney:2016} may be most familiar with the setting where $\mathcal{F}$ is the set of all length $n$ vectors (functions from $\{1,\ldots, n\}\rightarrow \R$) which can be written as $\bv{Ax}$ for a fixed matrix $\bv{A}\in \R^{n\times d}$. In this case, $p$ is taken to be a discrete uniform density over indices $1, \ldots, n$, and it is not hard to check that \eqref{eq:gen_lev_score} is equivalent to more familiar definitions of ``matrix leverage scores''.\footnote{In particular, \eqref{eq:gen_lev_score} is equivalent to the definition $\tau_{\mathcal{F},p}(i) = \bv{a}_i^T(\bv{A}^T\bv{A})^{-1}\bv{a}_i$ where $\bv{a}_i$ is the $i^\text{th}$ row of $\bv{A}$, and to $\tau_{\mathcal{F},p}(i) = \|u_i\|_2^2$, where $u_i$ is the $i^\text{th}$ row of any orthogonal span for $\bv{A}$'s columns. See \cite{AvronKapralovMusco:2017} for details.}

When $\mathcal{F}$ is the set of all degree $q$ polynomials, the inverse of the leverage scores is known as the \emph{Christoffel function}. In approximation theory, Christoffel functions are widely studied for different densities $p$ (e.g., Gaussian on $\R$ or uniform on $[-1,1]$) due to their connection to orthogonal polynomials \cite{Nevai:1986}. Recently, they have found applications in active polynomial regression
\cite{RauhutWard:2012,HamptonDoostan:2015,ChkifaCohenMigliorati:2015,CohenMigliorati:2017} and more broadly in machine learning \cite{PauwelsBachVert:2018,LasserrePauwels:2019}.
 
We study leverage scores for the class of \emph{Fourier sparse functions}. In particular, we define:\footnote{It can be observed that any degree $s$ polynomial can be approximated to arbitrarily high accuracy by a function in $\mathcal{T}_s$, by driving the frequencies $\lambda_1,\ldots, \lambda_s$ to zero and taking a Taylor expansion. So the leverage scores of $\mathcal{T}_s$ actually upper bound those of the degree $s$ polynomials \cite{ChenKanePrice:2016}.}
\begin{align}\label{eq:ts}
\mathcal{T}_s = \left \{f: f(x) = \sum_{j=1}^s a_j e^{i \lambda_j x}, a_j \in \C, \lambda_j \in \R \right \},
\end{align}
where each $\lambda_j$ is the frequency of a complex exponential with coefficient $a_j$.
For ease of notation we will denote the leverage scores of $\mathcal{T}_s$ for a distribution $p$ as $\tau_{s,p}(x)$ instead of the full $\tau_{\mathcal{T}_s,p}(x)$.

In approximation theory, the Fourier sparse leverage scores have been studied extensively, typically when $p$ is the uniform density on a finite interval \cite{Turan:1984,Nazarov:1993,BorweinTamas:1996,Kos:2008,lubinsky2015dirichlet,Erdelyi:2017}.
Recently, these scores have also become of interest in algorithms research due to their value in designing sparse recovery and sparse FFT algorithms in the ``off-grid'' regime \cite{ChenKanePrice:2016,ChenPrice:2019,ChenPrice:2019a}. They have also found applications in active learning for bandlimited interpolation, Gaussian process regression, and covariance estimation \cite{AvronKapralovMusco:2019,meyer2020statistical,EldarLiMusco:2020}.

\subsection{Closed form leverage score bounds}
When studying the leverage scores of a function class over a domain $\mathcal{S}$, one of the primary objectives is to determine the scores for all $x\in \mathcal{S}$. This can be challenging for two reasons: 
\begin{itemize}
	\item For finite domains (e.g., functions on $\mathcal{S} = \{1,\ldots, n\}$) it may be possible to directly solve the optimization problem in \eqref{eq:gen_lev_score}, but doing so is often computationally expensive. 
	\item For infinite domains (e.g., functions on $\mathcal{S} = [-1,1]$),  $\tau_{\mathcal{F},p}(x)$ is itself a function over $\mathcal{S}$, and typically does not have a simple closed form that is amenable to applications. 
\end{itemize}
Both of these challenges are addressed by shifting the goal from \emph{exactly determining} $\tau_{\mathcal{F},p}(x)$ to \emph{upper bounding} the leverage score function. In particular, the objective is to find some function $\bar{\tau}_{\mathcal{F},p}$ such that $\bar{\tau}_{\mathcal{F},p}(x) \geq {\tau}_{\mathcal{F},p}(x)$ for all $x\in \mathcal{S}$ and $\int_{x\in\mathcal{S}} \bar{\tau}_{\mathcal{F},p}(x)dy$ is as small as possible.

For linear functions over finite domains, nearly tight upper bounds on the leverage scores
% (which sum to within a constant factor of the true scores) 
 can be computed more quickly than the true scores \cite{MahoneyDrineasMagdon-Ismail:2012,CohenLeeMusco:2015}. Over infinite domains, it is possible to prove for some function classes that $\bar{\tau}_{\mathcal{F},p}(x)$ is always less than some fixed value $C$, sometimes called a Nikolskii constant or coherence parameter \cite{HamptonDoostan:2015,Migliorati:2015,AdcockCardenas:2020}. In other cases, simple closed form expressions can be proven too upper bound the leverage scores. For example, when $\mathcal{F}$ is the class of degree $q$ polynomials and $p$ is uniform on $[-1,1]$, the (scaled) Chebyshev density $\bar{\tau}_{\mathcal{F},p}(x) = \frac{2(q+1)}{\pi\sqrt{1-x^2}}$ upper bounds the leverage scores \cite{Lorch:1983,AvronKapralovMusco:2019}. 

\subsection{Our results}
\label{sec:our_results}
The main mathematical results of this work are new upper bounds on the leverage scores $\tau_{s,p}(\cdot)$ of the class of $s$-sparse Fourier functions $\mathcal{T}_s$, when $p$ is a Gaussian or Laplace distribution. These bounds extend known results for the uniform distribution, and are proven by leveraging several results from approximation theory on concentration properties of exponential sums \cite{Turan:1984,BorweinErdelyi:1995,BorweinErdelyi:2006,Erdelyi:2017}.
%relating the leverage scores under these densities to function concentration properties under the uniform density on an interval and leveraging a rich line of work studying the uniform density setting. 
We highlight the applicability of our bounds by developing two applications in machine learning:

\paragraph{Kernel Approximation (Section \ref{sec:oblivious}).} %Our first application is detailed in Section \ref{sec:oblivious}. 
We show that our leverage score upper bounds can be used as importance sampling probabilities to give a modified random Fourier features algorithm \cite{RahimiRecht:2007} with essentially tight spectral approximation bounds for Gaussian and Cauchy (rational quadratic) kernel matrices. In fact, we give a black-box reduction, proving that an upper bound on the Fourier sparse leverage scores for a distribution $p$ immediately yields an algorithm for approximating kernel matrices with kernel function equal to the Fourier transform of $p$. This reduction leverages tools from randomized numerical linear algebra, in particular column subset selection results \cite{drineas2006subspace,guruswami2012optimal}. We use these results to show that Fourier sparse functions can universally well approximate kernel space functions, and in turn that the leverage scores of these kernel  functions can be bounded using our Fourier sparse leverage score bounds. % with Fourier sparse functions was first introduced in \cite{AvronKapralovMusco:2019} in the context of active regression, and we believe it will be a powerful tool in understanding kernel approximation going forward.

Our results make progress on a central open question on the power of oblivious sketching methods in kernel approximation: in particular, whether oblivious methods like random Fourier features and TensorSketch \cite{PhamPagh:2013,chen2017relative,phillips2020gaussiansketch} can match the performance of non-oblivious methods like Nystr\"{o}m approximation \cite{GittensMahoney:2013, AlaouiMahoney:2015,MuscoMusco:2017}. This question was essentially closed for the polynomial kernel in \cite{AhleKapralovKnudsen:2020}. We give a positive answer for Gaussian and Cauchy kernels in one dimension. 
%Our results 
%fit into a larger line of work on both oblivious and non-oblivious leverage score-based random features methods for kernel matrix and neural network approximation%For the Gaussian and Cauchy kernels, they give the strongest known results that do not depend polynomially on an additive error parameter.\footnote{See Section \ref{sec:oblivious} for details -- a number of existing bounds \cite{AvronKapralovMusco:2017,li2018towards} have sample complexity depending linearly on $1/\lambda$, where $\lambda$ is the additive error parameter of Definition \ref{def:spectral_gaurantee}.}

\paragraph{Active Learning (Section \ref{sec:active}).} It is well known that leverage scores can be used in \emph{active sampling methods} to reduce the statistical complexity of linear function fitting problems like polynomial regression or Gaussian process (GP) regression \cite{ChenPrice:2019a,CohenMigliorati:2017}. The scores must be chosen with respect to the underlying data distribution $\mathcal{D}$ to obtain an accurate function fit under that distribution \cite{PauwelsBachVert:2018}. Theorems \ref{thm:gaussian} and \ref{thm:exp} immediately yield new active sampling results for regression problems involving $s$ arbitrary complex exponentials when the data follows a Gaussian or Laplacian distribution. 

While this result may sound specialized, it's actually quite powerful due to recent work of \cite{AvronKapralovMusco:2019}, which gives a black-box reduction from active sampling for Fourier-sparse regression to active sampling for a wide variety of problems in signal processing and Bayesian learning, including bandlimited function fitting and GP regression. Plugging our results into this framework gives algorithms with essentially optimal statistical complexity: the number of samples required depends on a natural statistical dimension parameter of the problem that is tight in many cases.

We note that any future Fourier sparse leverage score bounds proven for different distributions (beyond Gaussian, Laplace, and uniform) would generalize our applications to new kernel matrices and data distributions. Finally, while our contributions are primarily theoretical, we present experiments on kernel sketching in Section \ref{sec:experiments}. We study a 2-D Gaussian process regression problem, representative of typical data-intensive function interpolation tasks, showing that our oblivious sketching method substantially improves on the original random Fourier features method on which it is based \cite{RahimiRecht:2007}.%, allowing for fast approximate inversion of large, ill-conditioned kernel matrices.

\subsection{Notation}
%\begin{definition}[Kernel Ridge Leverage Score -- Data Point]
%\begin{align*}
%\tau_{\lambda,p,k}(z) = \sup_{f \in \mathcal{H}} \frac{|f(z)|^2 \cdot p(z)}{ \int_{\R^d} f(x)^2 p(x) \d x + \lambda \norm{f}_{\mathcal{H}}^2}.
%\end{align*}
%\end{definition}+

%The assumption that a function lies in $L_2(p)$ or $L_2$ will typically be implicit in our definitions where a finite norm is assumed, and we will state explicitly where neccesary.
Boldface capital letters denote matrices or quasi-matrices (linear maps from finite-dimensional vector spaces to infinite-dimensional function spaces). Script letters denote infinite-dimensional operators. Boldface lowercase letters denote vectors or vector-valued functions. Subscripts identify the entries of these objects. E.g., $\bv{M}_{j,k}$ is the $(j,k)$ entry of matrix $\bv{M}$ and $\bv{z}_j$ is the $j^{\text{th}}$ entry of vector $\bv{z}$.
$\bv{I}$ denotes the identity matrix. $\preceq$ denotes the Loewner ordering on positive semidefinite (PSD) matrices: $N \preceq M$ means that $M-N$ is PSD. $\bv{A}^*$ denotes the conjugate transpose of a vector or matrix.

\section{Fourier Sparse Leverage Score Bounds}
\label{sec:main_tech}
We now state our main leverage score bounds for the Gaussian and Laplace distributions. These theorems are of mathematical interest and form the cornerstone of our applications in kernel learning. We defer their proofs to Section \ref{app:bounds}.
\begin{theorem}[Gaussian Density Leverage Score Bound]\label{thm:gaussian}
	Consider the Gaussian density $g(x) = \frac{1}{\sigma\sqrt{2\pi}} e^{-x^2/(2\sigma^2)}$ and let:
	%\vspace{-.5em} 
	\begin{align*}
	\bar \tau_{s,g}(x) = \begin{cases} \frac{1}{\sqrt{2}\sigma} \cdot e^{-x^2/(4\sigma^2)} \text{ for } |x| \ge 6\sqrt{2}\sigma   \cdot \sqrt{s}\\
	 \frac{1}{\sqrt{2}\sigma}  \cdot e\cdot s \text{ for } |x| \le 6\sqrt{2}\sigma \cdot  \sqrt{s}.
	\end{cases} 
	\end{align*}
	We have $\tau_{s,g}(x) \le \bar \tau_{s,g}(x)$ for all $x \in \R$ and $\int_{-\infty}^{\infty}\bar \tau_{s,g}(x)\ dx = O(s^{3/2})$.
\end{theorem}
We do not know if the upper bound of Theorem \ref{thm:gaussian} is tight, but we know it is close. In particular, if $\mathcal{T}_s$ is restricted to any fixed set of frequencies $\lambda_1 > \ldots > \lambda_s$ it is easy to show that the leverage scores integrate to exactly $s$, and the leverage scores of $\mathcal{T}_s$ can only be larger. So no upper bound can improve on $\int_{-\infty}^{\infty} \bar \tau_{s,g}(x)\ dx = O(s^{3/2})$ by more than a $O(\sqrt{s})$ factor. Closing this  $O(\sqrt{s})$ gap, either by strengthening Theorem \ref{thm:gaussian}, or proving a better lower bound would be very interesting.

%For the Laplace density we prove:
\begin{theorem}[Laplace Density Leverage Score Bound]\label{thm:exp}
	Consider the Laplace density $z(x) = \frac{1}{\sqrt{2}\sigma} e^{-|x|\sqrt{2}/\sigma}$ and let:
	%\vspace{-.5em}
	\begin{align*}
	\bar \tau_{s,z}(x) = \begin{cases} \frac{\sqrt{2}}{\sigma} \cdot e^{-|x|\sqrt{2} /(6\sigma)} \text{ for } |x| \ge 9\sqrt{2} \sigma \cdot {s}\\
	\frac{\sqrt{2}}{ \sigma} \cdot \frac{e^2 \cdot s}{1+ |x|\sqrt{2}/\sigma} \text{ for } |x| \le 9\sqrt{2} \sigma \cdot {s}.
	\end{cases} 
	\end{align*}
	We have $\tau_{s,z}(x) \le \bar \tau_{s,z}(x)$ for all $x \in \R$ and $\int_{-\infty}^{\infty}\bar \tau_{s,z}(x)\ dx = O(s \ln s)$.
\end{theorem}
Again, we do not know if Theorem \ref{thm:exp} is tight, but $\int_{-\infty}^{\infty}\bar \tau_{s,z}(x)\ dx = O(s\ln s)$ cannot be improved below $s$. The best known upper bound for the uniform density also integrates to $O(s\ln s)$ \cite{Erdelyi:2017} and closing the $O(\ln s)$ gap for either distribution is an interesting open question.

\begin{figure}[h]
	\begin{subfigure}{.43\textwidth}
		\centering
		% include first image
		\includegraphics[width=\linewidth]{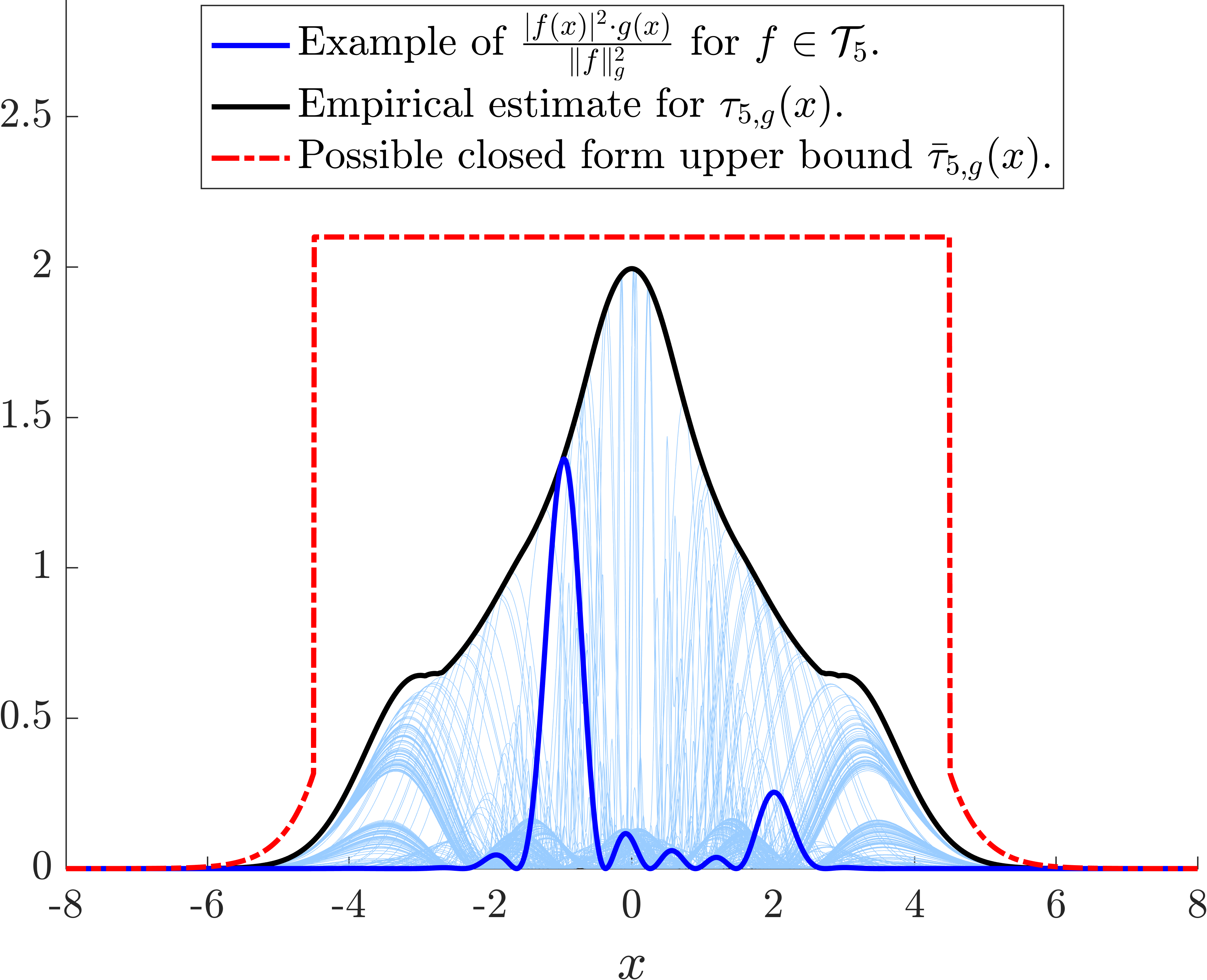}  
		\caption{Leverage scores for Gaussian density.}
		\label{fig:sub-first}
	\end{subfigure}\hfill
	\begin{subfigure}{.43\textwidth}
		\centering
		% include second image
		\includegraphics[width=\linewidth]{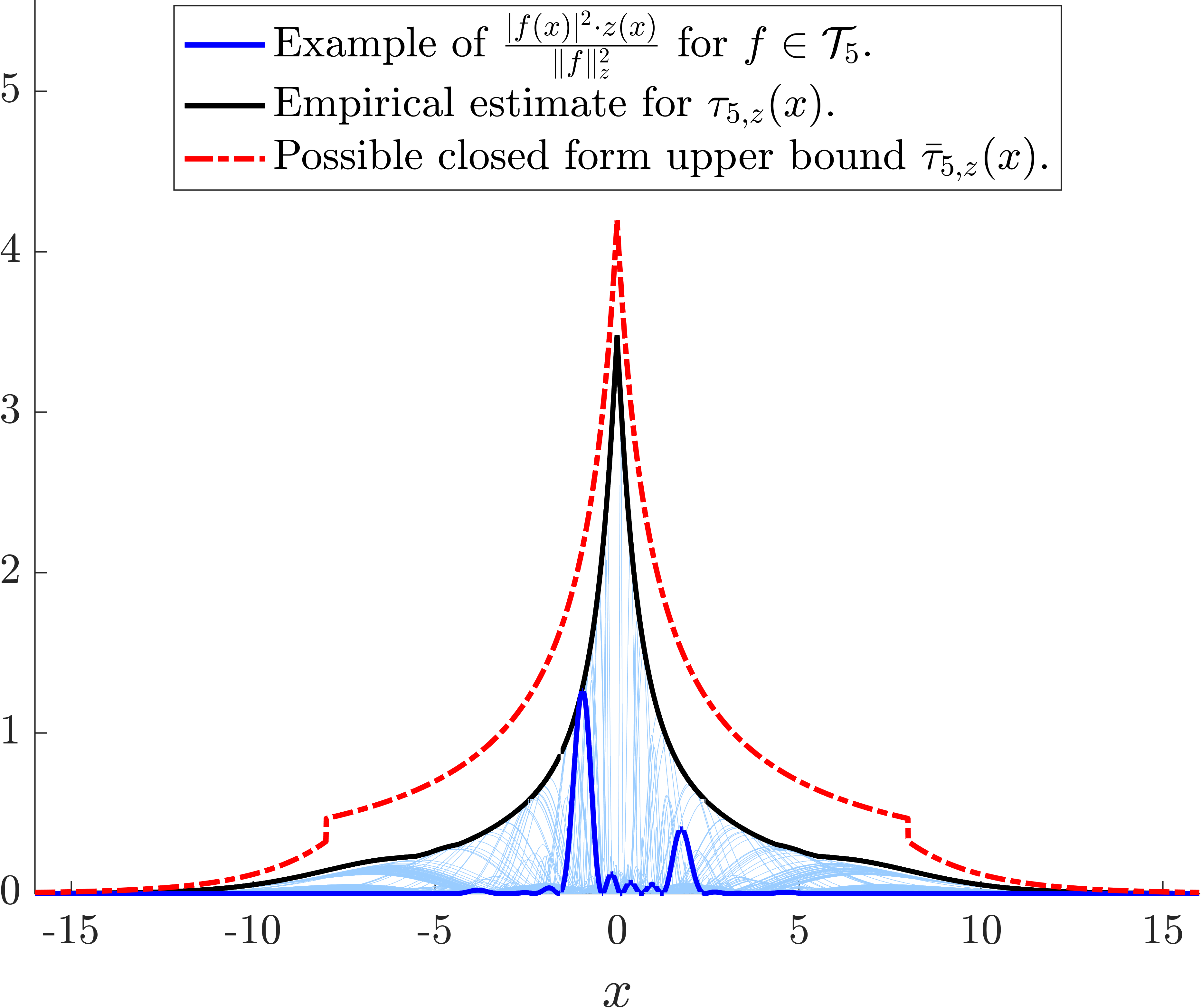} 
		\caption{Leverage scores for Laplace density.}
		\label{fig:sub-second}
	\end{subfigure}
	\caption[]{Empirically computed (see Appendix \ref{app:plot_info} for details) estimates for the Fourier sparse leverage scores, for sparsity $s = 5$. The solid blue lines are normalized magnitudes of $5$-sparse Fourier functions that ``spike'' well above their average. I.e., they plot $|f(x)|^2\cdot p(x)/\|f\|_p^2$ for various $f \in \mathcal{T}_5$. The leverage score function $\tau_{5,p}(x)$ is the supremum of all such functions. The dashed red lines are closed-form upper bounds for the leverage scores: establishing such bounds is our main research objective. For illustration, the ones plotted here are tighter than what we can currently prove, but they have the same functional form as Theorems \ref{thm:gaussian} and \ref{thm:exp} (just with different constants).}
	\label{fig:upper_bound}
	%\vspace{-.5em}
\end{figure}
%\footnotetext{A simple Monte Carlo heuristic algorithm for obtaining these estimates is detailed in Appendix \ref{app:plot_info}.}
Theorems \ref{thm:gaussian} and \ref{thm:exp} are proven in Section \ref{app:bounds} and the upper bounds visualized in Figure \ref{fig:upper_bound}. They build on existing results for when $p$ is the uniform distribution over an interval \cite{BorweinErdelyi:2006,Erdelyi:2017}. This case has been studied since the work of Tur\'{a}n, who proved the first bounds for $\mathcal{T}_s$ and related function classes that are \emph{independent} of the frequencies $\lambda_1, \ldots, \lambda_s$, and only depend on the sparsity $s$ \cite{Turan:1984,Nazarov:1993}. 
%We extend uniform density bounds to the Gaussian and Laplace densities by 
Our bounds take advantage of the exponential form of the Gaussian and Laplace densities $e^{-x^2}$ and $e^{-|x|^2}$. We show how for $f \in \mathcal T_s$ to write the weighted function $f(x) \cdot p(x)$ (whose norm under the uniform density equals $f$'s under $p$) in terms of a Fourier sparse function in an extension of $\mathcal{T}_s$ that allows for complex valued frequencies. Combining leverage score type bounds on this extended class \cite{BorweinErdelyi:2006,Erdelyi:2017} with growth bounds based on Tur\'{a}n's lemma \cite{Turan:1984,BorweinErdelyi:1995} yields our results.

When the minimum gap between frequencies in $f \in \mathcal{T}_s$ is lower bounded, we also give a tight bound (integrating to $O(s)$) based on Ingham's inequality \cite{ingham1936some}, applicable e.g., in our oblivious embedding results when data points are separated by a minimum distance.

%By extending this existing work to two other commonly used densities, Theorems \ref{thm:gaussian} and \ref{thm:exp} expand our understanding of Fourier sparse leverage scores.  Our key technique is to use the exponential form of the Gaussian and Laplace densities to write $f(x) \cdot p(x)$ 
%\Chris{cam maybe you can add some lofty commentary on how the proof is done.}

%Moreover, we expect they will find applications in many of the same domains as uniform density Fourier sparse leverage scores. 

\section{Kernel Approximation}\label{sec:oblivious}

Given data points\footnote{Results are stated for 1D data, where applications of kernel methods include time series analysis and audio processing. As shown in Section \ref{sec:experiments}, our algorithms easily extend to higher dimensions in practice. In theory, however, extended bounds would likely incur an exponential dependence on dimension, as in \cite{AvronKapralovMusco:2017}.} $x_1, \ldots, x_n \in \R$ and positive definite kernel function $k: \R \times \R \rightarrow \R$, let $\bv{K} \in \R^{n\times n}$ be the kernel matrix:
$\bv{K}_{i,j} = k(x_i,x_j)$ for all $i,j$.
$\bv{K}$ is the central object in kernel learning methods like kernel regression, PCA, and SVM. Computationally, these methods typically need to invert or find eigenvectors of $\bv{K}$, operations that require $O(n^3)$ time. When $n$ is large, this cost is intractable, even for data in low-dimensions. In fact, even the $O(n^2)$ space required to store $\bv{K}$ can quickly lead to a computational bottleneck. 
To address this issue, kernel approximation techniques like random Fourier features methods  \cite{RahimiRecht:2007}, Nystr\"{o}m approximation \cite{williams2001using,GittensMahoney:2013}, and TensorSketch \cite{PhamPagh:2013} seek to approximate $\bv{K}$ by a low-rank matrix.

These methods compute an explicit embedding $\bv{g}: \R \rightarrow \C^m$ with $m \ll n$ which can be applied to each data point $x_i$. If $\bv{G}\in \C^{m\times n}$ contains $\bv{g}(x_i)$ as its $i^\text{th}$ column, the goal is for $\bv{\tilde K} = \bv{G}^*\bv{G}$, which has rank $m$, to closely approximate $\bv{K}$.
%\footnote{Note also that, like $\bv{K}$, the approximation $\tilde{\bv{K}} = \bv{G}^*\bv{G}$ is positive semidefinite.}
I.e., for the inner product $\bv{\tilde K}_{i,j}  = \bv g(x_i)^* \bv g(x_j)$ to approximate $\bv{K}_{i,j}$. If the approximation is good, $\bv{\tilde K}$ can be used in place of $\bv{K}$ in downstream applications. It can be stored in $O(nm)$ space, admits $O(nm)$ time matrix-vector multiplication, and can be inverted exactly in $O(nm^2)$ time, all linear in $n$ when $m$ is small.%: linear in $n$ when $s$ is sufficiently small. 

\paragraph{Oblivious Embeddings}
Like sketching methods for matrices (see e.g., \cite{Woodruff:2014}) kernel approximation algorithms fall into two broad classes.
%\vspace{-.5em}
\begin{enumerate}
	\item Data \emph{oblivious} methods choose a random embedding $\bv{g}: \R \rightarrow \C^m$ without looking at the data $x_1, \ldots, x_n$. $\bv g(x_i)$ can then be applied independently, in parallel, to each data point. Oblivious methods include random Fourier features and TensorSketch methods.
	\item Data \emph{adaptive} methods tailor the embedding $\bv{g}: \R \rightarrow \C^m$ to the data $x_1, \ldots, x_n$. For example, Nystr\"{o}m approximation constructs $\bv g$ by projecting (in kernel space) each $x_i$ onto $m$ landmark points selected from the data.
\end{enumerate}
%\vspace{-.5em}
Data oblivious methods offer several advantages over adaptive methods: they are easy to parallelize, naturally apply to streaming or dynamic data, and are typically simpler to implement. However, data adaptive methods currently give more accurate kernel approximations than data oblivious methods \cite{MuscoMusco:2017}. A major open question in the area \cite{AvronKapralovMusco:2017,AhleKapralovKnudsen:2020} is if this gap is necessary.
%Perhaps surprisingly, in the matrix sketching literature data oblivious algorithms typically match the best data adaptive methods in speed and approximation quality. 

Our main contribution in this section is to establish that a significant gap between data oblivious and non-oblivious sketching \emph{does not exist} for the commonly used Gaussian and Cauchy kernels: for one-dimensional data we present a data oblivious method with runtime linear in $n$ that nearly matches the best adaptive methods in speed and approximation quality. 
%Previously, such a strong oblivious result was only known for the polynomial kernel \cite{AhleKapralovKnudsen:2020}.
\subsection{Formal results}

Prior work on randomized algorithms for approximating $\bv{K}$ considers several metrics of accuracy. %For example, we might ask how large an embedding dimension $m$ is necessary for $\|\bv{K} - \bv{\tilde K}\|_{\infty}$, $\|\bv{K} - \bv{\tilde K}\|_{F}$, or $\|\bv{K} - \bv{\tilde K}\|_{2}$ to be smaller than a specified tolerance $\epsilon$. In this work, 
We study the following popular approximation guarantee  \cite{AlaouiMahoney:2015,MuscoMusco:2017,AhleKapralovKnudsen:2020}:
\begin{definition}\label{def:spectral_gaurantee} For parameters $\epsilon,\lambda \ge 0$, we say
	$\tilde{\bv{K}}$ is an $(\epsilon,\lambda)$-spectral approximation for $\bv{K}$ if:
	%\vspace{-.5em}
	\begin{align}
	\label{eq:spectral_gaurantee}
	(1-\epsilon)(\bv{K} + \lambda\bv{I}) \preceq \tilde{\bv{K}} + \lambda\bv{I} \preceq (1+\epsilon)(\bv{K} + \lambda\bv{I}). 
	\end{align}
%Here $\bv{I}$ is an $n\times n$ identity matrix and $\preceq$ denotes the standard Loewner ordering on positive semidefinite matrices: $N \preceq M$ means that $M-N$ is positive semidefinite. $\epsilon < 1$ and $\lambda > 0$ are specified accuracy parameters.
\end{definition}
Definition \ref{def:spectral_gaurantee} can be used to prove guarantees for downstream applications: e.g., that $\tilde{\bv{K}}$ is a good preconditioner for kernel ridge regression with regularization $\lambda$, or that using $\tilde{\bv{K}}$ in place of $\bv{K}$ leads to statistical risk bounds. See \cite{AvronKapralovMusco:2017} for details. 
%The guarantee is also used more generally as a metric for PSD matrix approximation in randomized numerical linear algebra and the spectral graph algorithms literature \cite{CohenMuscoPachocki:2016,CohenMuscoMusco:2017,KapralovLeeMusco:2017}.
With \eqref{eq:spectral_gaurantee} as the approximation goal, the data adaptive Nystr{\"o}m method combined with leverage score sampling \cite{AlaouiMahoney:2015} yields the best known kernel approximations among algorithms with runtime linear in $n$. Specifically, for \emph{any} positive semidefinite kernel function the \emph{RLS} algorithm of \cite{MuscoMusco:2017} produces an embedding satisfying \eqref{eq:spectral_gaurantee} with $\epsilon = 0$ and with $m = O( s_\lambda \log s_\lambda)$ in $\tilde O(n s_\lambda^2)$ time where $s_\lambda$ is the statistical dimension of $\bv{K}$:
%
%The Nystr{\"o}m method selects a set of $m$ landmark datapoints and approximates the kernel by projecting the remaining points onto this set in kernel space. Thus it is not \emph{data oblivious} in that the embedding $\bv{g}(x)$ can not be computed independently for each data point.
 %Even the fastest randomized low-rank approximation algorithms \cite{ClarksonWoodruff:2013} would require at least $O(n^2)$ time.
%
%Accordingly, existing linear (in $n$) time algorithms target a slightly weaker level of compression, seeking to achieve \eqref{eq:spectral_gaurantee} with $m$ which depends on the so-called \emph{statistical dimension} of $\bv{K}$:
\begin{definition}[$\lambda$-Statistical Dimension]\label{def:statistical_dimension}
	The $\lambda$-statistical dimension $s_{\lambda}$ of a positive semidefinite matrix $\bv{K}$ with eigenvalues $\lambda_1 \geq \ldots \geq \lambda_n \geq 0$ is defined as $s_{\lambda} \eqdef \sum_{i=1}^n \frac{\lambda_i}{\lambda_i + \lambda}$.
\end{definition}
The statistical dimension is a natural complexity measure for approximation $\bv{K}$ and the embedding dimension of $O(s_\lambda \log s_\lambda)$ from \cite{MuscoMusco:2017} is near optimal.\footnote{It can be show that embedding dimension $m = \sum_{i=1}^n \mathbbm{1}[\lambda_i \geq \lambda]$ is necessary to achieve \eqref{eq:spectral_gaurantee}. Then observe that $s_{\lambda} \leq \sum_{i=1}^n \mathbbm{1}[\lambda_i \geq \lambda] + \frac{1}{\lambda} \sum_{\lambda_i < \lambda}\lambda_i$. For most kernel matrices encountered in practice, the leading term dominates, so $s_{\lambda}$ is roughly on the order of the optimal $m$.}
% \ref{def:spectral_gaurantee}. 
Our main result gives a similar guarantee for two popular kernel functions: the Gaussian kernel $k(x_i, x_j) = e^{-(x_i-x_j)^2/(2\sigma^2)}$ with width $\sigma$ and the Cauchy kernel  $k(x_i, x_j) = \frac{1}{1+(x_i-x_j)^2/\sigma^2}$ with width $\sigma$. The Cauchy kernel is also called the ``rational quadratic kernel'', e.g., in \texttt{sklearn} \cite{sklearn}.

%for a restricted class of kernel functions. In particular, it applies to kernels of the form:
%\begin{description}
%	\item[Gaussian Kernel:] $\bv{K}_{i,j} = k(x_i, x_j) = e^{-(x_i-x_j)^2/(2\sigma^2)}$, where $\sigma$ is a width parameter. This ubiquitous kernel is also known as the radial basis function or squared-exponential kernel.
%	\item[Cauchy Kernel:] $\bv{K}_{i,j} = k(x_i, x_j) = \frac{1}{1+(x_i-x_j)^2/\sigma^2}$, $\sigma$ is a width parameter \cite{souza2010kernel}. This kernel is also called the ``rational quadratic kernel'', e.g. in \texttt{sklearn} \cite{sklearn}.
%\end{description}

\begin{theorem}\label{thm:high_kernel}
	Consider any set of data points $x_1,\ldots, x_n \in \R$ with associated kernel matrix $\bv{K} \in \R^{n \times n}$ which is either Gaussian or Cauchy with arbitrary width parameter $\sigma$. There exists a randomized \emph{oblivious} kernel embedding $\bv{g}: \R \rightarrow \C^m$ such that, if $\bv{G} = [\bv{g}(x_1)\ldots, \bv{g}(x_n)]$,with high probability $\bv{\tilde K} = \bv{G}^*\bv{G}$ satisfies \eqref{eq:spectral_gaurantee} with embedding dimension $m = O(\frac{s_\lambda  }{\epsilon^2})$. $\bv{G}$ can be constructed in $\tilde{O}(n\cdot s_\lambda^{3.5}/\epsilon^4)$ time for Gaussian kernels and $\tilde{O}(n\cdot s_\lambda^3/\epsilon^4)$ time for Cauchy kernels.
\end{theorem}
Theorem \ref{thm:high_kernel} is a simplified statement of Corollary \ref{cor:full_kernel_result}, proven in Appendix \ref{app:kernels}. There we explicitly state the form of $\bv{g}$, which as discussed in Section \ref{sec:rffApproach} below, is composed of a random Fourier features sampling step followed by  a standard random projection. For one dimensional data, our method matches the best Nystr\"{o}m method in terms of embedding dimension up to a $1/\epsilon^2$ factor, and in terms of running time up to an $s_\lambda^{1.5}$ factor. 
It thus provides one of the first nearly optimal oblivious embedding methods for a special class of kernels.
The only similar known result applies to polynomial kernels of degree $q$, which can be approximated using the TensorSketch technique \cite{PhamPagh:2013,meister2019tight,avron2014subspace}. A long line of work on this method culminated in a recent breakthrough achieving embedding dimension $m = O \left (q^4 s_\lambda/\epsilon^2 \right )$, with embedding time $O(nm)$ \cite{AhleKapralovKnudsen:2020}. That method can be extended e.g., to the Gaussian kernel, via polynomial approximation of the Gaussian, but one must assume that the data lies within a ball of radius $R$ and the embedding dimension suffers polynomially in $R$.

\subsection{Our approach}\label{sec:rffApproach}
Theorem \ref{thm:high_kernel} is based on a modification of the popular \emph{random Fourier features (RFF)} method from \cite{RahimiRecht:2007}, and like the original method can be implemented in a few lines of code (see Section \ref{sec:experiments}). As for all RFF methods, it is based on the following standard result for shift-invariant kernel functions:

\begin{fact}[Bochner's Theorem]\label{fact:bochners}
For any shift invariant kernel $k(x,y) = k(x-y)$ where $k: \R \rightarrow \R$ is a positive definite function with $k(0) = 1$, the inverse Fourier transform given by $
p_k(\eta) = \int_{t \in \R} e^{2 \pi i \eta t} k(t) d t$
is a probability density function. I.e. $p_k(\eta) \geq 0$ for all $\eta \in \R$ and $\int_{\eta \in \R}p_k(\eta) = 1$. 
\end{fact}

As observed by Rahimi and Recht in \cite{RahimiRecht:2007}, Fact \ref{fact:bochners} inspires a natural class of linear time randomized algorithms for approximating $\bv{K}$. We begin by observing that $\bv{K}$ can be written as
$\bv{K} = \bs{\Phi}^*\bs{\Phi}$,
where $^*$ denotes the Hermitian adjoint and $\bs{\Phi}: \C^n \rightarrow L_2$ is the linear operator with $[\bs{\Phi}\bv w](\eta) = \sqrt{p_k(\eta)} \cdot \sum_{j=1}^n \bv w_j e^{-2\pi i \eta x_j}$ for $\bv w\in \C^n, \eta \in \R$.

It is helpful to think of $\bs{\Phi}$ as an infinitely tall matrix with $n$ columns and rows indexed by real valued ``frequencies'' $\eta \in \R$. RFF methods approximate $\bv{K}$ by subsampling and reweighting rows (i.e. frequencies) of $\bs{\Phi}$ independently at random to form a matrix $\bv{G} \in \C^{m \times n}$.
$\bv{K}$ is approximated by $\tilde{\bv{K}} = \bv{G}^*\bv{G}$.
In general, row subsampling is performed using a non-uniform importance sampling distribution. The following general framework for unbiased sampling is described in \cite{AvronKapralovMusco:2017}: 
\begin{definition}[Modified RFF Embedding]\label{def:rff}
Consider a shift invariant kernel $k: \R \rightarrow \R$ with inverse Fourier transform $p_k$. For a chosen PDF $q$ whose support includes that of $p_k$, the Modified RFF embedding $\bv{g}(x): \R\rightarrow \C^m$ is obtained by sampling $\eta_1,\ldots, \eta_m$ independently from $q$ and defining:
\begin{align*}
\bv g (x) = \frac{1}{\sqrt{m}} \left [\sqrt{\frac{p_k(\eta_1)}{q(\eta_1)}} e^{-2 \pi i \eta_1 x},\ldots, \sqrt{\frac{p_k(\eta_m)}{q(\eta_m)}} e^{-2 \pi i \eta_m x} \right ]^*.
\end{align*}
\end{definition}
It is easy to observe that for the modified RFF method $\E[ \bv g (x)^* \bv g (y)] = k(x,y)$ and thus $\E[\bv{G}^*\bv{G}] = \bv{K}$.
%
 %following claim is immediate: 
%%\begin{claim}
%For any $x,y \in \R$, $\E[ \bv g (x)^* \bv g (y)] = k(x,y)$. For a data set $x_1, \ldots, x_n$ if $\bv{G} \in \R^{m\times n}$ contains %$\bv g(x_i)$ as its $i^\text{th}$ row, then $\E[\bv{G}^*\bv{G}] = \bv{K}$. 
%\end{claim}
So, the feature transformation $\bv{g}(\cdot)$ gives an unbiased approximation to $\bv{K}$ for any sampling distribution $q$ used to select frequencies. However, a good choice for $q$ is critical in ensuring that $\bv{G}^*\bv{G}$ concentrates closely around its expectation with few samples. 
The original Fourier features method makes the natural choices $q = p_k$, which leads to approximation bounds in terms of $\norm{\bv{K} - \bv{\tilde K}}_\infty$ \cite{RahimiRecht:2007}. 
% To satisfy spectral approximation guarantees like Definition \ref{def:spectral_gaurantee}, it is necessary to sample frequencies with a more carefully chosen distribution. 
\cite{AvronKapralovMusco:2017} provides a stronger result by showing that sampling proportional to the so-called \emph{kernel ridge leverage function} is sufficient for an approximation satisfying Definition \ref{def:spectral_gaurantee} with $m = O(s_\lambda\log s_\lambda/\epsilon^2)$ samples. That function is defined as follows:
\begin{definition}[Kernel Ridge  Leverage Function]\label{def:kl}
	Consider a positive definite, shift invariant kernel $k: \R \rightarrow \R$, a set of points $x_1,\ldots, x_n \in \R$ with associated kernel matrix $\bv{K} \in \R^{n \times n}$, and a ridge parameter $\lambda \ge 0$. The $\lambda$-ridge leverage score of a frequency $\eta \in \R$ is given by:
	\begin{align*}
	\tau_{\lambda,\bv{K}}(\eta) =\sup_{\bv w \in \C^n, \bv{w} \neq 0} \frac{|[\bs{\Phi}\bv w](\eta)|^2}{\|\bs{\Phi}\bv w\|_{2}^2 + \lambda \| \bv w\|_2^2}. %=  \bs \phi_\eta^* (\bv{K}+\lambda \bv I)^{-1} \bs \phi_\eta.
	\end{align*}
	%where $\bs \phi_\eta \in \C^{n}$ has $\bs \phi_\eta (j) = e^{-2 \pi i \eta x_j} \cdot \sqrt{p_k(\eta) }$.
	%Here $\|f\|_{L_2}^2$ denotes the standard $L_2$ norm $\int_{\eta\in \R}|f(\eta)|^2d\eta$. 
\end{definition}
Definition \ref{def:kl} is closely related to the standard leverage score of \eqref{eq:gen_lev_score}. It measures the worse case concentration of a function $\bs{\Phi} \bv{w}$ in the span of our kernelized data points at a frequency $\eta$. Since $\norm{\bs{\Phi} \bv{w}}_2^2 = \bv{w}^* \bs{\Phi}^* \bs{\Phi} \bv{w} = \bv{w}^* \bv{K} \bv{w}$, leverage score sampling from this class directly aims to preserve $\bv{w}^* \bv{K} \bv{w}$ for worse case $\bv{w}$ and thus achieve the spectral guarantee of Definition \ref{def:spectral_gaurantee}. 
Due to the additive error $\lambda \bv{I}$ in this guarantee, it suffices to bound the concentration with regularization term $\lambda \norm{\bv{w}}_2^2$ in the denominator.% when $\bv{w}$ has sufficiently small norm, 

Of course, the above ridge leverage function is data dependent. To obtain an oblivious sketching method \cite{AvronKapralovMusco:2017} suggests proving closed form upper bounds on the function, which can be used in its place for sampling. They prove results for the Gaussian kernel, but the bounds require that data lies within a ball of radius $R$, so do not achieve an embedding dimension linear in $s_\lambda$ for any dataset. 
%$m = \tilde O \left (\frac{R \sqrt{\log(n/\lambda)} + \log(n/\lambda)^2}{\epsilon^2} \right )$.
We improve this result by showing that it is possible to bound the \emph{kernel ridge leverage function} in terms of the \emph{Fourier sparse leverage function} for the density $p_k$ given by  the kernel Fourier transform:
\begin{theorem}\label{thm:sparseReduction} Consider a positive definite, shift invariant kernel $k: \R \rightarrow \R$, any points $x_1,\ldots, x_n \in \R$ and the associated kernel matrix $\bv{K}$, with statistical dimension $s_\lambda$.  Let $s = 6\lceil s_\lambda\rceil +1$. Then:%For all $\eta \in \R$:
	%\vspace{-.5em}
	\begin{align*}
	\forall \eta \in \R, \quad \tau_{\lambda,\bv{K}}(\eta)  \le  (2 +6 s_\lambda) \cdot  \tau_{ s, p_k}(\eta).
	\end{align*}
\end{theorem}
We prove Theorem \ref{thm:sparseReduction} in Appendix \ref{app:kernels}. We show that $\bv{\Phi} \bv{w}$ can be approximated by an $s = 6\lceil s_\lambda\rceil +1$ Fourier sparse function, so bounding how much it can spike 
(i.e., which bounds the ridge leverage score of Def. \ref{def:kl}) reduces to
bounding the Fourier sparse leverage scores. %This general technique of approximating kernel space functions with Fourier sparse functions was first introduced in \cite{AvronKapralovMusco:2019} in the context of active regression, and we believe it will be a powerful tool in understanding kernel approximation going forward.
With Theorem \ref{thm:sparseReduction} in place, we immediately obtain a modified random Fourier features method for any kernel $k$, given an upper bound the Fourier sparse leverage scores of $p_k$. The Fourier transform of the Gaussian kernel is Gaussian, so Theorem \ref{thm:gaussian} provides the required bound. The Fourier transform of the Cauchy kernel is the Laplace distribution, so Theorem \ref{thm:exp} provides the required bound. 

\paragraph{Final Embeddings via Random Projection.} In both cases, Theorem \ref{thm:sparseReduction} combined with our leverage scores bounds does not achieve a tight result alone, yielding embeddings with $m = O(\poly(s_\lambda))$. To achieve the linear dependence on $s_\lambda$ in Theorem \ref{thm:high_kernel}, we show that it suffices to post-process the modified RFF embedding $\bv{g}$ with a standard oblivious random projection method  \cite{cohen2016optimal}. Proofs are detailed in Appendix \ref{sec:rp}, with a complete statement of the random features + random projection embedding algorithm given in Corollary \ref{cor:full_kernel_result}.

It is worth noting that, given any approximation 
% that when our kernel is linear -- i.e., $\bv{K} = \bv{X} \bv{X}^T$ for given matrix $\bv{X} \in \R^{n \times d}$, the best known oblivious sketching bounds achieve  Definition \ref{def:spectral_gaurantee}  with dimension $O \left (\frac{s_\lambda}{\epsilon^2}\right  )$ \cite{cohen2016optimal}, making this a target for our methods. This also means that, 
$\bv{\tilde K} = \bv{G}^* \bv{G}$ satisfying Definition \ref{def:spectral_gaurantee}, we can always apply oblivious random projection to $\bv{G}$ to further reduce the embedding to the target dimension $O \left (\frac{s_\lambda}{\epsilon^2}\right  )$, while maintaining the guarantee of Definition \ref{def:spectral_gaurantee} up to constants on the error parameters.\footnote{We also need the slightly stronger condition that $\bv{\tilde K}$'s statistical dimension is close to that of $\bv{K}$. This condition holds for essentially all known sketching methods.} Thus, the main contribution of Theorem \ref{thm:high_kernel} is achieving a lower initial dimension of $\bv{G}$ via this sampling step, which directly translates into a faster runtime to produce the final embedding. Our initial embedding dimension, and hence runtime depends polynomially on $s_\lambda$ and $\epsilon$. 
Existing work \cite{AvronKapralovMusco:2017,AhleKapralovKnudsen:2020} makes an additional assumption that the data points fall in some radius $R$, and their initial embedding dimension and hence runtime suffers polynomially in this parameter. Related results make no such assumption, but depend linearly on $1/\lambda$ \cite{AvronKapralovMusco:2017,li2018towards}, a quantity which can be much larger than $s_\lambda$ in the typical case when $\bv{K}$ has decaying eigenvalues.

%is to shows that for any shift-invariant kernel, leverage-score-based random Fourier features sampling can in theory achieve Definition \ref{def:spectral_gaurantee} with $m = O\left(\frac{s_\lambda \log s_\lambda}{\epsilon^2} \right )$, however this requires a tight upper bound on the Fourier leverage scores. They give an upper bound specifically fo

\section{Active Learning}\label{sec:active}

We next consider a general active learning problem that encompasses classic problems in both signal processing and machine learning, including e.g., bandlimited function approximation and active Gaussian process regression. Informally, given the ability to make noisy measurements of some function $f$, the goal is to fit a function $\tilde f$ with small deviation from $f$ under some data density $p$, under the assumption that  $f$ has Fourier transform constrained according to some frequency density $q$. For example, when $q$ is uniform on a bounded interval, $f$ is bandlimited. When $q$ Gaussian, $f$ obeys a `soft bandlimit' tending towards using lower frequencies with higher density  under $q$.

Throughout this section we use the following notation: for any density $p$ over $\R$ let $L_2(p)$ denote the space of square integrable functions with respect to $p$, i.e., $f$ with $\norm{f}_{p}^2 = \int_{x \in \R} |f(x)|^2 p(x) dx < \infty$. For $f,g \in L_2(p)$ we denote the inner product $\langle f,g\rangle_p \eqdef \int_{x \in \R} f(x)^* g(x) p(x) dx$, where $f(x)^*$ is the conjugate transpose of $f(x)$.
We define the weighted Fourier transform with respect to data and frequency densities $p$ and $q$ as:
\begin{definition}[Weighted Fourier Transform]\label{def:weigthedFT}
Let $p,q$ be probability densities on $\R$. 
% representing a data distribution and a frequency distribution respectively. 
Define the weighted Fourier transform $\Fmu: L_2(p) \rightarrow L_2(q)$ by:\footnote{As in \cite{AvronKapralovMusco:2019}, we can generalize the weighted Fourier transform to be weighted by any two measures over $\R$. This allows, for example, the use of discrete measures. We focus on the case when the measures correspond to density functions $p,q$ for simplicity of exposition.}% of $f \in L_2(p)$ with respect to $p$ as:
	\begin{align}
	\label{eq:p_transform}
	\left[\Fp \,f\right](\eta) \eqdef \int_{\R} f(x) e^{-2\pi i \eta x}\, p(x) d x.
	\end{align}
	The adjoint $\Fmu^*$ such that $\langle g, \Fmu f \rangle_q = \langle \Fmu^* g, f \rangle_p$ is the inverse Fourier transform operator:
	%Similarly, define the inverse Fourier transform of $g \in L_2(q)$ with respect to $q$ as:
	\begin{align}
	\label{eq:mu_transform}
	\left[\Fmu^* \,g\right](x) \eqdef \int_{\R} g(\eta) e^{2\pi i \eta x}\, q(\eta) d \eta.
	\end{align}
	%Note that since we assume $f \in L_2(p)$ and $g \in L_2(q)$,  $[\Fp f](\eta)$ and $[\Fmu g](x)$ are bounded and thus in $L_2(q)$ and $L_2(p)$ as well.
\end{definition}
With Definition \ref{def:weigthedFT} in place we can formally define our main active regression problem of interest:
\begin{problem}[Active Function Fitting]
	\label{prob:unformal_interp}
	Let $p,q$ be probability densities on $\R$ representing data and frequency densities respectively.
	%Given a known probability measure $\mu$ on $\R$, 
	%Define the inverse Fourier transform of $g \in L_2(q)$ with respect to $q$ as:
	%\begin{align}
	%\label{eq:mu_transform}
	%\left[\Fmu^* \,g\right](x) \eqdef \int_{\R} g(\eta) e^{2\pi i \eta x}\, q(\eta) d \eta.
	%\end{align}
	Suppose a time domain function $y \in L_2(p)$ can be written as $y = \Fmu^* \, g$ for some frequency domain function $g \in L_2(q)$ and, for any $x \in \supp(p)$, we can query $y(x) + n(x)$ for some fixed noise function $n\in L_2(p)$. Then, for error parameter $\lambda \ge 0$, our goal is to recover, using as few queries as possible, an approximation $\tilde y \in L_2(p)$ satisfying:
	\begin{align}
	\label{eq:main_guarantee}
	\| y - \tilde{y}\|_p^2 \leq C\|n\|_p^2 + \lambda \|g\|_{q}^2,
	\end{align}
	%where $\|x\|_{\mu}^2 \eqdef \int_{\R} |x(\xi)|^2\, d\mu(\xi)$ is the energy of the function $x$ with respect to $\mu$, while  $\|z\|_T^2\eqdef \frac{1}{T}\int_0^T |z(t)|^2 dt$, so that $\| y - \tilde{y}\|_T^2$ is our mean squared error and $\|n\|_T^2$ is the mean squared noise level. 
	where $C \ge 1$ is a fixed positive constant.
\end{problem}
The first error term of \eqref{eq:main_guarantee} depends on $\norm{n}_p^2$, which in general is necessary, since the noise $n$ is adversarial. Information theoretically, we might hope to achieve $C \rightarrow 1$ as we take more and more samples, but we focus on achieving within a small constant factor of this ideal bound. The second term $\lambda \|g\|_{q}^2$ is also necessary in general: it is higher when $y$'s Fourier energy under the frequency density $q$ is larger, making $y$ harder  to learn. By decreasing $\lambda$ we obtain a better approximation, at the cost of higher sample complexity.

As discussed, 
Problem \ref{prob:unformal_interp} captures a wide range of classical function fitting problems. See \cite{AvronKapralovMusco:2019} for details and an exposition of prior work.
\begin{itemize}
\item When $q$ is uniform on an interval $[-F,F]$, $f= \Fmu^* g$ is bandlimited with bandlimit $F$. Thus Problem \ref{prob:unformal_interp} corresponds bandlimited approximation, which lies at the core of modern signal processing and Nyquist sampling theory \cite{Whittaker:1915,Nyquist:1928,Kotelnikov:1933,Shannon:1949}. Classically, this problem is considered over an infinite time horizon with access  to infinite samples at a certain rate. Significant work also studies the problem in the finite sample regime, when $p$ is uniform over an interval \cite{LandauPollak:1961,LandauPollak:1962,SlepianPollak:1961,XiaoRokhlinYarvin:2001,OsipovRokhlin:2014}.
\item When $q$ is uniform over a union of intervals, $f$ is composed of frequencies restricted to these intervals and so Problem \ref{prob:unformal_interp} corresponds to multiband function approximation over data density $p$. This problem is also central in signal processing and studied in both the infinite and finite sample regimes \cite{Landau:1967a,LandauWidom:1980,FengBresler:1996,MishaliEldar:2009, LakeyHogan:2012}. 
\item When $q$ is a general density, Problem \ref{prob:unformal_interp} is closely related to Gaussian process regression (also kriging/kernel ridge regression) \cite{HandcockStein:1993, RasmussenWilliams06,stein2012interpolation} over data distribution $p$ with covariance kernel $k_q$ given by the Fourier transform of $q$. $q$ corresponds to the expected power spectral density of a Gaussian process drawn with this covariance kernel. For example, if $q$ is Gaussian, $k_q$ is the Gaussian kernel. If  $q$ is Cauchy, $k_q$ is the exponential kernel. If $q$ is a mixture of Gaussians, so is $k_q$, a so-called spectral mixture kernel \cite{wilson2013gaussian}.
\end{itemize}

Related to the last example above, it is not hard to show that Problem \ref{prob:unformal_interp} can be solved by an infinite dimensional kernel ridge regression problem, where the kernel space corresponds to the class of functions $\Fmu^* w$ for $w \in L_2(q)$ and the target is the noisy function $y+n$.% using the kernel $k$ whose Fourier transform is given by $q$. 

\begin{claim}[Claim 4 of \cite{AvronKapralovMusco:2019}]
	\label{claim:regression_reduction}
	Consider the setting of Problem \ref{prob:unformal_interp}. Let 
	%Let $y = \Fmu^* x$, $n \in L_2(p)$ be an arbitrary noise function, and let 
	$\tilde{g} \in L_2(q)$ satisfy:
	\begin{align}
	\label{eq:approx_regress_1}
		\|\Fmu^* \tilde{g} - (y+n)\|_p^2 + \lambda \|\tilde{g}\|_q^2 \leq C\cdot\min_{w \in L_2(q)}\left[  \|\Fmu^* w - (y+n)\|_p^2 + \lambda \|w\|_q^2\right]
	\end{align}
	for some $C \ge 1$.
	Then
	\begin{align*}
		\|y - \Fmu^* \tilde{g}\|_p^2 \leq 2C\lambda  \|g\|_q^2 + 2(C+1)\|n\|_p^2.
	\end{align*}
	That is, $\tilde y = \Fmu^* \tilde{g}$ solves Problem \ref{prob:unformal_interp} with error parameters $\lambda' = 2C\lambda$ and $C' = 2(C+1)$.
\end{claim}
We note that Claim 4 of \cite{AvronKapralovMusco:2019} is stated in the case when $p$ is the uniform density  on an interval, however the proof is via a simple application of triangle inequality and holds for  any density $p$. Throughout this section, we will employ several results from \cite{AvronKapralovMusco:2019} that are stated in the case when $p$ is uniform on an interval but generalize to any density $p$.

\subsection{Active function fitting via kernel leverage score sampling}

Of course, the optimization problem of Claim \ref{claim:regression_reduction} cannot be solved exactly, as it requires full access to $y+n$ on $\supp(p)$. The key idea is to solve the problem approximately by sampling $x \in \supp(p)$ according to their ridge leverage scores and querying $y$ at the sampled points.
\begin{definition}[Kernel operator ridge leverage function]\label{def:ridgeScores}
	For probability densities $p,q$ on $\R$ and ridge parameter  $\lambda \ge 0$, define the $\lambda$-ridge leverage function for $x \in \R$ as:
	\begin{align}
	\label{eq:leverage_def}
		\tmu(x) =  \sup_{\{w \in L_2(q)\, \norm{w}_q > 0\}} \frac{p(x) \cdot \left|[\Fmu^* w](x) \right|^2 }{\norm{\Fmu^*w}_p^2 + \lambda \norm{w}_q^2}.
	\end{align}
\end{definition}
The above ridge leverage scores are closely related to the standard leverage scores of \eqref{eq:gen_lev_score}, for the class of functions $\{f: f = \Fmu^* w \text{ for } w \in L_2(q)\}$, which we fit in Problem \ref{prob:unformal_interp}. 
Intuitively, we hope to sample our function in locations where this class can place significant mass (weighted by the data density $p$), so that we can accurately solve the regression problem of Claim \ref{claim:regression_reduction}.

Typically however, the standard leverage scores of this function class are unbounded. For example, when $q$ is uniform on an interval, this is the space of all bandlimited functions, which may be arbitrarily spiky. The ridge scores account for this by including a regularization term involving $\norm{w}_q^2$ which controls the energy of the function and in turn, how spiky it can be. As Problem \ref{prob:unformal_interp} allows error in terms of $\norm{w}_q^2$, sampling by these scores still suffices for an accurate solution. We note that if $\norm{w}_q^2$ were allowed to be unbounded, i.e., if we set $\lambda = 0$, it would be impossible to solve Problem \ref{prob:unformal_interp} for most common frequencies densities $q$ with a finite number of samples.% without additional assumptions.

Definition \ref{def:ridgeScores} is closely related to Definition \ref{def:kl}, the kernel leverage scores used in our modified RFF algorithm. We note two differences: 1) the leverage function is defined over data points $x \in \R$ rather than frequencies $\eta \in \R$ and 2) both the data and frequency domains are continuous, while in Def. \ref{def:kl} the data domain is discrete set of $n$ points. Notationally, a minor difference is that in Def. \ref{def:kl} the density $p$ is `baked into' the Fourier operator $\bs{\Phi}$ through a weighting of  $\sqrt{p(\eta)}$  on each row.

The ridge leverage function of Definition \ref{def:ridgeScores} has received recent attention in the machine learning literature \cite{PauwelsBachVert:2018,LasserrePauwels:2019,fanuel2019nystr}. $\Fmu w$ lies in the kernel Hilbert space corresponding to the kernel $k_q$ whose Fourier transform is $q$. $\norm{w}_q^2$ is the norm of the function in the kernel Hilbert space. \cite{PauwelsBachVert:2018} focuses on bounding the leverage function in the limit as $\lambda \rightarrow 0$. In this limiting case, the function can be shown to converge to a simple transformation of the data density $p$. It is due to this kernel interpretation, which we will see more clearly  in our following bounds, that we use the term \emph{kernel operator ridge leverage function.}

As in the discrete kernel matrix case, the ridge leverage scores integrate to the statistical dimension of the associated kernel operator, which in this case is infinite dimensional.

\begin{definition}[Kernel operator statistical dimension]\label{def:statDim}
	For probability densities $p,q$ define the kernel operator $\Kmu: L_2(p)\rightarrow L_2(p)$ as $\Kmu = \Fmu^* \Fp$. 
	The $\lambda$-statistical dimension of $\Kmu$ is defined as:
	\begin{align}
	\label{eq:stat_dim_def}
	\smu \eqdef \tr(\Kmu (\Kmu + \lambda \mathcal{I})^{-1}) = \sum_{i=1}^\infty \frac{\lambda_i\left(\Kmu\right)}{\lambda_i\left(\Kmu\right) + \lambda},
	\end{align}	
	where $\mathcal{I}$ is the identity operator on $L_2(p)$ and $\lambda_i(\Kmu)$ is the $i^{th}$ largest eigenvalue of $\Kmu$.
	By Theorem 5 of \cite{AvronKapralovMusco:2019}, $\int_{x \in \R} \tmu(x) d x = \smu$.
\end{definition}

The work of \cite{AvronKapralovMusco:2019} shows that the kernel operator statistical dimension $\smu$ essentially characterizes the sample complexity of Problem \ref{prob:unformal_interp}. Under very mild assumptions (see Section 6 of \cite{AvronKapralovMusco:2019} for details), they show that any algorithm solving Problem \ref{prob:unformal_interp} must use $\Omega(\smu)$ samples. Conversely, 
by sampling data points according to the kernel operator ridge leverage score function (Def.  \ref{def:ridgeScores}), or a tight upper bound on this function, one can achieve a sample complexity nearly matching this lower bound. Additionally, the algorithm that achieves this complexity is simple and efficient, based on standard kernel ridge regression. Details are discussed in  Appendix \ref{app:active}.

	\vspace{-.5em}
\subsection{Kernel leverage score bound via Fourier sparse leverage scores}\label{sec:klsbv}

In sum, to solve Problem \ref{prob:unformal_interp} with near optimal sample complexity, it suffices to find a function $\ttmu$ that tightly upper bounds the true kernel operator leverage function $\tmu$ of Definition \ref{def:ridgeScores}. 
%
%In sum, Theorem \ref{thm:baseSampling}, combined with Claim \ref{clm:krr} and Claim \ref{claim:regression_reduction} lets us solve the active function fitting problem (Problem \ref{prob:unformal_interp}) with near optimal sample complexity, if we can find a sampling distribution $\ttmu$ that tightly upper bounds the true kernel operator leverage function $\tmu$ of Definition \ref{def:ridgeScores}. 
We do this using a similar approach that of Section \ref{sec:rffApproach}: we show how to well
approximate any function $\Fmu^* w$ via a Fourier sparse function in $\mathcal T_s$, with sparsity $s$ linear in the statistical dimension $\smu$. Using this approximation, we give a blackbox bound  on $\tmu$ in terms of the Fourier sparse leverage scores under the data distribution $p$, giving the following analog to Theorem \ref{thm:sparseReduction}:

\begin{theorem}[Kernel operator leverage function bound]\label{thm:css2} Let $s = \lceil 36 \cdot \smu \rceil +1$. 
For all $x \in \R$:
	\begin{align*}
	\tmu(x) \le (2+8 \smu) \cdot \tau_{ s, p}(x).
	\end{align*}
\end{theorem}

With Theorem \ref{thm:css2} in hand, we obtain the following result for solving Problem \ref{prob:unformal_interp}, which is stated in more detail in Appendix \ref{app:cor}.

%we can apply the leverage score sampling result of Theorem \ref{thm:baseSampling}, and Claim \ref{clm:krr} to give a bound for solving Problem \ref{prob:unformal_interp} with sample complexity depending polynomially on the statistical dimension $\smu$.
\begin{corollary}[Active Function Fitting -- Gaussian or Exponential Density]\label{cor:gaussActive}
Consider the active regression set up of Problem \ref{prob:unformal_interp}.
Let $p$ be the Gaussian density $p(x) = \frac{1}{\sigma  \sqrt{2\pi}} e^{-x^2/(2\sigma^2)}$.

 For any frequency density $q$ and $\lambda > 0$, let $\smu$ be the $\lambda$-statistical dimension of $\Kmu$. Let $s = \lceil 36 \smu \rceil +1$ and let $\bar \tau_{s,p}(x)$ be the leverage score bound of Theorem \ref{thm:gaussian}.
	Let 
	$m = c \cdot \smu^{5/2} \cdot \left(\log \smu  + 1/\delta\right)$  for a sufficiently large constant $c$. Let $x_1,\ldots,x_m$ be time points sampled independently according to the density proportional to $\bar \tau_{s,p}(x)$. There is a polynomial time solvable kernel ridge regression problem on $x_1,\ldots,x_m$ whose solution $\tilde y$ satisfies with probability $\ge 1-\delta$:%be computed from these points using kernel ridge regression according to the procedure of Theorem \ref{thm:baseSampling} and Claim \ref{clm:krr}.
	 %Then with probability $\ge 1- \delta$:
	\begin{align}
	\label{12again2}
		\|y - \tilde y\|_p^2 = 8 \norm{n}_p^2 + 6 \lambda \|{g}\|_q^2.
	\end{align}
	
	An identical bound holds when $p$ is the Laplacian density $p(x) = \frac{1}{\sqrt{2}\sigma} e^{-|x| \sqrt{2}/\sigma}$, $\bar \tau_{s,p}(x)$ is the leverage score bound of Theorem \ref{thm:exp}, and $m = c \cdot \smu^{2} \cdot \left(\log \smu  + 1/\delta\right)$.
\end{corollary}

\paragraph{Universal Sampling.} We remark that the sampling distribution of Corollary \ref{cor:gaussActive} is \emph{independent of the frequency density $q$}. That is, we can fit a wide range of Fourier constrained functions (bandlimited, multiband, Gaussian process with any underlying kernel, etc.) with a single \emph{universal} sampling scheme. This is surprising and reflects the universality of Fourier sparse functions in approximating all of these function classes.% through the frequency subset election result of Theorem \ref{thm:css}.

\vspace{-.5em}
\paragraph{Achieving Optimal Sample Complexity.} The sample complexity bounds of Corollary \ref{cor:gaussActive} are polynomial in $\smu$ rather than linear, as is essentially optimal. We note that a near linear bound can be obtained by subsampling the kernel ridge regression problem on $x_1,\ldots,x_m$ using standard finite matrix leverage score sampling  techniques, discussed in more detail in Appendix \ref{app:cor}. %could be obtained by simply applying a second sampling step to the final kernel ridge regression problem of Claim \ref{clm:krr}, using the ridge leverage scores of the finite kernel matrix $\bv{K}$ \cite{Sarlos:2006,DrineasMahoneyMuthukrishnan:2006}. This is analogous to the final finite-dimensional random projection discussed in Section \ref{sec:rffApproach}. A full proof requires an extension of Theorem \ref{thm:baseSampling}, which applies to an approximate solution of the finite ridge regression problem. This extension was shown in \cite{AvronKapralovMusco:2019}.

It may be possible to avoid this second round of sampling by improving our bounds on the kernel leverage scores (Def. \ref{def:ridgeScores}). In \cite{AvronKapralovMusco:2019} sample complexity $O(\smu \log \smu)$ is shown when $p$ is the uniform density over an interval. This proof starts from a bound essentially equivalent to Theorem \ref{thm:css2}. It then tightens the bound via a shifting argument that bounds the kernel leverage scores of $x$ near the edge of the interval with the leverage scores of $x$ closer to the center. It is not immediately clear how to extend such an argument to the case when $p$ is the Gaussian or Laplace density, but we believe that doing so may be possible. In general, we conjecture that a simple closed form leverage score bound achieving within a constant factor of the optimal sample complexity exists.

\section{Experimental Results}
\label{sec:experiments}
We now illustrate the potential of Fourier sparse leverage score bounds by empirically evaluating the modified random Fourier features (RFF) method of Section \ref{sec:oblivious}. We implement the method without the final JL projection, and use simplifications of the frequency distributions from Theorems \ref{thm:gaussian} and \ref{thm:exp}, which work well in experiments. For data in $\R^d$ for $d > 1$, we extend these distributions to their natural spherically symmetric versions. See Section \ref{app:experiments} for details and Figure \ref{fig:2d_vis} for a visualization.

\begin{figure}[h]
	\begin{subfigure}{.24\textwidth}
		\centering
		% include first image
		\includegraphics[width=\linewidth]{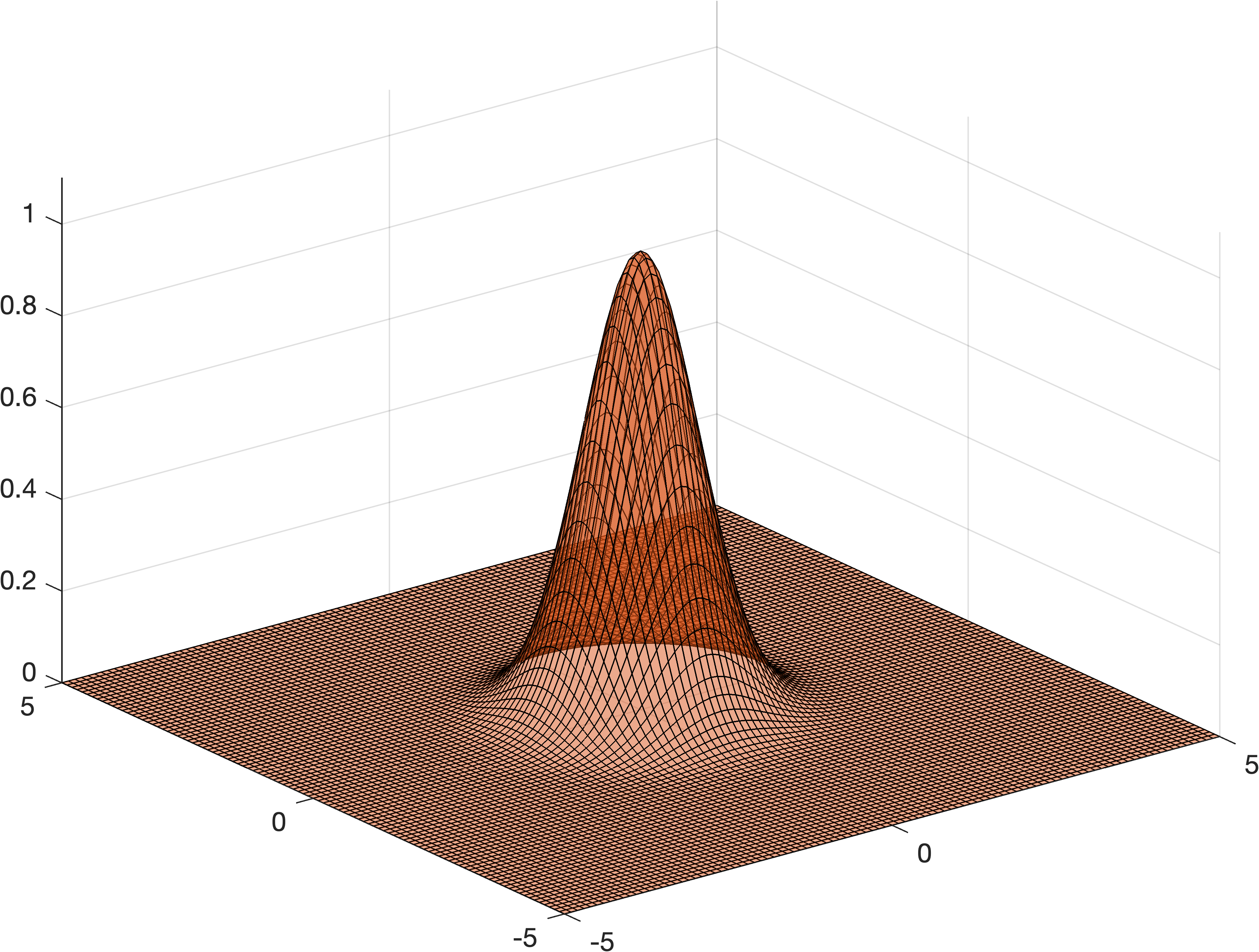}  
		\caption{Classical RFF Distribution, Gaussian kernel.}
	\end{subfigure}~
	\begin{subfigure}{.24\textwidth}
		\centering
		% include second image
		\includegraphics[width=\linewidth]{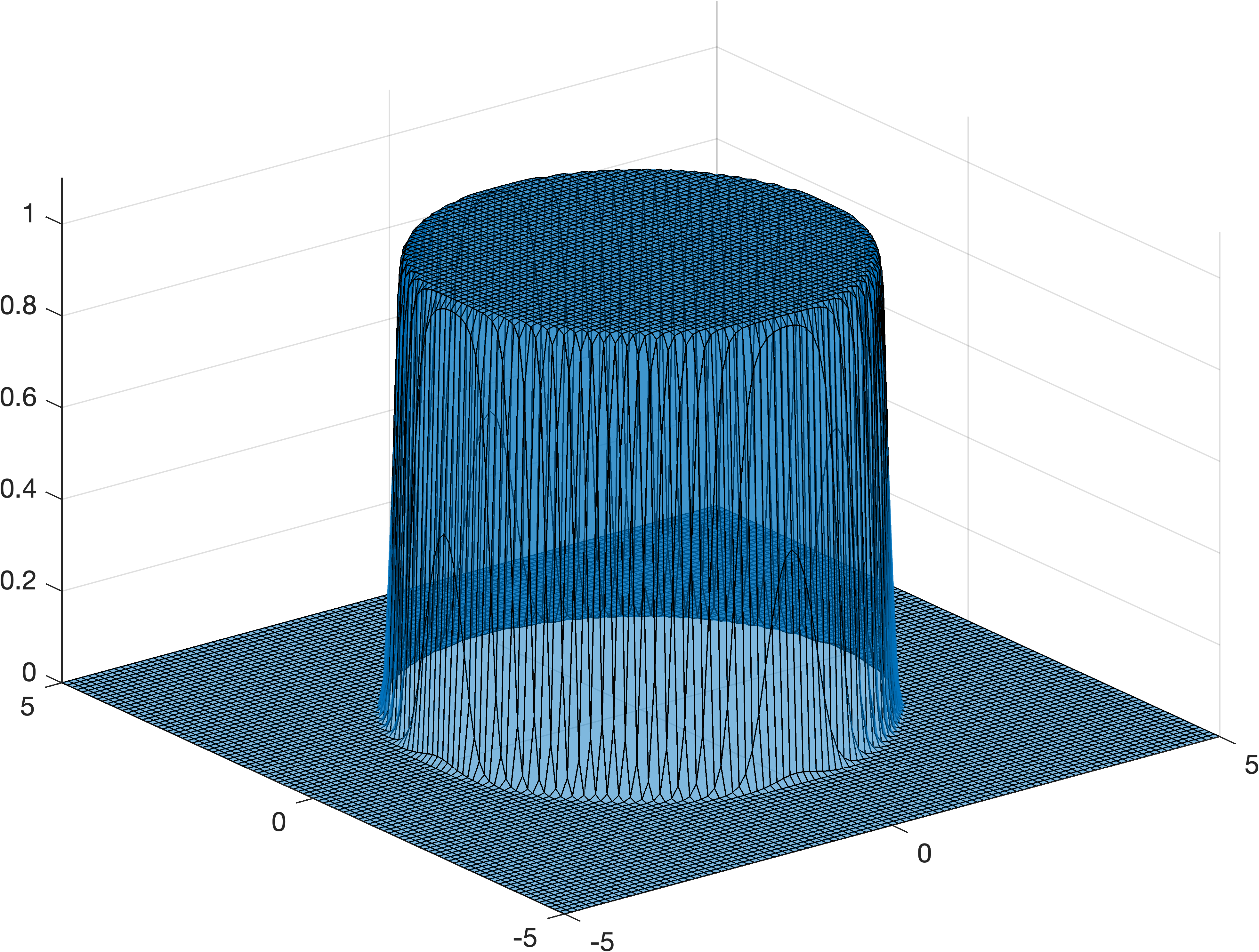} 
		\caption{Modified RFF Distribution, Gaussian kernel.}
	\end{subfigure}~
	\begin{subfigure}{.24\textwidth}
	\centering
	% include second image
	\includegraphics[width=\linewidth]{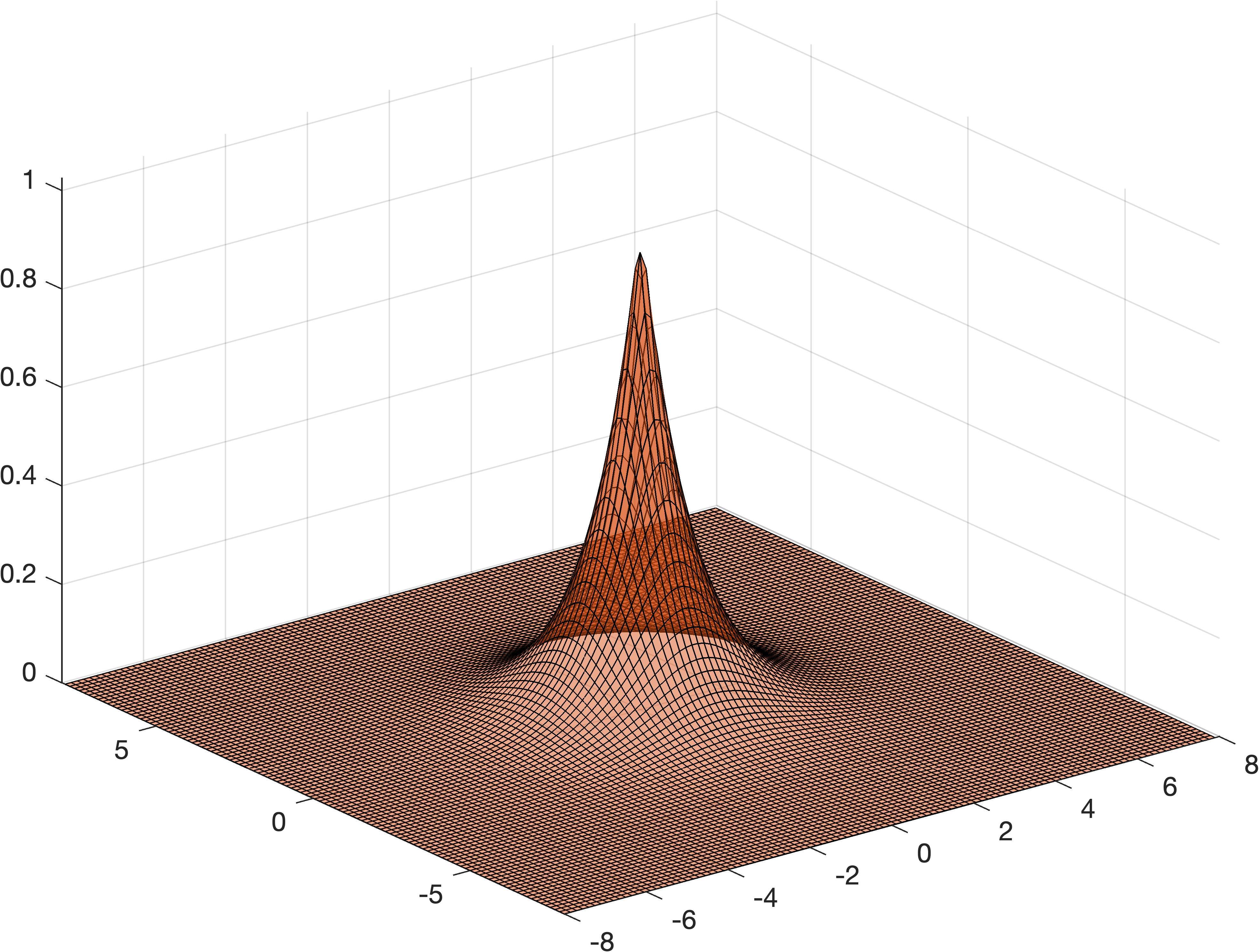} 
			\caption{Classical RFF Distribution, Cauchy kernel.}
\end{subfigure}~
	\begin{subfigure}{.24\textwidth}
	\centering
	% include second image
	\includegraphics[width=\linewidth]{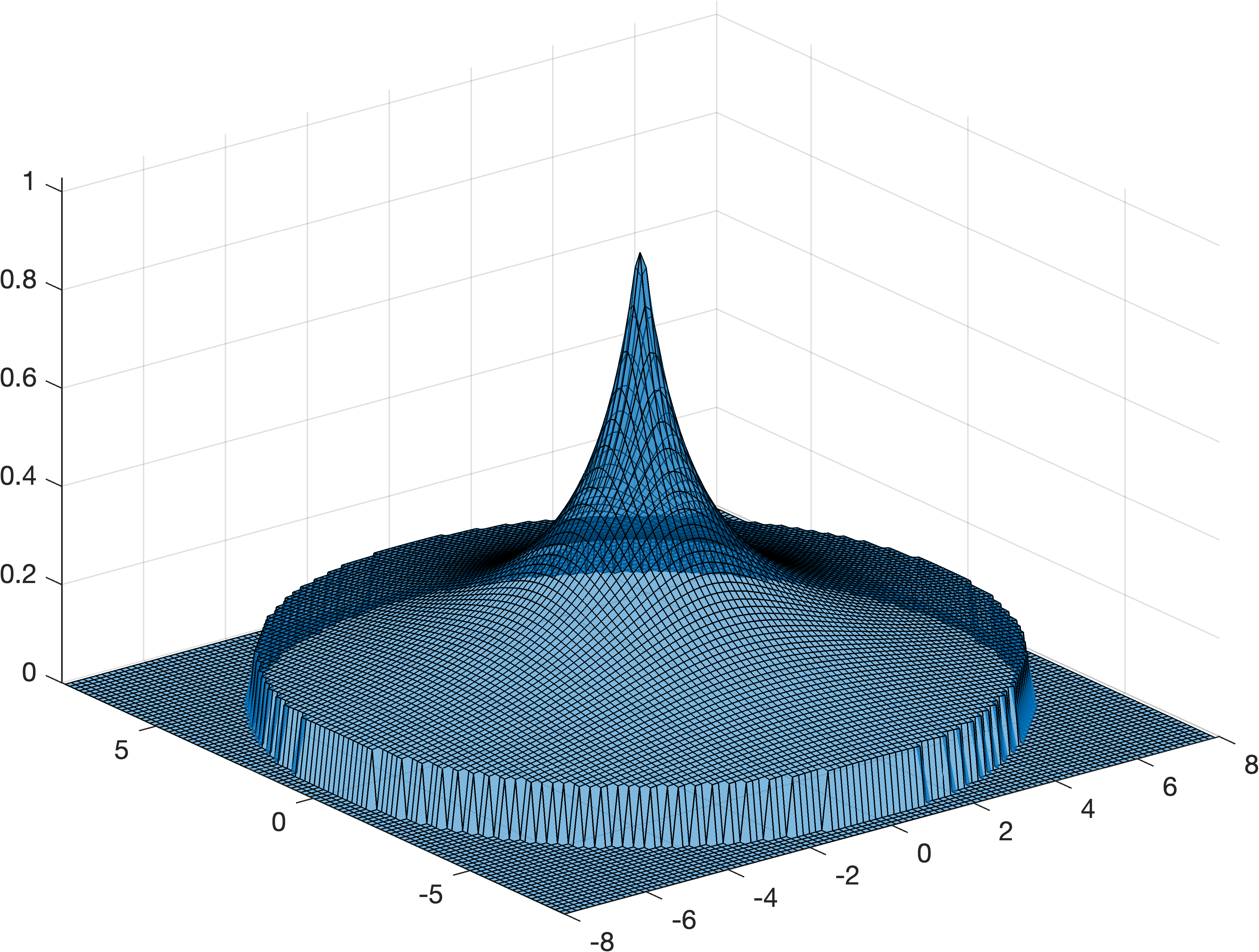} 
	\caption{Modified RFF Distribution, Cauchy kernel.}
\end{subfigure}
	\caption[]{Distributions used to sample random Fourier features frequencies $\eta_1, \ldots, \eta_m$. The ``Classical RFF'' distributions are from the original paper by Rahimi, Recht \cite{RahimiRecht:2007}. The ``Modified RFF'' distributions are simplified versions of the leverage score upper bounds from Thoerems \ref{thm:gaussian} and \ref{thm:exp}. Notably, our modified distributions sample \emph{high frequencies} (i.e. large $\ell_2$ norm) with higher probability than Classical RFF, leading to theoretical and empirical improvements in kernel approximation.}
	\label{fig:2d_vis}
	%\vspace{-.5em}
\end{figure}
We compare our method against the classical RFF method on a kernel ridge regression problem involving precipitation data from Slovakia \cite{NetelerMitasova:2013}, a benchmark GIS data set, which is representative of many 2D function interpolation problems. 
See Figure \ref{fig:data_vis} for a description.
The regression solution requires computing $(\bv{K} + \lambda \bv I)^{-1}\bv{y}$, where $\bv{y}$ is a vector of training data. Doing so with a direct method is slow since $\bv{K}$ is large and dense, so an iterative solver is necessary. However, when cross validation is used to choose a kernel width $\sigma$ and regularization parameter $\lambda$, the optimal choices lead to a poorly conditioned system, which leads to slow convergence. 

\begin{figure}[h]
	\begin{subfigure}{.42\textwidth}
		\centering
		% include first image
		\includegraphics[width=\linewidth]{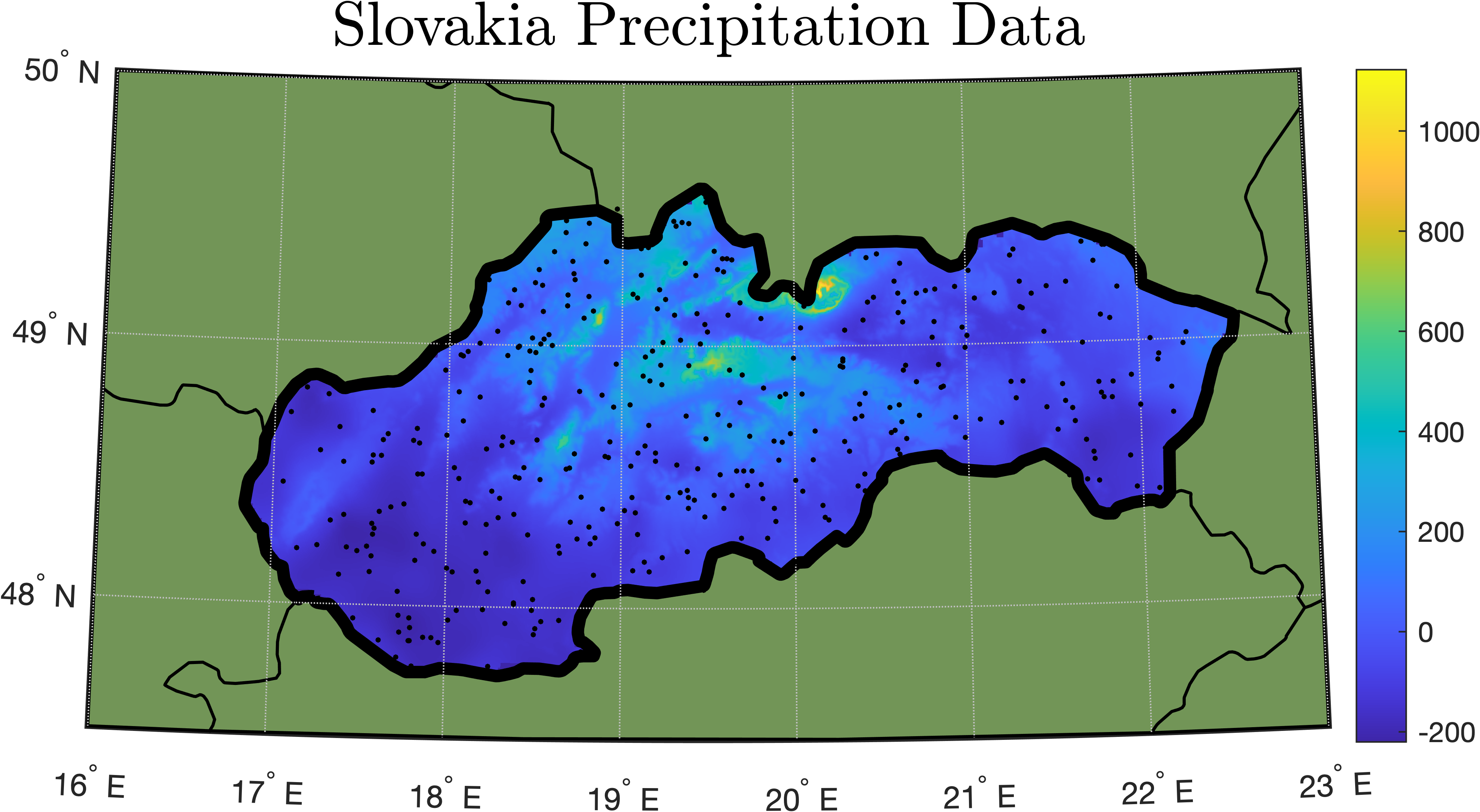}  
				\caption{Original precipitation data with samples.}
	\end{subfigure}
	\hfill
	\begin{subfigure}{.42\textwidth}
		\centering
		% include second image
		\includegraphics[width=\linewidth]{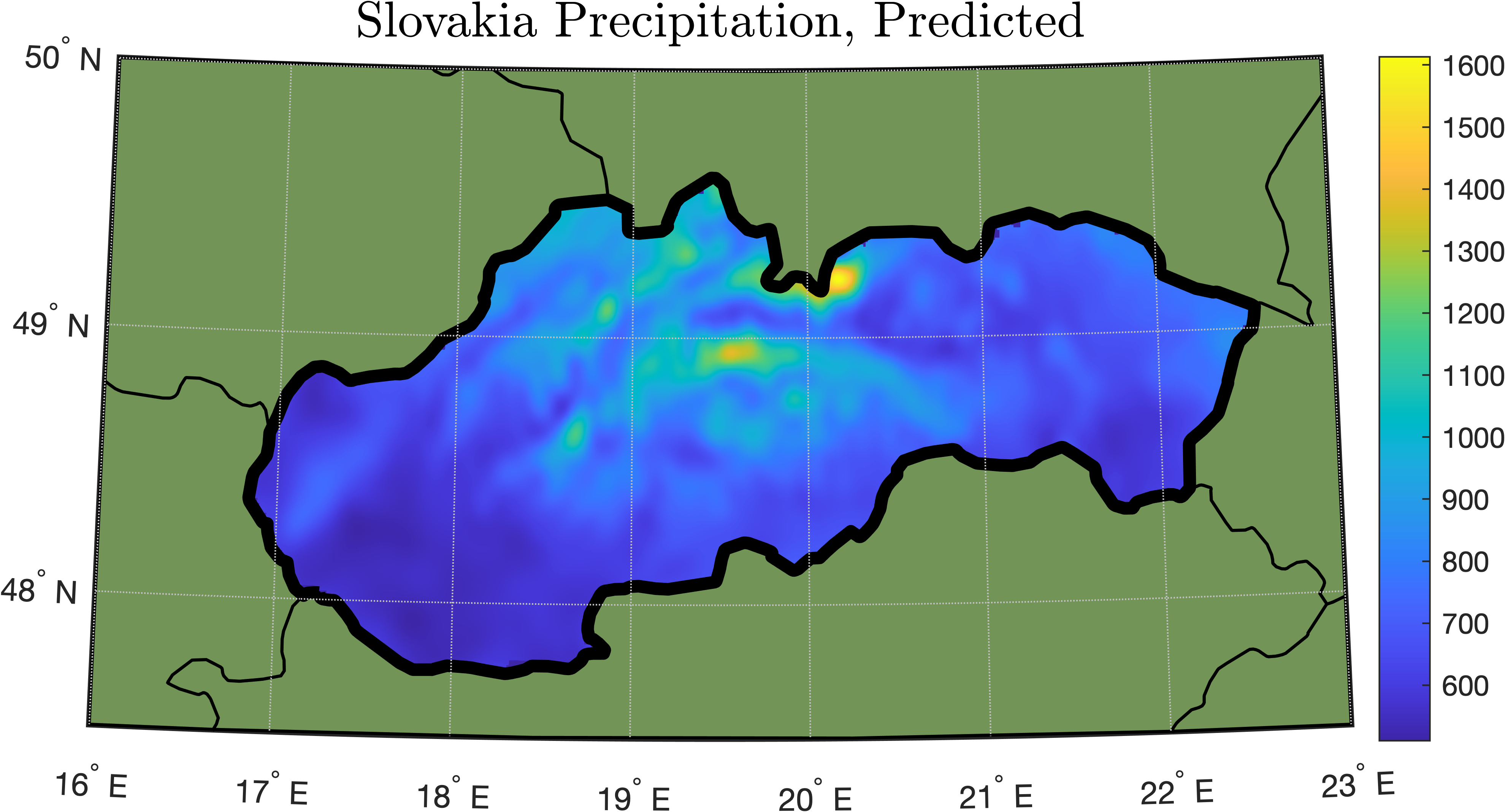} 
				\caption{Interpolated precipitation data.}
	\end{subfigure}
	\caption{The left image shows precipitation data for Slovakia in mm/year 
		Data was constructed using experimental methods and an advanced runoff analysis and is available for
		at n = 196k locations on a regular lat/long grid \cite{NetelerMitasova:2013}. Our goal is to approximate this precipitation function based on 6400 training samples from randomly selected locations (visualized as black dots). The right image shows the prediction given by a kernel regression model with Gaussian kernel, which was computed efficiently using our modified random Fourier method along with a preconditioned CG method.}
	\label{fig:data_vis}
\end{figure}

There are two ways to solve the problem faster using a kernel approximation: either $\tilde{\bv{K}}$ can be used in place of $\bv{K}$ when solving $(\bv{\tilde K} + \lambda \bv I)^{-1}\bv{y}$, or it can be used as a preconditioner to accelerate the iterative solution of $(\bv{K} + \lambda \bv I)^{-1}\bv{y}$. We explore the later approach because \cite{AvronKapralovMusco:2017} already empirically shows the effectiveness of the former. While their modified RFF algorihm is different than ours \emph{in theory}, we both make similar practical simplifications (see Section \ref{app:experiments}), which lead our empirically tested methods to be almost identical for the Gaussian kernel. Results on preconditioning are shown in Figure \ref{fig:regress_error}. Our modified RFF method leads to substantially faster convergence for a given number of random feature samples, which in turn leads to better downstream prediction error.
%We evaluated how convergence improves by using the precondition conjugate gradient (PCG) method with preconditioner  $\tilde{\bv{K}}$ obtained via a random Fourier features embedding. For details, see Appendix \ref{app:experiments}. Results are 

The superior performance of the modified RFF method can be explained theoretically: our method is designed to target the spectral approximation guarantee of Definition \ref{def:spectral_gaurantee}, which \emph{is guaranteed to ensure good preconditioning} for $\bv{K} + \lambda \bv{I}$ \cite{AvronKapralovMusco:2017}. On the other hand, the classical RFF method actually achieves better error than our method in other metrics like $\|\bv{K} - \tilde{\bv{K}}\|_2$, both in theory \cite{Tropp:2015} and empirically (Figure \ref{fig:regress_error}). However,  for preconditioning, such bounds  will not necessarily ensure fast convergence. The key observation is that the spectral guarantee requires better approximation in the \emph{small} eigenspaces of $\bv{K}$. By more aggressively sampling higher frequencies that align with these directions (see Figure \ref{fig:2d_vis}) the modified method obtains a better approximation.

\begin{figure}[h]
	\begin{subfigure}{.44\textwidth}
		\centering
		% include first image
		\includegraphics[width=\linewidth]{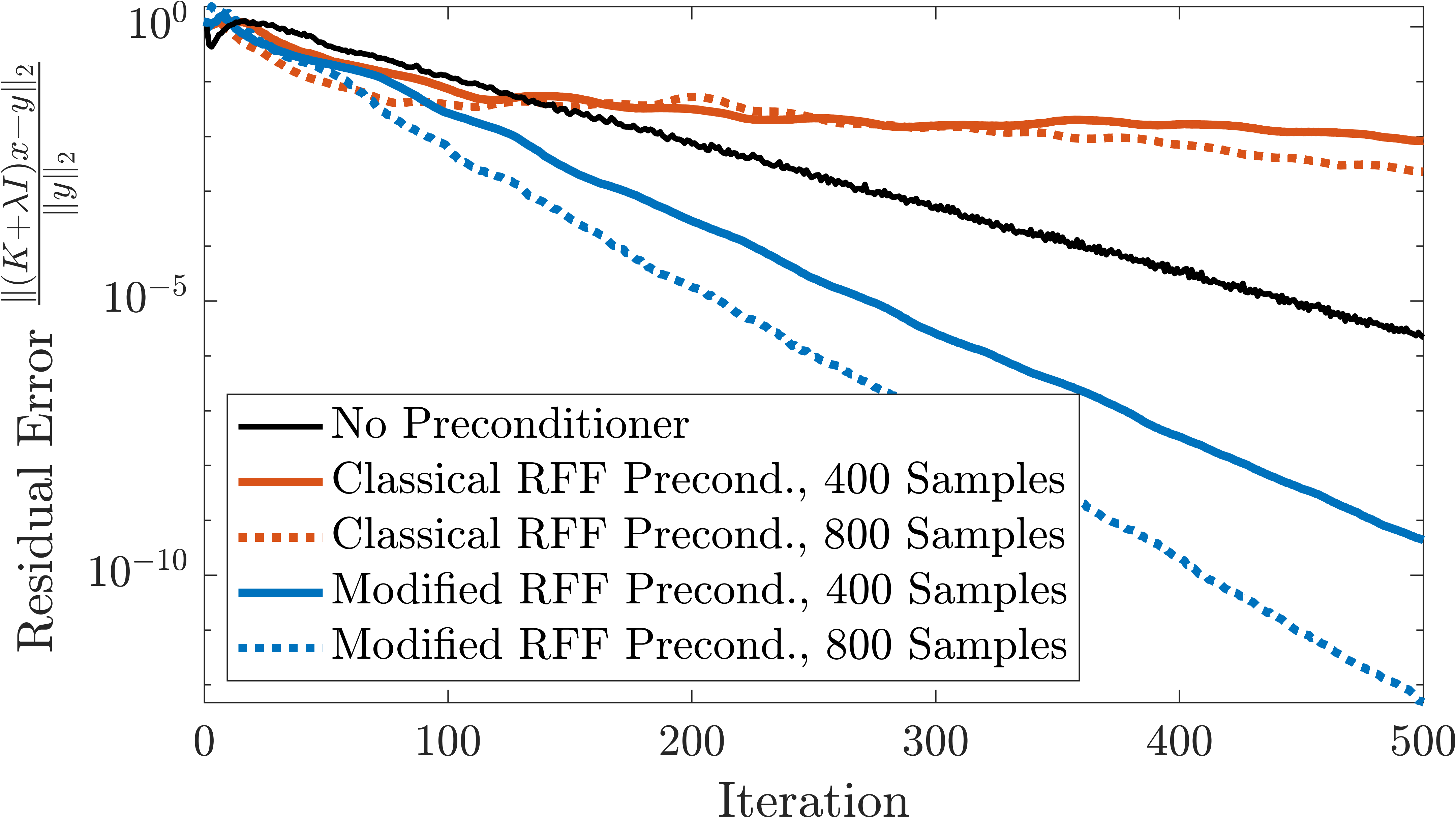}  
		\caption{Preconditioned CG Convergence.}
	\end{subfigure}\hfill
	\begin{subfigure}{.44\textwidth}
	\centering
	% include second image
	\includegraphics[width=\linewidth]{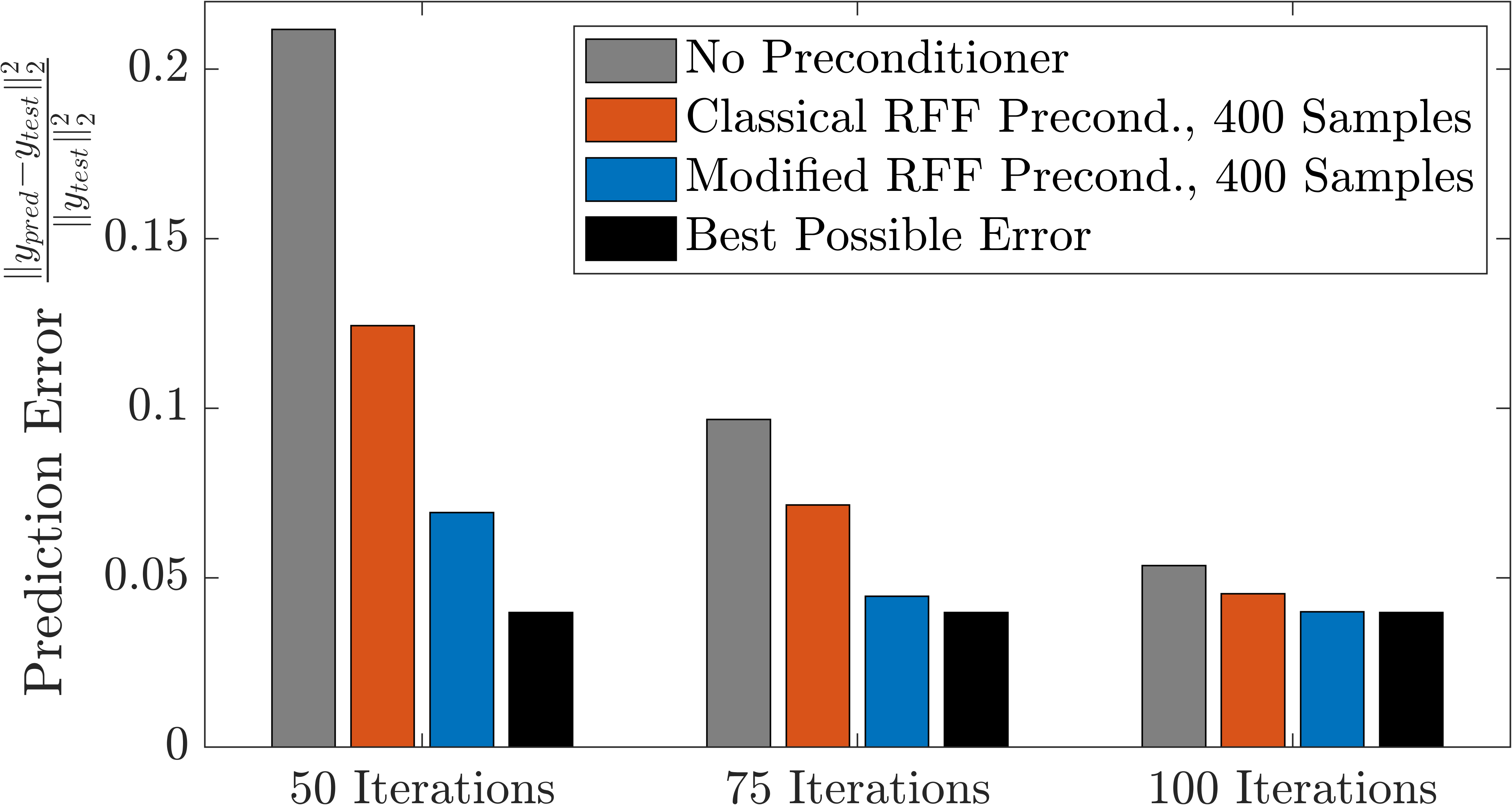} 
		\caption{Resulting test error for kernel regression.}
\end{subfigure}
\caption[]{The left plot shows residual convergence when solving $\min_{\bv{x}}\|(\bv{K} + \lambda \bv{I})\bv{x} - \bv{y}\|$ using PCG. Baseline convergence (the black line) is slow, so we preconditioned with both a classical RFF approximation and our modified RFF approximation. Classical RFF accelerates convergence in the high error regime, but slows convergence eventually. Our method significantly accelerates convergence, with better performance as the number of RFF samples increases. On the right, we show that better system solve error leads to better downstream predictions. The black bar represents the relative error of a prediction computed by exactly inverting $\bv{K} + \lambda \bv{I}$. An approximate solution obtained using our preconditioner approaches this ideal error more rapidly than the other approaches. }
\label{fig:regress_error}
%\vspace{-1em}
\end{figure}

%
%To better understand these results, we directly computed the spectral norm error between the randomized RFF kernel approximations and $\bv{K}$.
% Perhaps surprising, we see in Figure \ref{fig:regress_error} that the classical RFF method achieves \emph{better} spectral norm error , despite our modified methods superior preconditioning performance. 
%We believe this can explained theoretically: the classical RFF sampling distribution actually leads to stronger bounds on $\|\bv{K} - \tilde{\bv{K}}\|_2$ than our approach, as proven in . 
%However,
\begin{figure}[h]
	\begin{subfigure}{.44\textwidth}
		\centering
		% include first image
		\includegraphics[width=\linewidth]{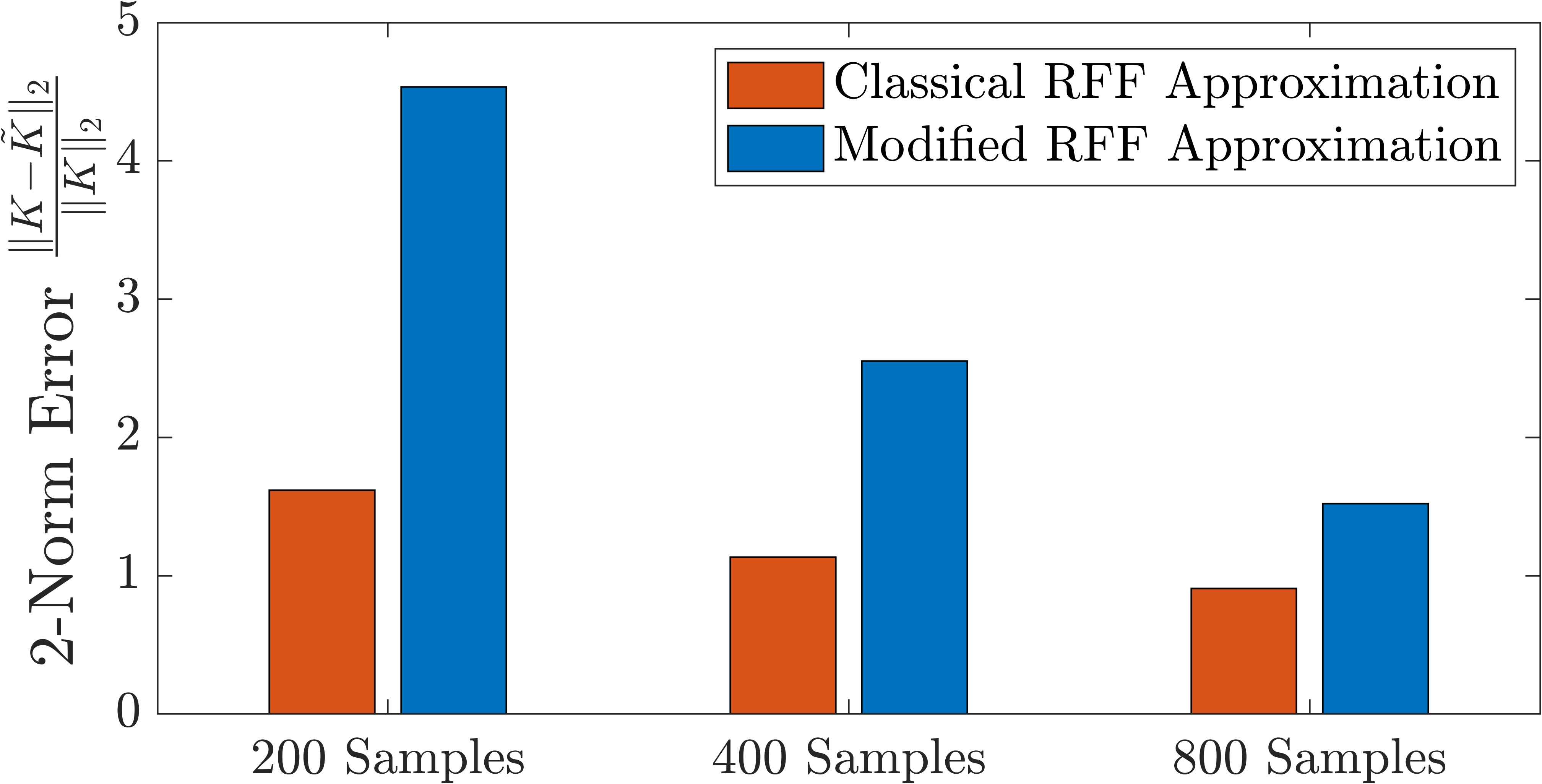}  
		\caption{Spectral Norm Error.}
	\end{subfigure}\hfill
	\begin{subfigure}{.44\textwidth}
		\centering
		% include second image
		\includegraphics[width=\linewidth]{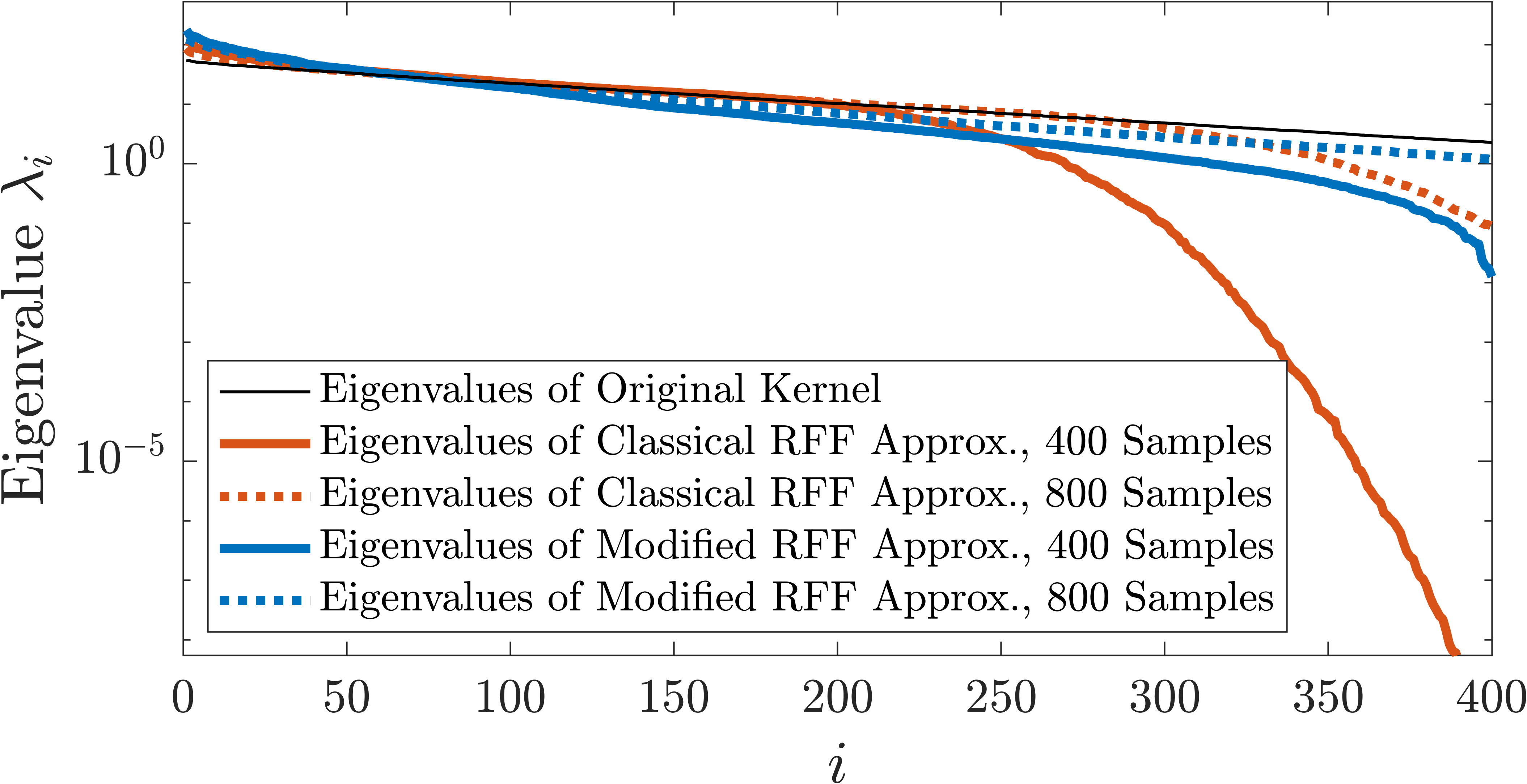} 
		\caption{Eigenvalues Comparison.}
	\end{subfigure}
	\caption[]{The left plot compares relative spectral norm errors for randomized kernel approximations for a Gaussian kernel matrix $\bv{K}$. The classical RFF method actually has \emph{better} error. However, as shown in the right plot, the modified method better approximates the small eigenvalues of $\bv{K}$, which is necessary for effective preconditioning as it leads to a better relatively condition number.}
	\label{fig:preconditioning}
	\vspace{-1em}
\end{figure}

\vspace{-.5em}
\subsection{Further details for experiments}
\label{app:experiments}

We now describe some details of the modified random Fourier features sampling algorithm and preconditioning approach used in our empirical evaluation above.

\paragraph{Details of Sampling.}
In our results, we employ simplified versions of the leverage score upper bounds from  Theorems \ref{thm:gaussian} and \ref{thm:exp}. In particular, in both of these theorems, the leverage score upper bound distributions are piecewise, following a different functions for frequencies above and below a certain cutoff $F$. $F$ equals $6\sqrt{2}\sigma \cdot  \sqrt{s}$ and  $9\sqrt{2} \sigma \cdot {s}$ in Theorems \ref{thm:gaussian} and \ref{thm:exp}, respectively.
The bulk of each distribution is on values of $|\eta| \leq F$, so we ignore the ``tail'' part of each distribution when sampling. This does not seem to significantly effect the experimental results. More over, using the empirical leverage score score distributions from Figure \ref{fig:upper_bound} as guidance, we used tighter values for $F$ than we were able to prove theoretically. For example, setting $F = 4\sigma$ seems sufficient to capture the bulk of the Fourier sparse leverage score distribution for the Gaussian measure, so this is the value we used in our experiments. I.e. samples of $\eta$ were drawn uniformly from the ball $\{\eta: \|\eta\|_2 < 4\sigma\}$.

When sampling we also use the same trick from \cite{RahimiRecht:2007} to achieve a real valued embedding, which makes it easier to work with the embedding downstream (e.g., when implementing the preconditioned solver). In particular, instead of including $C\cdot e^{-2 \pi i \eta^T x}$ in the embedding, where $C$ is the appropriate constant as in Definition \ref{def:rff}, we can include an entry equal of $C \cdot \cos(2 \pi\eta^T x + \beta)$ where $\beta$ is a uniform random variable from $[0,2\pi]$. It's not hard to check that the corresponding real valued embedding will still satisfy $\E[\bv{G}^*\bv{G}] = \bv{K}$, and experimentally, approximation quality does not appear to suffer. 

\paragraph{Details of Preconditioning.}
When solving $(\bv{K} + \lambda \bv I)^{-1}\bv{z}$ with a preconditioner, each iteration of the preconditioned solver requires 1) computing $(\bv{\tilde K} + \lambda \bv I)^{-1}\bv{z}$ for some vector $\bv{z}$ and 2) multiplying $\bv{K} + \lambda \bv{I}$ by a vector $\bv{w}$. The first step can be done efficiently whenever $\bv{\tilde K} = \bv{G}^*\bv{G}$ where $\bv{G} \in \C^{n\times m}$, which is the type of approximation we get from a random Fourier features method. In particular, let $\bv{G} = \bv{U}\bs{\Sigma}\bv{V}^T$ be $\bv{G}$'s singular value decomposition. Due to the simplification discussed above, $\bv{G}$ is always real-valued in our setting, and so is its SVD. We have $\bv{U}\in \R^{m\times m}$, $\bs{\Sigma}\in \R^{m\times m}$, and $\bv{V}\in \R^{n\times m}$. The SVD can be computed in $O(m^2n)$ time and more importantly, this operation is very fast when $\bv{G}$ fits in memory, which is often possible even when $\bv{K}\in \R^{n\times n}$ does not. So, for both classical RFF preconditioning and modified RFF preconditioning, we choose values for $m$ that allow for fast computation of the SVD, and compute the decomposition as a preprocessing step. 

Then, it is not hard to check that $(\bv{\tilde K} + \lambda \bv I)^{-1}\bv{z} = \bv{V}\left(\bs{\Sigma} + \lambda \bv{I}_{m\times m}\right)^{-1}\bv{V}^T \bv{z} + \frac{1}{\lambda}(\bv{z} - \bv{V}^T \bv{z})$, which can be computed in $O(mn)$ time. This is much faster than the cost of multiplying a vector by $\bv{K} + \lambda \bv{I}$, so the cost of preconditioning ends up being a lower order term in the solver complexity: it increases the cost of each iteration by just a small factor. 

\vspace{-.5em}
\section{Fourier Sparse Leverage Score Bounds -- Proofs}\label{app:bounds}

We now prove our main Fourier sparse leverage score bounds under the Gaussian and Laplace densities (Theorems \ref{thm:gaussian} and \ref{thm:exp}). When the minimum gap between frequencies in $f \in \mathcal{T}_s$ is bounded, we also give an improved bound based on Ingham's inequality.%, applicable e.g. in our oblivious embedding results when data points are separated by a minimum distance.

For notation in this section, we let $\norm{f}_2^2 = \int_{x \in \R} |f(x)|^2 dx$  denote the $L_2$ norm of any complex valued function $f: \R \rightarrow \C$. 
%$L_2$ denotes the space of square integrable functions, i.e., $f$ with $\norm{f}_2^2 < \infty$. %For $f \in L_2$ we denote the $L_2$ norm by $\norm{f}_2^2 = \int_{x \in \R} |f(x)|^2 dx$. 
We denote the $L_2$ norm over an interval by $\norm{f}_{[a,b]}^2 = \int_{a}^b |f(x)|^2 dx$ and the $L_2$ norm under any density $p$ over $\R$ as $\norm{f}_p^2 = \int_{x \in \R} |f(x)|^2 \cdot p(x) \ d x$. %We denote the $L_\infty$ norm over an interval by $\norm{f}_{L_\infty[a,b]} \eqdef \max_{x \in [a,b]} |f(x)|$.

\subsection{Foundational bounds}

We build on a number of existing bounds on the uniform density leverage scores and related concentration properties of an extended class of Fourier sparse functions with possibly complex frequencies. This class and its variants have been studied extensively, e.g., in  \cite{Turan:1984,Nazarov:1993,BorweinErdelyi:1995,BorweinTamas:1996,borwein2000pointwise,Kos:2008,lubinsky2015dirichlet,Erdelyi:2017}.
\begin{align}\label{eq:calE}
\mathcal{E}_s = \left \{f: f(x) = \sum_{j=1}^s a_j e^{\lambda_j x}, a_j \in \C, \lambda_j \in \C \right \}.
\end{align}
We also consider the subclasses where $\mathcal{E}_s^+$ and $\mathcal{E}_s^-$,  which are defined analogously to $\mathcal{E}_s$ but with frequencies $\lambda_j \in \C$ required to have non-negative (respectively, non-positive) real components. Note that our main class of interest $\mathcal T_s$ defined in \eqref{eq:ts} is contained in all three of these extended classes. %For notational conciseness we denote the leverage score $\tau_{\mathcal{E}_s,p}(x)$ by 

We first use a bound on the uniform density leverage score at any point $x$ on an interval, in terms of its distance from the edge of the interval.
\begin{lemma}\label{lem:unif} For any $a,b \in \R$ with $a < b$, $x \in (a,b)$, and $f \in \mathcal{E}_s$ with $f \not\equiv 0$:
\begin{align*}
\frac{|f(x)|^2}{\norm{f}_{[a, b]}^2} \le \frac{s}{\min(x-a,b-x)}.
\end{align*} 
\end{lemma}
Lemma \ref{lem:unif} is stated, up to a constant factor $2$ in Theorem 7.1 \cite{Erdelyi:2017}. We prove it here for completeness and improve this constant.
\begin{proof}
It is shown in equation (3) of \cite{BorweinErdelyi:2006} that for any $g \in \mathcal{E}_s$ with $g \not\equiv 0$,
\begin{align}\label{eq:borwein}
\frac{|g(0)|^2}{\norm{g}_{[-1,1]}^2} \le s.
\end{align}
%\end{proof}
%
%, using the following lemma, shown in \cite{BorweinErdelyi:2006}:
%\begin{lemma}\label{lem:unif0} For any $a,b \in \R$ with $a < b$, any $\delta \in \left (0,\frac{1}{2}(b-a) \right )$, and any $f \in \mathcal{E}_s$ with $f \not\equiv 0$:
%\begin{align*}
%\frac{\norm{f}_{L_\infty[a+\delta,b-\delta]}^2}{\norm{f}_{[a,b]}^2} \le \frac{s}{\delta}.
%\end{align*}
%\end{lemma}
%Lemma \ref{lem:unif0} is obtained directly from equation (3) in \cite{BorweinErdelyi:2006}, after shifting and scaling $f$ to $f'(x) = f \left (\frac{x-}{d}\right )$
%\begin{proof}[Proof of Lemma \ref{lem:unif}]
%Let $c = (a+b)/2$ and assume that $x \in [c,b]$. 
For $x \in (a,b)$, let $\delta = \min(x-a,b-x)$ and $g(z) = f \left ( x - \delta \cdot z\right )$. 
Note that if $f \in \mathcal E_s$ and $f \not\equiv 0$, we have $g \in \mathcal E_s$ and $g \not\equiv 0$. Additionally, we have $g(0) = f(x)$ and $\norm{f}_{[a,b]}^2 \ge \norm{f}_{[x-\delta,x+\delta]}^2 = \delta \cdot \norm{g}_{[-1,1]}^2$. Applying \eqref{eq:borwein} we then have:
%
%Then applying Lemma \ref{lem:unif0} with $\delta = b-x$ implies for all $f \in \mathcal E_s$:
\begin{align*}
\frac{|f(x)|^2}{\norm{f}_{[a,b]}^2} \le \frac{|g(0)|^2}{\delta \cdot \norm{g}_{[-1,1]}^2} \le \frac{s}{\delta},
\end{align*}
which completes the proof.
%Similarly, when $x \in [a,c]$, applying  Lemma \ref{lem:unif0} with $\delta = x-a$, implies for all $f \in \mathcal E_s$:
%\begin{align*}
%\frac{|f(x)|^2}{\norm{f}_{[a,b]}^2} \le \frac{\norm{f}_{L_\infty[x, a+b-x]}^2}{\norm{f}_{[a,b]}^2} = \frac{\norm{f}_{L_\infty[a+\delta,b-\delta]}^2}{\norm{f}_{[a,b]}^2} \le  \frac{2s}{\delta} = \frac{2s}{x-a}.
%\end{align*}
%Together these two bounds give Lemma \ref{lem:unif}.
\end{proof}
We note that Lemma \ref{lem:unif} can be combined with Lemma 3.2 of \cite{Denisov:2016}, which tightens bounds proven in \cite{Erdelyi:2017} and \cite{Kos:2008} to give the following bound for the uniform density leverage scores:  
\begin{corollary}[Uniform Density Leverage Score Bound]\label{thm:unif}
	Consider the uniform density $u(x) = \frac{1}{2\sigma}$ for $x \in [-\sigma,\sigma]$, $u(x) = 0$ otherwise, and let
	\begin{align*}
	\bar \tau_{s,z}(x) = \begin{cases} 
	\frac{s}{\sigma-|x|}\text{ for } |x| \leq \sigma (1-\frac{4}{\pi s}) \\
	\frac{\pi}{4\sigma}s^2 \text{ for } \sigma (1-\frac{4}{\pi s}) < |x| \leq \sigma \\
	0 \text{ for } |x| > \sigma
	\end{cases} 
	\end{align*}
	We have $\bar \tau_{s,u}(x) \ge \tau_{s,u}(x)$ for all $x \in \R$ and $\int_{-\infty}^{\infty}\bar \tau_{s,u}(x)\ dx = 2s(1 + \ln(\frac{\pi}{4} s)) = O(s\ln s)$.
\end{corollary}
Corollary \ref{thm:unif} mirrors our Theorems \ref{thm:gaussian} and \ref{thm:exp}, and as mentioned in Section \ref{sec:main_tech}, no upper bound can improve on the integral of $O(s\ln s)$ by more than a $\ln s$ factor. Understanding if this $\ln s$ can be eliminated or if it is necessary is an interesting open question.

We also employ a bound  due to Tur\'{a}n \cite{BorweinErdelyi:1995}, which plays a central role in his book \cite{Turan:1984}. %It is often referred to as Tur\'{a}n's lemma.
\begin{lemma}[Tur\'{a}n's lemma]\label{lem:growth0}
For any $g \in \mathcal E_s^+$ and $\alpha,\beta > 0$:
\begin{align*}
|g(0)| \le \left (\frac{2e(\alpha+\beta)}{\beta} \right )^{s} \cdot \norm{g}_{[\alpha,\alpha+\beta]}.
\end{align*}
\end{lemma}

Tur\'{a}n's lemma can be used to bound the growth of any function in $\mathcal{E}_s^- \supset \mathcal{T}_s$ outside of an interval in terms of its norm on that interval.
\begin{lemma}[Lemma 12.2 \cite{Erdelyi:2017}]\label{lem:growth} For any $a \in \R$, $d  > 0$, $x \ge a+d$, and $f \in \mathcal{E}_s^-$:
\begin{align*}
|f(x)| \le \left ( \frac{2e (x-a)}{d} \right )^{s} \cdot  \norm{f}_{[a, a+d]}.
\end{align*} 
%where $\mathcal{E}_s^- \supset \mathcal{T}_s$ is defined analogously to $\mathcal{E}_s$ in \eqref{eq:calE} but with frequencies $\lambda_j \in \C$ required to have non-positive real components. 
\end{lemma}
\begin{proof}
Let $f \in {\mathcal E}_s^{-}$. Let $g \in {\mathcal E}_s^+$ be defined by $g(t) := f(x-t)$.
We define $\alpha := x-(a+d)$ and $\beta := d$.
Applying Lemma \ref{lem:growth0} with $g \in {\mathcal E}_s^+$ we have
$$|f(x)| = |g(0)| \leq \left( \frac{2e(\alpha+\beta)}{\beta} \right)^s \|g\|_{[\alpha,\alpha+\beta]}  
= \left( \frac{2e(x-a)}{d} \right)^s \|f\|_{[a,a+d]}.$$
%which gives the bound after squaring both sides and rearranging.
\end{proof}

Finally, our gap-based result apply to the following restricted class of $\mathcal T_s$:
\begin{align}\label{eq:tsGap}
\mathcal{T}_{s,\gamma} = \left \{f: f(x) = \sum_{j=1}^s a_j e^{i \lambda_j x}, a_j \in \C, \lambda_j \in \R\text{ with } \min_{j,k} |\lambda_k - \lambda_{j}| \ge \gamma > 0 \right \}.
\end{align}
We denote the leverage score of this class with respect to a density $p$ by $\tau_{s,\gamma,p}(x)$. In bounding these scores we use the following bound due to Ingham \cite{ingham1936some}:
\begin{lemma}[Ingram's Inequality]\label{lem:gap} 
For any $\gamma > 0$,  $f \in \mathcal{T}_{s,\gamma}$ with coefficients $a_1,\ldots,a_s$,
%Suppose $\lambda_1 < \lambda_2 < \cdots < \lambda_s$ satisfy the gap condition
%\begin{align}\label{eq:gap}
%\lambda_k - \lambda_{k-1} \ge \gamma > 0, \quad\forall k = 2,3,\ldots,s.
%\end{align} 
%Let $f(x) = \sum_{j=1}^s a_j e^{i \lambda_j x}$ for $a_j \in \C$.
and $T > \pi/\gamma, $ %there exist constants $c_1 = c_1(T,\gamma) > 0$ and $c_2 = c_2(T,\gamma) > 0$ such that
$$c_1(T,\gamma) \sum_{j=1}^s{|a_j|^2} \leq \norm{f}_{[-T,T]}^2 
\leq c_2(T,\gamma) \sum_{j=1}^s{|a_j|^2}\, ,$$
where 
$$c_1(T,\gamma) := \frac{4T}{\pi}\left(1 - \frac{\pi^2}{T^2\gamma^2} \right) \qquad \text{and} \qquad
c_2(T,\gamma) := \frac{16T}{\pi}\left(1 + \frac{\pi^2}{T^2\gamma^2} \right).$$
\end{lemma}

Setting $T = 2\pi/\gamma$ in Ingram's inequality gives:
\begin{corollary}\label{cor:gap} 
For any $\gamma > 0$ and $f \in \mathcal{T}_{s,\gamma}$ with coefficients $a_1,\ldots,a_s$, we have:
%Suppose $\lambda_1 < \lambda_2 < \cdots < \lambda_s$ satisfy the gap condition \eqref{eq:gap} and $f(x) = \sum_{j=1}^s a_j e^{i \lambda_j x}$ for $a_j \in \C$. We have:
\begin{align*}
\frac{6}{\gamma} \sum_{j=1}^s{|a_j|^2} \leq \norm{f}_{[-2\pi/\gamma,2\pi/
\gamma]}^2 
\leq  \frac{40}{\gamma} \sum_{j=1}^s{|a_j|^2}.
\end{align*} 
\end{corollary}

\subsection{Bounds for the Gaussian density}\label{sec:gauss}

Our leverage score bound for the Gaussian density (Theorem \ref{thm:gaussian}) is split into two components -- a uniform bound on $\tau_{s,g}(x)$ for all $ x\in \R$ (Claim \ref{clm:small}) and a bound for $x$ restricted to have sufficiently large magnitude (Claim \ref{clm:large}). Combining these two results gives the two part bound of Theorem \ref{thm:gaussian}. In this section we focus solely on the unit width Gaussian density: $g(x) = \frac{1}{\sqrt{\pi}} e^{-x^2}$. Bounds under this density can immediately be translated into bounds for any width $\sigma > 0$ via scaling. While they are not applicable to our algorithmic results, we give leverage score lower bounds as well, which help clarify the tightness of the bounds given.

%\Cam{Corresponds to Thm 2.1 of Tamas write up but slightly tighter constant. TODO: Include lower bound in appendix or here.}
\begin{claim}[Gaussian Leverage Bound  -- Uniform Bound]\label{clm:small}
Letting $g(x) = \frac{1}{\sqrt{\pi}} e^{-x^2}$, for all $x \in \R$:
\begin{align*}
\frac{s}{3\pi} \le \tau_{\mathcal{E}_s,g} \le e\cdot s.
\end{align*}
As a consequence $\tau_{s,g}(x) \le e\cdot s$.
\end{claim}
\begin{proof}
For any $f \in \mathcal{E}_s$ and $a \in \R$, define the shifted and weighted function $w_a(x) = f(x+a) \cdot e^{-(x+a)^2/2}.$
We can write:
\begin{align*}
w_a(x) &= \sum_{j=1}^s a_j e^{ i \lambda_j (x+a) } e^{-x^2/2} e^{-a^2/2} e^{-x a}\\
&= e^{-x^2/2} \cdot  \sum_{j=1}^s \left (a_j  \cdot e^{i \lambda_j a} \cdot e^{-a^2/2} \right ) \cdot e^{ (i \lambda_j - a) x}.
\end{align*}
If we let $h_a(x)=  \sum_{j=1}^s \left (a_j  \cdot e^{i \lambda_j a} \cdot e^{-a^2/2} \right ) \cdot e^{ (i \lambda_j - a) x}$ we thus have $w_a(x) = e^{-x^2/2} \cdot h_a(x)$ and $h_a(x) \in \mathcal{E}_s$.
Applying Lemma \ref{lem:unif} with $[a,b] = [-1,1]$ and $x= 0$ gives:
\begin{align*}
\frac{|h_a(0)|^2}{\norm{h_a}_{[-1,1]}^2} \le s.
\end{align*} 
This gives 
\begin{align*}
\frac{|w_a(0)|^2}{\norm{w_a}_{2}^2}  \le \frac{|w_a(0)|^2}{\norm{w_a}_{[-1,1]}^2} \le s \cdot \frac{e^0}{e^{-1}} = e \cdot s.
\end{align*} 
Plugging in $a = x$ this gives:
\begin{align*}
\frac{|f(x)|^2 \cdot \frac{1}{\sqrt{\pi}} e^{-x^2}}{\norm{f}_{g}^2}= \frac{|w_x(0)|^2}{\norm{w_x}_{2}^2}  \le e\cdot s.
\end{align*}
where  we use that $\norm{f}_{g}^2 = \frac{1}{\sqrt{\pi}}\norm{w_a}_{2}^2$ for any $a$ due to the weighting $e^{-(x+a)^2/2}$. Thus, we have $\tau_{\mathcal{E}_s,g}(x) \le e\cdot s$, completing the upper bound.

For the lower bound, let $w_t \in {\mathcal E}_s$ be defined by 
$$w_t(x) := f(x-t)e^{tx}\,, \qquad f(x) := \sum_{j=0}^{s-1}{e^{ijx}}\,.$$
We have
\begin{align}\label{eq:41}
|w_t(t)|^2 \cdot e^{-t^2} = s^2e^{2t^2} \cdot e^{-t^2} = s^2 e^{t^2}.
\end{align}
%Observe that 
%$$\int_0^\pi{|f(x)|^2 \, dx} = \pi s\,,$$
%and hence 
Additionally,
\begin{align*}%\label{eq:42}
 \int_{{t \in \R}}{|w_t(x)|^2 e^{-x^2} \, dx} & = \int_{{x \in \R}}{|f(x-t)|^2e^{2tx} e^{-x^2} \, dx} \\
&= e^{t^2} \int_{{x \in \R}}{|f(x-t)|^2 e^{-(x-t)^2} \, dx}\\
&= e^{t^2} \int_{u \in \R}{|f(u)|^2e^{-u^2} \, du}. 
\end{align*}
Since $f$ is a sum of complex exponentials with integer frequencies with period $2\pi$, we can bound:
\begin{align}\label{eq:42}
 \int_{{t \in \R}}{|w_t(x)|^2 e^{-x^2} \, dx} &\le e^{t^2} \int_{u \in \R}{|f(u)|^2e^{-u^2} \, du} \nonumber\\
 & \leq e^{t^2} \left( \int_0^\pi{|f(u)|^2e^{-u^2} \, du} \right) \cdot  \left( 2 + 2\sum_{k=1}^\infty{e^{-(k\pi)^2}} \right)\nonumber\\
& \leq \, e^{t^2} \cdot 3\pi s,
\end{align} 
where the last bound follows from the fact that $\int_0^\pi{|f(x)|^2 \, dx} = \pi s.$
Combining \eqref{eq:41} and \eqref{eq:42} we obtain the lower bound of the theorem. 

\end{proof}

%\Cam{Corresponds to Thm 2.3 of Tamas write up but slightly tighter bound on $x$ and simplified expression for leverage score bound.}
\begin{claim}[Gaussian Leverage Bound -- Large $x$]\label{clm:large} Letting $g(x) = \frac{1}{\sqrt{\pi}} e^{-x^2}$, when
$|x| \ge 6 \sqrt{s}$,
\begin{align*}
\tau_{s,g}(x) \le e^{-x^2/2}.
\end{align*}
\end{claim}
\begin{proof}
Applying Lemma \ref{lem:growth} with $a = 0$ and $d =  x/2$ gives that for any $f \in \mathcal{T}_s$:
\begin{align*}
\frac{|f(x)|^2}{ \norm{f}_{[0,x/2]}^2} \le \left (4e \right)^{2s}.
\end{align*}
This gives in turn that:
\begin{align}\label{eq:this}
\tau_{s,g}(x) \le \frac{e^{-x^2} \cdot (4e)^{2s}}{e^{-(x/2)^2} } \le e^{-3/4 \cdot x^2 + 6s}.
\end{align}
When $|x| \ge 6 \sqrt{s}$, $6s \le \frac{x^2}{6}$ and so \eqref{eq:this} gives $\tau_{s,g}(x) \le e^{(-3/4 + 1/6) \cdot x^2} \le e^{-x^2/2} $, completing the claim.
\end{proof}

We can prove Theorem \ref{thm:gaussian} directly from Claims \ref{clm:small} and \ref{clm:large}.
\begin{proof}[Proof of Theorem \ref{thm:gaussian}]
For the Gaussian density $g(x) = \frac{1}{\sigma\sqrt{2\pi}} \cdot e^{-x^2/(2\sigma^2)}$: 
\begin{align*}
 \tau_{s,g}(x) = \sup_{f \in \mathcal{T}_s} \frac{|f(x)|^2 \cdot e^{-x^2/(2\sigma^2)}}{\int_{-\infty}^{\infty} | f(y)|^2 e^{-y^2/(2\sigma^2)}\ dy}
\end{align*}
Let $\bar g(x) = \frac{1}{\sqrt{\pi}} e^{-x^2}$ be the Gaussian density with variance $1/2$. For any $f \in \mathcal{T}_s$, let $f_\sigma = f(\sqrt{2} \sigma \cdot x)$. Note that $f_\sigma \in \mathcal{T}_s$. We have:
\begin{align*}
\frac{|f(x)|^2 \cdot e^{-x^2/(2\sigma^2)}}{\int_{-\infty}^{\infty} | f(y)|^2 e^{-y^2/(2\sigma^2)}\ dy} = \frac{|f_\sigma(x/(\sqrt{2} \sigma))|^2 \cdot e^{-x^2/(2\sigma^2)}}{\int_{-\infty}^{\infty} | f_\sigma(y/(\sqrt{2}\sigma) )|^2 e^{-y^2/(2\sigma^2)}\ dy} = \frac{|f_\sigma(x/(\sqrt{2}\sigma)|^2 \cdot \bar g(x/(\sqrt{2}\sigma))}{\sqrt{2} \sigma \cdot \int_{-\infty}^{\infty} | f_\sigma(y)|^2 \cdot \bar g(y) \ dy}.
\end{align*}
Thus, $\tau_{s,g}(x) = \frac{1}{\sqrt{2} \sigma} \cdot \tau_{s,\bar g}(x/(\sqrt{2}\sigma))$. 
By Claims  \ref{clm:small} and \ref{clm:large}, if we define:
\begin{align*}
\bar \tau_{s,g}(x) = \begin{cases} \frac{1}{\sqrt{2} \sigma} \cdot e^{-x^2/(4\sigma^2)} \text{ for } |x| \ge 6 \sqrt{2} \sigma \cdot \sqrt{s}\\
 \frac{1}{\sqrt{2} \sigma} \cdot e \cdot s \text{ for } |x| \le 6\sqrt{2} \sigma \cdot \sqrt{s}
\end{cases} 
\end{align*}
we have 
\begin{align*}
\tau_{s,g}(x) = \frac{1}{\sqrt{2} \sigma} \cdot \tau_{s,\bar g}(x/(\sqrt{2}\sigma)) \le \bar \tau_{s,g} (x).
\end{align*}
Further,
\begin{align*}
\int_{\infty}^\infty \bar \tau_{s,g} (x) \ dx = \int_{-6 \sqrt{2} \sigma \cdot \sqrt{s}}^{6\sqrt{2} \sigma \cdot \sqrt{s}} \frac{es}{\sqrt{2} \sigma} \ dx + \frac{2}{\sigma} \int_{6\sqrt{2} \sigma \cdot \sqrt{s}}^{\infty} e^{-x^2/(4\sigma^2)} \ dx \le 12e s^{3/2} + 1,
\end{align*}
which completes the theorem.
\end{proof}

\subsection{Bounds for the Laplace density}\label{sec:exp}

We now give bounds for the Laplace density, again focusing on the unit width case and then proving Theorem \ref{thm:exp} via a simple scaling argument. Again, our bound is split into two components: a uniform bound for all $x$ and an improved bound for $x$ with large enough magnitude.

%\Cam{Equivalent to Thm. 2.2 of Tamas write up, but with a slightly better constant.}
\begin{claim}[Laplace Leverage Bound -- Universal]\label{clm:expSmall}
 Letting $z(x) = \frac{1}{2} e^{-|x|}$, for all $x \in \R$
 \begin{align*}
 \tau_{\mathcal{E}_s,z}(x) \le \frac{e^2 \cdot s}{1+|x|}.
 \end{align*}
 As a consequence, $\tau_{s,g}(x) \le \frac{e^2 \cdot s}{1+|x|}$.
\end{claim}
\begin{proof}
Assume that $x$ is nonnegative. The same bound holds for negative $x$, since for any $f \in \mathcal{E}_s$, letting $f'(x) = f(-x)$, $f' \in \mathcal{E}_s$ as well.
For any $f \in \mathcal{T}_s$ define the weighted function $w(x) = f(x) \cdot \frac{1}{\sqrt{2}} e^{-x/2}$. We can see that $w(x) \in \mathcal{E}_s$ as defined in \eqref{eq:calE} by writing:
\begin{align*}
w(x) &= \frac{1}{\sqrt{2}} \sum_{j=1}^s a_j e^{ i \lambda_j x } e^{-x/2} = \frac{1}{\sqrt{2}}\sum_{j=1}^s a_j e^{ (i \lambda_j - 1/2) x}.
\end{align*}
We define the `correctly' weighted function $h(x) = f(x) \cdot \frac{1}{\sqrt{2}} e^{-|x|/2}$. Note that for any $y \in [-1,0]$, we have $h(y) \ge e^{-1} \cdot w(y)$.  Thus, we have:
\begin{align*}
\frac{|f(x)|^2 \cdot \frac{1}{2} e^{-|x|}}{\norm{f}_{z}^2} = \frac{|h(x)|^2}{\norm{h}_{2}^2} \le 
\frac{|h(x)|^2}{\norm{h}_{[-1,2x+1]}^2} \le e^2 \cdot \frac{|w(x)|^2}{\norm{w}_{[-1,2x+1]}^2}.
\end{align*}
Applying Lemma \ref{lem:unif} with $[a,b] = [1,2x+1]$ then gives:
\begin{align*}
\frac{|f(x)|^2 \cdot \frac{1}{2} e^{-|x|}}{\norm{f}_{z}^2} \le e^2 \cdot \frac{|w(x)|^2}{\norm{w}_{[-1,2x+1]}^2} \le \frac{e^2 \cdot s}{1+x},
\end{align*} 
completing the claim.
%where we used that for any $y \ge 0$, $|f(y)|^2 \cdot e^{-|y|} = |w(y)|^2$.
%For $x \ge 1$ we have $\frac{2k}{x} \le \frac{4k}{1+|x|}$, giving the desired bound. For $x \le 1$, 
\end{proof}

\begin{claim}[Laplace Leverage Bound -- Large $x$]\label{clm:expLarge}
 Letting $z(x) = \frac{1}{2} e^{-|x|}$, when $|x| > 18s$,
 \begin{align*}
 \tau_{s,z}(x) \le e^{-|x|/6}.
 \end{align*}
\end{claim}
\begin{proof}
The proof is close to that of Claim \ref{clm:large} for the Gaussian density. As in Claim \ref{clm:expSmall},
assume without loss of generality that $x$ is nonnegative, so $x > 12s$.% The same bound holds for negative $x < -12s$, since for any $f \in \mathcal{T}_k$, letting $f'(x) = f(-x)$, $f' \in \mathcal{T}_k$.
Applying Lemma \ref{lem:growth} with $a = 0$ and $d =  x/2$ gives that for any $f \in \mathcal{T}_s$:
\begin{align*}
\frac{|f(x)|^2}{ \norm{f}_{[0,x/2]}^2} \le \left (4e \right)^{2s}.
\end{align*}
This gives:
\begin{align}\label{eq:this}
\tau_{s,z}(x) \le \frac{e^{-x} \cdot (4e)^{2s}}{e^{-x/2} } \le e^{-x/2 + 6s}.
\end{align}
When $x \ge 18 s$, $6s \le \frac{x}{3}$ and so \eqref{eq:this} gives $\tau_{s,g}(x) \le e^{(-1/2 + 1/3) \cdot x^2} \le e^{-x/6} $, completing the claim.
\end{proof}
We can prove Theorem \ref{thm:exp} directly from Claims \ref{clm:expLarge} and \ref{clm:expSmall}.
\begin{proof}[Proof of Theorem \ref{thm:exp}]
As in the proof of Theorem \ref{thm:exp}, we can observe that for the Laplace density $z(x) = \frac{1}{\sqrt{2}\sigma} \cdot e^{-|x|\sqrt{2}/\sigma}$, if we let $\bar z(x) = \frac{1}{2} e^{-|x|}$ be the density with variance $2$, we have:
$\tau_{s,z}(x) = \frac{\sqrt{2}}{ \sigma} \cdot \tau_{s,\bar z}(x\sqrt{2}/\sigma)$. 
By Claims \ref{clm:expSmall} and \ref{clm:expLarge}, if we define:
\begin{align*}
\bar \tau_{s,z}(x) = \begin{cases} \frac{\sqrt{2}}{\sigma} \cdot e^{-|x|\sqrt{2} /(6 \sigma)} \text{ for } |x| \ge 9\sqrt{2} \sigma \cdot{s}\\
\frac{\sqrt{2}}{\sigma} \cdot \frac{e^2 \cdot s}{1+ |x|\sqrt{2}/\sigma} \text{ for } |x| \le 9\sqrt{2} \sigma \cdot{s}
\end{cases} 
\end{align*}
we have 
\begin{align*}
\tau_{s,z}(x) = \frac{1}{\sqrt{2} \sigma} \cdot \tau_{s,\bar z}(x\sqrt{2}/\sigma) \le \bar \tau_{s,z} (x).
\end{align*}
Further,
\begin{align*}
\int_{\infty}^\infty \bar \tau_{s,z} (x) \ dx &= \frac{2\sqrt{2} e^2}{\sigma} \cdot \int_{0}^{9\sqrt{2} \sigma {s}}  \frac{s}{1+|x|\sqrt{2}/\sigma} \ dx + \frac{2\sqrt{2} }{\sigma} \int_{9\sqrt{2}  \sigma \cdot{s}}^{\infty} e^{-x\sqrt{2}/(6 \sigma)} \ dx\\
&=  2e^2 s \cdot \int_{0}^{18 {s}}  \frac{1}{1+x} \ dx + 2 \int_{18{s}}^{\infty} e^{-x/6} \ dx\\
&\le 2e^2 s \cdot \ln(18s+1) + 1,
\end{align*}
which completes the theorem.
\end{proof}

\subsection{Gap-based bounds}

Finally, we show how to obtain tighter bounds for the Gaussian density when considering functions in $\mathcal T_{s,\gamma}$, whose frequencies have minimum gap $\gamma > 0$ (see \eqref{eq:tsGap}). We show:
\begin{claim}\label{clm:gaussianGap}
Letting $g(x) = \frac{1}{\sqrt{\pi}} e^{-x^2}$, for all $x \in \R$:
\begin{align*}
\tau_{s,\gamma,g}(x) \le \left( \frac{\gamma}{6} \, e^{4\pi^2/\gamma^2} \right) \cdot  se^{-x^2}.
\end{align*}
%$$|f(t)|^2e^{-t^2} \leq \left( \frac{\gamma}{6} \, e^{4\pi^2/\gamma^2} \right) ne^{-t^2} \left( \int_{\R}{|f(x)|^2e^{-x^2} \, dx} \right)\,, \qquad t \in {\R}\,,$$
%for all exponential sums $f$ of the form 
%$$f(x) := \sum_{k=1}^n{a_ke^{i\lambda_k x}}\,, \qquad x \in {\R}\,.$$
\end{claim}
The above leverage score upper bound is just a scaling of the data density $e^{-x^2}$. For $\gamma = \Omega(1)$, it integrates to $O(s)$, within a constant factor of the lower bound $\int_{x \in \R}  \tau_{s,\gamma,g}(x)\, dx \ge s$ given by restricting $\mathcal{T}_{s,\gamma}$ to just just a single fixed set of frequencies. 

 Claim \ref{clm:gaussianGap} can be turned into a leverage score bound  for the Gaussian density of any width, using the simple scaling argument of Theorem \ref{thm:gaussian} giving:
\begin{theorem}[Gaussian Leverage Bound -- Gap Condition]\label{clm:gaussianGap2} Consider the Gaussian density with variance $\sigma^2 > 0$, $g(x) = \frac{1}{\sigma \sqrt{2\pi}} e^{-x^2/(2\sigma^2)}$, and let:
\begin{align*}
\bar \tau_{s,\gamma,g}(x) \le  \left( \frac{\gamma}{6} \, e^{4\pi^2/\gamma^2} \cdot  s \right) \cdot \left ( \frac{1}{\sqrt{2} \sigma} \cdot e^{-x^2/(2\sigma^2)} \right ).
\end{align*}
We have $\tau_{s,\gamma,g}(x) \le \bar \tau_{s,\gamma,g}(x)$ for all $x \in \R$ and $\int_{-\infty}^\infty \bar \tau_{s,\gamma,g}(x) dx = \frac{\gamma}{6} \, e^{4\pi^2/\gamma^2} \cdot  \sqrt{\pi} s$.
\end{theorem}

\begin{proof}[Proof of Claim \ref{clm:gaussianGap}]
Consider $f \in \mathcal T_{s,\gamma}$ with $f(x) := \sum_{j=1}^s{a_je^{i\lambda_j x}},$ and $\min_{j,k}  |\lambda_k-\lambda_j| \ge \gamma > 0$. Combining the Cauchy-Schwarz inequality with Ingham's inequality (Lemma \ref{lem:gap}), we obtain
\begin{align*}
|f(x)|^2 = \left| \sum_{j=1}^s{a_j e^{i\lambda_j x}} \right|^2 &\leq \left( \sum_{j=1}^s{\left| e^{i\lambda_j t} \right|^2} \right) \left( \sum_{j=1}^s{|a_j|^2} \right) \\  
& \leq   \frac{\gamma s}{6}  \int_{-2\pi/\gamma}^{2\pi/\gamma}{\left| \sum_{j=1}^s{a_j e^{i\lambda_j x}} \right|^2 \, dx} \\ 
& \leq \left( \frac{\gamma}{6} \, e^{4\pi^2/\gamma^2} \right) s \int_{\R}{|f(x)|^2e^{-x^2} \, dx}.
\end{align*}
Hence
$$g(x) \cdot |f(x)|^2 \leq \left( \frac{\gamma}{6} \, e^{4\pi^2/\gamma^2} \right) s e^{-x^2} \cdot \norm{f}_g^2,$$
completing the claim.
\end{proof}

\bibliographystyle{alpha}
\bibliography{gaussianScores}	

\clearpage
\appendix

\section{Kernel Approximation -- Omitted Proofs}
\label{app:kernels}
As discussed in Section \ref{sec:oblivious}, our result on oblivious kernel embedding (Theorem \ref{thm:high_kernel}) is based on a result from \cite{AvronKapralovMusco:2017}, which shows that strong kernel approximations can be obtained via random Fourier features methods which sample by the kernel ridge leverage scores of Definition \ref{def:kl}:

\begin{theorem}[Kernel Embedding via Leverage Score Sampling, \cite{AvronKapralovMusco:2017}]\label{thm:rff}
%	Consider the setting of Definition \ref{def:kl}. 
	Let $s_\lambda$ denote the $\lambda$-statistical dimension of $\bv{K}$.
	Given a function $\bar \tau_{\lambda,\bv{K}} (\eta)$ with:
	$$\bar \tau_{\lambda,\bv{K}} (\eta) \ge \tau_{\lambda,\bv{K}}(\eta)\text{ for all }\eta \in \R\text{ and }T \eqdef \int_{\eta \in \R} \bar \tau_{\lambda,\bv{K}} (\eta)  d\eta,$$ if we apply modified RFF sampling (Definition \ref{def:rff}) with density $q(\eta)= \frac{\bar \tau_{\lambda,\bv{K}} (\eta)}{T}$ and sample size $m = \frac{3 T \ln(16 s_\lambda/\delta)}{\epsilon^2}$, then with probability $\ge 1-\delta$, $\bv{G}^* \bv{G}$ is an $(\epsilon,\lambda)$-spectral approximation of $\bv{K}$.
\end{theorem}

\subsection{Kernel leverage score bounds via Fourier sparse approximation}
To make use of Theorem \ref{thm:rff}, we need access to an upper bound $\bar \tau_{\lambda,\bv{K}}(\eta)$ on the kernel ridge leverage scores. We remark that $\int_{\eta \in \R} \tau_{\lambda,\bv K}(\eta) d\eta = \tr(\bv{K}+\lambda \bv I)^{-1} \bv{K}) = s_\lambda$ \cite{AvronKapralovMusco:2017}. Thus, if $\bar \tau_{\lambda,\bv{K}}(\eta)$ is a tight bound, Theorem \ref{thm:rff} yields an embedding dimension $m = \tilde O(s_\lambda/\epsilon^2)$. 
Our goal is to obtain a nearly tight bound by reducing the problem of bounding $\tau_{\lambda, \bv{K}}$ to that of bounding the Fourier sparse leverage score under the density $p_k$ given by  the kernel Fourier transform. We prove:
\begin{reptheorem}{thm:sparseReduction} Consider a positive definite, shift invariant kernel $k: \R \rightarrow \R$, any points $x_1,\ldots, x_n \in \R$ and the associated kernel matrix $\bv{K}$, with statistical dimension $s_\lambda$.  Let $s = 6\lceil s_\lambda\rceil +1$. Then:%For all $\eta \in \R$:
	%\vspace{-.5em}
	\begin{align*}
	\forall \eta \in \R, \quad \tau_{\lambda,\bv{K}}(\eta)  \le  (2 +6 s_\lambda) \cdot  \tau_{ s, p_k}(\eta).
	\end{align*}
\end{reptheorem}
As discussed in Section \ref{sec:oblivious}, we prove Theorem \ref{thm:sparseReduction} by first showing that any function $\bv{\Phi} \bv{w}$ in the span of our kernelized data points is well approximated by  via an $O(s_\lambda)$ sparse Fourier function.
%
%
%Fourier sparse function, and then bounding how much it can spike (i.e., the kernel ridge leverage score of Def. \ref{def:kl}) using the Fourier sparse leverage score bounds of Theorems \ref{thm:gaussian} and \ref{thm:exp}. This general technique of approximating kernel space functions with Fourier sparse functions was first introduced in \cite{AvronKapralovMusco:2019} in the context of active regression, and we believe it will be a powerful tool in understanding kernel approximation going forward.
%
This Fourier sparse approximation result is based on the well-known fact that any matrix with bounded statistical dimension can be well approximated via projection onto a small subset of rows or columns \cite{drineas2006subspace,guruswami2012optimal,boutsidis2014near}. In particular, we show via a simple reformulation of known results:
\begin{theorem}[Row Subset Selection]\label{thm:rss} Consider the setting of Theorem \ref{thm:sparseReduction}. For $t = 6 \cdot \lceil s_\lambda\rceil $, there exists a subset of $t$ indices ${i_1},\ldots,{i_t} \subseteq [n]$ and $\bv{Z} \in \R^{t \times n}$ % and a coefficient matrix $\bv{Z} \in \R^{s \times s}$ 
	such that, letting $\bs{\Phi}_t: \C^t \rightarrow L_2$ be the operator with $[\bs{\Phi}_t \bv w](\eta) = \sqrt{p_k(\eta)} \cdot \sum_{j=1}^t \bv w_j e^{-2\pi i \eta x_{i_j}} $ (i.e., the operator containing the $t$ columns of $\bs{\Phi}$ corresponding to the indices $i_1,\ldots,i_t$):
	\begin{align*}
	\tr(\bv{K} - \bv{Z}^T \bs{\Phi}_t^* \bs{\Phi}_t \bv{Z}) \le 3 \lambda s_\lambda\, \text{ and }\, \bv{Z}^T \bs{\Phi}_t^* \bs{\Phi}_t \bv{Z} \preceq \bv{K}.
	\end{align*}
\end{theorem}
\begin{proof}
	Let $\bv{B} \in \R^{n \times n}$ be any matrix squareroot of $\bv{K}$ with $\bv{B}^T \bv{B} = \bv{K}$. Since $\bv{B}^T \bv{B} = \bs{\Phi}^* \bs{\Phi}$ it suffices to prove the existence of a subset of indices ${i_1},\ldots,{i_t} \subseteq [n]$ and a matrix $\bv{Z} \in \R^{t \times n}$ such that, letting $\bv{B}_t$ contain the columns of $\bv{B}$ corresponding to those indices:
	\begin{align}\label{eq:bsuffices}
	\tr(\bv{K} - \bv{Z}^T \bv{B}_t^T \bv{B}_t \bv{Z}) \le 3 \lambda s_\lambda \text{ and } \bv{Z}^T \bv{B}_t^T \bv{B}_t \bv{Z} \preceq \bv{K}.
	\end{align}
	Let $\bv{Z} = \bv{B}_t^+ \bv{B}$. Letting $\bv{P}_t  = \bv{B}_t \bv{B}_t^+$ be the orthogonal projection matrix onto the columns of $\bv{B}_t$, we can see that $\bv{Z}^T \bv{B}_t^T \bv{B}_t \bv{Z} = \bv{B}^T \bv{P}^2_t \bv{B} = \bv{B}^T \bv{P}_t \bv{B} $. We first observe that for any $\bv x \in \R^n$:
	\begin{align*}
	\bv{x}^T \bv{Z}^T \bv{B}_t^T \bv{B}_t \bv{Z}\bv {x} = \norm{\bv{P}_t \bv{B}\bv{x}}_2^2 \le \norm{ \bv{B}\bv{x}}_2^2 =\bv{x}^T \bv{K}\bv {x},
	\end{align*}
	which proves that $\bv{Z}^T \bv{B}_t^T \bv{B}_t \bv{Z} \preceq \bv{K}$, giving the second part of \eqref{eq:bsuffices}. To prove the first part of \eqref{eq:bsuffices} we employ an optimal column-based matrix reconstruction result \cite{guruswami2012optimal}, Theorem 1.1, which shows that there  exists a set of $s = 6 \cdot \lceil s_\lambda\rceil$ indices such that:
	\begin{align}\label{eq:css}
	\norm{\bv{B} - \bv{B}_t \bv{Z}}_F^2 \le 1.5 \norm{\bv{B} - \bv{B}_{2\lceil s_\lambda \rceil}}_F^2,
	\end{align}
	where $\bv{B}_{ 2 \lceil s_\lambda\rceil }$ is the best rank-$2 \lceil s_\lambda\rceil$ approximation to $\bv{B}$ (given by projecting $\bv{B}$ onto its top $2 \lceil s_\lambda\rceil$ singular vectors). Since $\bv{B}_t \bv{Z}$ is the projection of $\bv{B}$ onto the column space of $\bv{B}_t$ we can write via the Pythagorean theorem:
	\begin{align*}
	\norm{\bv{B} - \bv{B}_t \bv{Z}}_F^2 = \norm{\bv{B}}_F^2 - \norm{\bv{B}_t \bv{Z}}_F^2 = \tr(\bv{B}^T \bv{B}) - \tr(\bv{Z}^T \bv{B}_t^T \bv{B}_t \bv{Z}) = \tr(\bv{B}^T \bv{B} - \bv{Z}^T \bv{B}_t^T \bv{B}_t \bv{Z}).
	\end{align*}
	Thus, in combination with \eqref{eq:css}, if we can show $\norm{\bv{B} - \bv{B}_{2 \lceil s_\lambda\rceil}}_F^2 \le 2 \lambda s_\lambda$, we will have 
	\begin{align*}
	\tr(\bv{B}^T \bv{B} - \bv{Z}^T \bv{B}_t^T \bv{B}_t \bv{Z}) \le 3 \lambda s_\lambda,
	\end{align*}
	yielding the first part of \eqref{eq:bsuffices} and the theorem.
	This bound follows from the fact that 
	$\norm{\bv{B} - \bv{B}_{\lceil 2s_\lambda\rceil}}_F^2 = \sum_{i=2 \lceil s_\lambda\rceil+1}^n \lambda_i(\bv{K}).$
	%and the claim:
	%\begin{claim} For any $\lambda > 0$, $\sum_{i= \lceil s_\lambda\rceil+1}^n \lambda_i(\bv{K}) \le 2 \lambda s_\lambda$.
	%\end{claim}
	%\begin{proof}
	We can apply the following claim, which quantifies the eigenvalue decay of a matrix in terms of its statistical dimension:
	\begin{claim}\label{clm:skBound} For any positive semidefinite $\bv{K} \in \R^{n \times n}$ with statistical dimension $s_\lambda$:
		$$%\lambda_{2 \lceil s_\lambda\rceil+1}(\bv{K}) < \lambda \text{ and }  
		\sum_{i=2 \lceil s_\lambda\rceil+1}^n \lambda_i(\bv{K}) \le 2 \lambda s_\lambda.$$
	\end{claim}
	\begin{proof}
		Let $I_\lambda$ be the number of eigenvalues of $\bv{K}$ that are $\ge \lambda$. We have:
		\begin{align*}
		s_\lambda = \sum_{i=1}^n \frac{\lambda_i(\bv{K})}{\lambda_i(\bv{K}) + \lambda} &= \sum_{i=1}^{I_\lambda} \frac{\lambda_i(\bv{K})}{\lambda_i(\bv{K}) + \lambda}  + \sum_{i=I_\lambda+1}^{n} \frac{\lambda_i(\bv{K})}{\lambda_i(\bv{K}) + \lambda} \\
		&\ge \frac{1}{2} \cdot I_\lambda + \frac{1}{2 \lambda } \sum_{i=I_\lambda+1}^{n} \lambda_i(\bv{K}),
		\end{align*}
		where the second line follows from that fact that $\lambda_i(\bv{K}) \ge \lambda$ for $i \le I_\lambda$ and $\lambda_i(\bv{K}) < \lambda$ for $i > I_\lambda$
		%\end{proof}
		Rearranging we have $2 \lceil s_\lambda\rceil \ge 2 s_\lambda \ge I_\lambda$ and $2 s_\lambda \ge \frac{1}{\lambda } \sum_{i=I_\lambda+1}^{n} \lambda_i(\bv{K})$, and in turn:
		\begin{align*}
		2 s_\lambda \ge  \frac{1}{\lambda} \sum_{i=2 \lceil s_\lambda \rceil +1}^{n} \lambda_i(\bv{K}) \implies 2 \lambda s_\lambda \ge \sum_{i=2 \lceil s_\lambda \rceil +1}^{n} \lambda_i(\bv{K}).
		\end{align*}
	\end{proof}
	Claim \ref{clm:skBound} directly gives that  $\norm{\bv{B} - \bv{B}_{2 \lceil s_\lambda\rceil}}_F^2 = \sum_{i=2 \lceil s_\lambda\rceil+1}^n \lambda_i(\bv{K}) \le 2 \lambda s_\lambda$, completing the proof of Theorem \ref{thm:rss}.
\end{proof} 

\begin{proof}[Proof of Theorem \ref{thm:sparseReduction}]
	Applying Theorem \ref{thm:rss} we can bound the kernel leverage score by breaking the function $\bs{\Phi} \bv{w}$ into its projection onto $\bs{\Phi}_t$, which after a change of density is a $t = \lceil 6s_\lambda \rceil$-sparse Fourier function in $\mathcal{T}_t$, and the residual.
	\begin{align}\label{eq:split1}
	\tau_{\lambda,\bv{K}}(\eta) = \sup_{\bv w \in \C^n, \bv{w} \neq 0} \frac{|[\bs{\Phi} \bv w](\eta)|^2}{\|\bs{\Phi}\bv w\|_{2}^2 + \lambda \|\bv w\|_2^2} &\le \frac{2 |[\bs{\Phi}_t \bv{Z } \bv w](\eta)|^2}{\|\bs{\Phi}\bv w\|_{2}^2 + \lambda \|\bv w\|_2^2} + \frac{2|[\bs{\Phi} \bv w](\eta) - [\bs{\Phi}_t \bv{Z } \bv w](\eta)|^2}{\|\bs{\Phi}\bv w\|_{2}^2 + \lambda \|\bv w\|_2^2}\nonumber\\
	&\le \frac{2 |[\bs{\Phi}_t \bv{Z } \bv w](\eta)|^2}{\|\bs{\Phi}\bv w\|_{2}^2} + \frac{2|[\bs{\Phi}\bv w](\eta) - [\bs{\Phi}_t \bv{Z } \bv w](\eta)|^2}{ \lambda \|\bv w\|_2^2}.
	\end{align}
	Since by Theorem \ref{thm:rss}, $\bv{Z}^T \bs{\Phi}_t^* \bs{\Phi}_t \bv{Z} \preceq \bv{K}$ we have 
	$$\|\bs{\Phi}\bv w\|_{2}^2 = \bv w^T \bv{K} \bv w \ge \bv w^T \bv{Z}^T \bs{\Phi}_t^* \bs{\Phi}_t \bv{Z} \bv w = \|\bs{\Phi}_t \bv{Z} \bv w\|_{2}^2,$$ which combined with \eqref{eq:split1} gives: 
	\begin{align}\label{eq:inter1}
	\tau_{\lambda,\bv{K}}(\eta) &\le \frac{2 |[\bs{\Phi}_t \bv{Z } \bv w](\eta)|^2}{\|\bs{\Phi}_t \bv{Z} \bv w\|_{2}^2} + \frac{2|[\bs{\Phi} \bv w](\eta) - [\bs{\Phi}_t \bv{Z } \bv w](\eta)|^2}{ \lambda \|\bv w\|_2^2}\nonumber\\
	&\le 2 \tau_{ t, p_k}(\eta) + \frac{2 |[\bs{\Phi} \bv w](\eta) - [\bs{\Phi}_t \bv{Z } \bv w](\eta)|^2}{ \lambda \|\bv w\|_2^2}.
	\end{align}
	The second bound follows from the fact that $\frac{[\bs{\Phi}_t \bv{Z }\bv w](\eta)}{\sqrt{p_k(\eta)}} \in \mathcal{T}_{t}$. 
	It remains to bound the second term of \eqref{eq:inter1}.  Let $\bv z(\eta) \in \C^n$ be the vector with $\bv z(\eta)_j = \left [e^{-2 \pi i \eta x_j} - \sum_{k=1}^t \bv{Z}_{k,j} \cdot  e^{-2 \pi i \eta x_{i_k}}\right ]  \cdot \sqrt{p_k(\eta)}$. Then we can bound via Cauchy-Schwarz:
	\begin{align}\label{eq:cs}
	\frac{|[\bs{\Phi} \bv w](\eta) - [\bs{\Phi}_s \bv{Z }\bv  w](\eta)|^2}{ \lambda \|\bv w\|_2^2} = \frac{|\bv z(\eta)^* \bv w|^2}{\lambda \norm{\bv w}_2^2} \le \frac{\norm{\bv z(\eta)}_2^2}{\lambda}.
	\end{align}
	We bound $\norm{\bv z(\eta)}_2^2$ as:
	\begin{claim}\label{clm:zBound} Let $\bv z(\eta) \in \C^n$ be as defined above. $\norm{\bv z(\eta)}_2^2 \le \tau_{ t +1, p_k}(\eta) \cdot 3 \lambda s_\lambda$. 
		\
	\end{claim}
	Combining Claim \ref{clm:zBound} with \eqref{eq:inter1} and \eqref{eq:cs} yields:
	\begin{align*}
	\tau_{\lambda,\bv{K}}(\eta) &\le 2 \tau_{ t , p_k}(\eta) + 6\tau_{ t +1, p_k}(\eta) \cdot  s_\lambda \le (2 +6 s_\lambda) \cdot  \tau_{ 6t +1, p_k}(\eta),
	\end{align*}
	which completes the theorem after recalling that we set $t = \lceil s_\lambda\rceil$ in Theorem \ref{thm:rss}.
\end{proof}

\begin{proof}[Proof of Claim \ref{clm:zBound}]
	Consider the function $g_j(\eta) = \bv z(\eta)_j$ and $g(\eta) = \sum_{j=1}^n |g_j(\eta)|^2  = \norm{\bv z(\eta)}_2^2$. %Note that 
	\begin{align*}
	g_j(\eta) = \left [e^{-2 \pi i \eta x_j} - \sum_{k=1}^s \bv{Z}(k,j) \cdot  e^{-2 \pi i \eta x_{i_k}}\right ]  \cdot \sqrt{p_k(\eta)}
	\end{align*}
	and thus, $h(\eta) \eqdef \frac{g_j(\eta)}{\sqrt{p_k(\eta)}} \in \mathcal T_{t+1}$ and so:
	\begin{align*}
	\frac{|g_j(\eta)|^2}{\norm{g_j}_{2}^2} =\frac{ p_k(\eta)\cdot |h(\eta)|^2}{\norm{h}_{p_k}^2} \le \tau_{t+1,p_k}(\eta).
	\end{align*}
	This gives:
	\begin{align*}
	\norm{\bv z(\eta)}_2^2 = \sum_{j=1}^n |g_j(\eta)|^2 &\le \tau_{t+1,p_k}(\eta) \cdot \sum_{j=1}^n \norm{g_j}_{2}^2\\
	&= \tau_{t+1,p_k}(\eta) \cdot  \tr(\bv{K} - \bv{Z}^T \bs{\Phi}_s^* \bs{\Phi}_s \bv{Z})\\
	& \le \tau_{t+1,p_k}(\eta) \cdot 3 s_\lambda,
	\end{align*}
	where the last bound follows from Theorem \ref{thm:rss}.
\end{proof}

\subsection{Oblivious kernel embedding via keverage score-based RFF}

We finally combine our Fourier sparse leverage score bounds of Theorems \ref{thm:gaussian} and \ref{thm:exp} with the kernel ridge leverage score bound of Theorem \ref{thm:sparseReduction}  and the leverage score sampling result of Theorem \ref{thm:rff} to give oblivious kernel embedding results for the kernels corresponding to the Fourier transforms of the Gaussian and Laplace densities -- i.e., the Gaussian and Cauchy (rational quadratic) kernel.
\begin{corollary}[Modified RFF Embedding -- Gaussian Kernel]\label{cor:gaussian}
	Consider any set of points $x_1,\ldots, x_n \in \R$ and the associated Gaussian kernel matrix $\bv{K} \in \R^{n \times n}$ with $\bv{K}_{i,j} = e^{-(x_i-x_j)^2/(2\sigma^2)}$. Let $s_\lambda$ be the $\lambda$-statistical dimension of $\bv{K}$, $s = 6 \lceil s_\lambda \rceil +1$, and $q(\eta)$ be the density proportional to:
	\begin{align*}
	q(\eta) \propto \begin{cases} e^{-\eta ^2 \cdot \pi^2 \cdot \sigma^2} \text{ for } |\eta| \ge \frac{3\sqrt{2}}{\sigma \pi } \cdot \sqrt{s}\\
	e\cdot s\text{ for } |\eta| \le \frac{3\sqrt{2}}{\sigma \pi } \cdot  \sqrt{s}.
	\end{cases} 
	\end{align*}
	The modified RFF embedding (Def. \ref{def:rff}) with density $q(\eta)$ and sample size $m = O\left (\frac{s_\lambda^{5/2} \cdot \log(s_\lambda/\delta)}{\epsilon^2}\right )$, satisfies $\bv{G}^*\bv{G}$ is an $(\epsilon,\lambda)$-spectral approximation of $\bv{K}$ with probability $\ge 1-\delta$.
	%\begin{align*}
	%(1-\epsilon)(\bv{K} + \lambda\bv{I}) \preceq \bs{\tilde \Phi}^* \bs{\tilde \Phi}+ \lambda\bv{I} \preceq (1+\epsilon)(\bv{K} + \lambda\bv{I}). 
	%\end{align*}
	The embedding $\bv{g}(x_i) \in \C^m$, can be constructed obliviously in $O(m)$ time.
\end{corollary}
\begin{proof}
	For the Gaussian kernel with width $\sigma$, the Fourier transform density is also Gaussian with variance $\frac{1}{4 \pi^2 \sigma^2}$:
	\begin{align*}
	p_k(\eta) = \int_{t \in \R} e^{2 \pi i \eta t} e^{-\frac{t^2}{2 \sigma^2}} dt = \sigma  \sqrt{2 \pi} \cdot  e^{-2  \sigma^2 \pi^2 \eta^2}.
	\end{align*}
	Applying Theorem \ref{thm:sparseReduction} we have: $\tau_{\lambda,\bv{K}}(\eta) \le (2 +6s_\lambda) \cdot \tau_{s,p_k}(\eta)$ 
	for $s = 6 \lceil s_\lambda \rceil +1$. In turn, applying Theorem \ref{thm:gaussian} gives $\tau_{\lambda,\bv{K}}(\eta) \le \bar \tau_{\lambda,\bv{K}}(\eta)$ %(2 + 3s_\lambda) \cdot \bar \tau(\eta)$
	where:
	\begin{align*}
	\bar \tau_{\lambda,\bv{K}}(\eta) = \begin{cases} (2 + 6s_\lambda) \cdot \pi \sqrt{2}\cdot \sigma \cdot e^{-\eta ^2 \cdot \pi^2 \cdot \sigma^2} \text{ for } |\eta| \ge \frac{3\sqrt{2}}{\sigma \pi } \cdot \sqrt{s}\\
	(2 + 6s_\lambda)  \cdot \pi \sqrt{2} e \cdot \sigma\cdot s\text{ for } |\eta| \le \frac{3\sqrt{2}}{\sigma \pi } \cdot  \sqrt{s}.
	\end{cases} 
	\end{align*}
	Thus, by Theorem \ref{thm:rff}, if we let $q(\eta)$ be the density proportional to $\bar \tau_{\lambda,\bv{K}}(\eta)$, a random Fourier features approximation satisfies the guarantee of the Theorem with sample size $m$ given by:
	\begin{align*}
	m = O \left (\frac{\int_{\eta \in \R} \bar \tau_{\lambda,\bv{K}}(\eta) d\eta\, \cdot \log(s_\lambda/\delta)}{\epsilon^2} \right ) = O\left (\frac{s_\lambda^{5/2} \cdot \log(s_\lambda/\delta)}{\epsilon^2}\right ),
	\end{align*}
	since by Theorem \ref{thm:gaussian}, $\int_{\eta \in \R} \bar \tau_{\lambda,\bv{K}}(\eta) d\eta = (2 + 6s_\lambda) \cdot O(s^{3/2}) = O(s_\lambda^{5/2})$.
	
	Finally, we observe that $q(\eta)$ is just a mixture of a Gaussian density with a uniform density, and hence can be sampled from in $O(1)$ time. Thus each embedding $\bv{g}(x_i) \in \C^m$ can be constructed obliviously in $O(m)$ time.
\end{proof} 
%%%%%%%%%%
We give a very similar result for the Cauchy (also known as rational quadratic) kernel using our Laplacian distribution leverage score bound of Theorem \ref{thm:exp}.
\begin{corollary}[Modified RFF Embedding -- Cauchy Kernel]\label{cor:cauchy}
	Consider any set of points $x_1,\ldots, x_n \in \R$ and the associated Cauchy kernel matrix $\bv{K} \in \R^{n \times n}$ with $\bv{K}_{i,j} = \frac{1}{1+(x_i-x_j)^2/\sigma^2}$. Let $s_\lambda$ be the $\lambda$-statistical dimension of $\bv{K}$, $s = 6 \lceil s_\lambda \rceil +1$, and $q(\eta)$ be the density proportional to:
	\begin{align*}
	q(\eta) \propto \begin{cases} e^{-|\eta| \cdot \sigma \pi /3} \text{ for } |\eta| \ge \frac{9s }{\sigma \pi}\\
	\frac{e^2 s}{1+ |\eta|\cdot 2 \sigma \pi} \text{ for } |\eta| \le  \frac{9s}{\sigma \pi}.
	\end{cases} 
	\end{align*}
	The modified RFF embedding (Def. \ref{def:rff}) with density $q(\eta)$ and sample size $m =  O\left (\frac{s_\lambda^{2} \log(s_\lambda) \cdot \log(s_\lambda/\delta)}{\epsilon^2}\right )$ satisfies $\bv{G}^*\bv{G}$ is an $(\epsilon,\lambda)$-spectral approximation of $\bv{K}$ with probability $\ge 1-\delta$.
	The embedding $\bv{g}(x_i) \in \C^m$, can be constructed obliviously in $O(m)$ time.
\end{corollary}
\begin{proof}
	For the Cauchy kernel with width $\sigma$, the Fourier transform density is a Laplace density:
	\begin{align*}
	p_k(\eta) = \int_{t \in \R} e^{2 \pi i \eta t} \frac{1}{1+(t/\sigma)^2} dt = \sigma \pi \cdot e^{-|\eta| \cdot 2\sigma \pi}.
	\end{align*}
	Applying Theorem \ref{thm:sparseReduction} we have: $\tau_{\lambda,\bv{K}}(\eta) \le (2 +6s_\lambda) \cdot \tau_{s,p_k}(\eta)$ 
	for $s = 6 \lceil s_\lambda \rceil +1$. In turn, applying Theorem \ref{thm:exp} gives $\tau_{\lambda,\bv{K}}(\eta) \le \bar \tau_{\lambda,\bv{K}}(\eta)$ %(2 + 3s_\lambda) \cdot \bar \tau(\eta)$
	where:
	\begin{align*}
	\bar \tau_{\lambda,\bv{K}}(\eta) = \begin{cases} (2+6s_\lambda)\cdot 2 \sigma \pi \cdot e^{-|\eta| \cdot \sigma \pi /3} \text{ for } |\eta| \ge \frac{9s }{\sigma \pi}\\
	(2+6s_\lambda) \cdot 2 \sigma \pi \cdot \frac{e^2 s}{1+ |\eta|\cdot 2 \sigma \pi} \text{ for } |\eta| \le  \frac{9s}{\sigma \pi}.
	\end{cases} 
	\end{align*}
	Thus, by Theorem \ref{thm:rff}, if we let $q(\eta)$ be the density proportional to $\bar \tau_{\lambda,\bv{K}}(\eta)$, a random Fourier features approximation satisfies the guarantee of the theorem with sample size $m$ given by:
	\begin{align*}
	m = O \left (\frac{\int_{\eta \in \R} \bar \tau_{\lambda,\bv{K}}(\eta) d\eta\, \cdot \log(s_\lambda/\delta)}{\epsilon^2} \right ) = O\left (\frac{s_\lambda^{2} \log(s_\lambda) \cdot \log(s_\lambda/\delta)}{\epsilon^2}\right ),
	\end{align*}
	%s
	%		\begin{align*}
	%	\bar \tau_{k,z}(x) = \begin{cases} 2 \sigma \pi \cdot e^{-|x| \cdot \sigma \pi /3} \text{ for } |x| \ge \frac{6}{\sigma \pi} s\\
	%	2 \sigma \pi \cdot \frac{2e^2 s}{1+ |x|\cdot 2 \sigma \pi} \text{ for } |x| \le  \frac{6}{\sigma \pi} s.
	%	\end{cases} 
	%	\end{align*}
	%sigma = 1/(\sqrt{2} sigma pi)
	%
	since by Theorem \ref{thm:exp}, $\int_{\eta \in \R} \bar \tau_{\lambda,\bv{K}}(\eta) d\eta = (2 + 6s_\lambda) \cdot O(s \log s) = O(s_\lambda^2 \log s_\lambda)$.
	
	Finally, observe that $q(\eta)$ is just a mixture of a Laplacian density with a density of the form $\frac{1}{1+|\eta| \cdot 2 \sigma \pi}$. Both can be sampled from in $O(1)$ time using, e.g., inverse transform sampling. Thus each embedding $\bv{g}(x_i)$ can be constructed obliviously in $O(m)$ time. \end{proof} 

\subsection{Final embedding via random projection}\label{sec:rp}

Corollaries \ref{cor:gaussian} and \ref{cor:cauchy} give oblivious embeddings into $\poly(s_\lambda)$ dimensions via leverage score-based RFF sampling. These oblivious embeddings can be further compressed via standard oblivious random projection time to give an oblivious embedding algorithm achieving the target dimension, linear in $s_\lambda$. Specifically we apply a stable rank approximate matrix multiplication result from \cite{cohen2016optimal}:

\begin{theorem}[Random Projection Spectral Approximation]\label{thm:rp} For any $\bv{Z} \in \R^{n \times s}$ and  $\bv{M} = \bv{ZZ}^T$ with $\lambda$-statistical dimension $s_\lambda$, if $\bs{\Pi} \in \R^{s \times m}$ has independent sub-Gaussian entries with variance $1/m$ for $m = O \left (\frac{s_\lambda + \log(1/\delta)}{\epsilon^2} \right )$, then with probability $\ge 1-\delta$, $\bv{Z} \bs{\Pi} \bs{\Pi}^T \bv{Z}^T$  is an $(\epsilon,\lambda)$-spectral approximation of $\bv{M}$.
	%$$(1-\epsilon)(\bv{M} + \lambda\bv{I}) \preceq \bv{Z} \bs{\Pi} \bs{\Pi}^T \bv{Z}^T + \lambda\bv{I} \preceq (1+\epsilon)(\bv{M} + \lambda\bv{I}). $$
\end{theorem}
A simple example of $\bs{\Pi}$ that satisfies the theorem is one with independent $\pm 1/\sqrt{m}$ entries. See \cite{cohen2016optimal} for more details on sketching matrices that may be used, including sparse ones.
\begin{proof}
	% Let $\bv{P}_{2\lceil s_{\lambda}\rceil} \in \R^{n \times n}$ be the projection onto $\bv{M}$'s top $2\lceil s_{\lambda}\rceil$ eigenvectors and $\bv{P}_{\setminus 2\lceil s_{\lambda}\rceil} = \bv{I} - \bv{P}_{2\lceil s_{\lambda}\rceil}$ be its complement. We can write:
	% $\bv{M} = \bv{Z} \bv{P}_{2\lceil s_{\lambda}\rceil} \bv{Z}^T + \bv{Z} \bv{P}_{\setminus 2\lceil s_{\lambda}\rceil} \bv{Z}^T$. Further we have:
	% \begin{align*}
	% \end{align*}
	Let $\bv{B} = (\bv M + \lambda \bv I)^{-1/2} \bv{Z}$. To prove the theorem it suffices to show that with probability  $\ge 1-\delta$, $\norm{\bv{B} \bs{\Pi} \bs{\Pi}^T \bv{B}^T - \bv{BB}^T}_2 \le \epsilon$ as this gives:
	\begin{align*}
	-\epsilon \bv{I} &\preceq \bv{B} \bs{\Pi} \bs{\Pi}^T \bv{B}^T - \bv{BB}^T  \preceq \epsilon \bv{I}\\
	-\epsilon (\bv M + \lambda \bv I) &\preceq \bv{Z} \bs{\Pi} \bs{\Pi}^T \bv{Z}^T - \bv{ZZ}^T \preceq \epsilon (\bv M + \lambda \bv I) \\
	\bv{M} -\epsilon(\bv{M} + \lambda\bv{I}) &\preceq \bv{Z} \bs{\Pi} \bs{\Pi}^T \bv{Z}^T \preceq \bv{M} +\epsilon(\bv{M} + \lambda\bv{I})\\
	(1-\epsilon)(\bv{M} + \lambda\bv{I}) &\preceq \bv{Z} \bs{\Pi} \bs{\Pi}^T \bv{Z}^T + \lambda \bv I \preceq (1+\epsilon)(\bv{M} + \lambda\bv{I}),
	\end{align*}
	which gives the theorem.
	
	To prove that $\norm{\bv{B} \bs{\Pi} \bs{\Pi}^T \bv{B}^T - \bv{BB}^T}_2 \le \epsilon$ with probability $\ge 1-\delta$ we invoke Theorem 1 of \cite{cohen2016optimal}, which gives that for our setting of $m$, with probability $\ge 1-\delta$:
	\begin{align}\label{srankBound}
	\norm{\bv{B} \bs{\Pi} \bs{\Pi}^T \bv{B}^T - \bv{BB}^T}_2  \le \epsilon \cdot (\norm{\bv B}_2^2 + \norm{\bv B}_F^2/s_\lambda).
	\end{align}
	We have $\norm{\bv B}_2^2 = \norm{(\bv{M}+ \lambda \bv{I})^{-1/2} \bv{M} (\bv{M}+ \lambda \bv{I})^{-1/2}}_2 \le 1$. Additionally, 
	\begin{align*}
	\norm{\bv{B}}_F^2 &= \sum_{i=1}^n \lambda_i \left ((\bv{M}+ \lambda \bv{I})^{-1/2} \bv{M} (\bv{M}+ \lambda \bv{I})^{-1/2}  \right )\\
	&=\sum_{i=1}^n \frac{\lambda_i(\bv{M})}{\lambda_i(\bv{M})+\lambda} = s_\lambda,
	\end{align*}
	giving that $\norm{\bv B}_F^2/s_\lambda = 1$. Thus, by \eqref{srankBound} we have with probability $\ge 1-\delta$, $\norm{\bv{B} \bs{\Pi} \bs{\Pi}^T \bv{B}^T - \bv{BB}^T}_2  \le 2\epsilon$, which completes the theorem after adjusting constants.
\end{proof}

To apply Theorem \ref{thm:rp} to the modified RFF embeddings produced by Corollaries \ref{cor:gaussian} and \ref{cor:cauchy}, we must argue that these embeddings preserve statistical dimension. We do this via an extension of Theorem \ref{thm:rff}. Variants of this type of bound are known in the finite matrix approximation setting (e.g., Lemma 20 of \cite{CohenMuscoMusco:2017}).

\begin{theorem}[Leverage Score Sampling Preserves Kernel Statistic Dimension]\label{thm:rffSD}
	Consider the setting of Theorem \ref{thm:rff}. Letting $s_\lambda(\bv{G}^* \bv{G})$ and $s_\lambda(\bv{K})$ be the $\lambda$-statistical dimensions of $\bv{G}^* \bv{G}$ and $\bv{K}$ respectively, with probability $\ge 1-\delta$ we have:
	$s_\lambda(\bv{G}^* \bv{G}) \le 4 s_\lambda(\bv{K}).$
\end{theorem}
\begin{proof}
	Following Definition \ref{def:rff}, the $j^{th}$ row of $\bv{G}$ is given by $\sqrt{\frac{1}{m \cdot q(\eta_j)}} \cdot \bs{\phi}_{\eta_j}$ where $\bs \phi_{\eta_j} \in \C^{n}$ has $[\bs \phi_{\eta_j}]_k = e^{-2 \pi i \eta_j x_k} \cdot \sqrt{p_k(\eta_j) }$. We can write:
	\begin{align*}
	s_\lambda(\bv{G}^* \bv{G})  &= \tr(\bv{G}^* \bv{G} (\bv{G}^* \bv{G} +\lambda \bv{I})^{-1})\\
	&= \tr(\bv{G} (\bv{G}^* \bv{G} +\lambda \bv{I})^{-1} \bv{G}^*)\\
	& =\frac{1}{m} \sum_{j=1}^m \frac{1}{q(\eta_j)} \cdot \bs{\phi}_{\eta_j}^* (\bv{G}^* \bv{G} +\lambda \bv{I})^{-1} \bs{\phi}_{\eta_j}.
	\end{align*}
	Assuming that the spectral approximation guarantee of Theorem \ref{thm:rff} holds, we have $(\bv{G}^* \bv{G} +\lambda \bv{I})^{-1} \le \frac{1}{1-\epsilon} \bs{\phi}_{\eta_j}^* (\bv{K} +\lambda \bv{I})^{-1} \bs{\phi}_{\eta_j} \le 2 \bs{\phi}_{\eta_j}^* (\bv{K} +\lambda \bv{I})^{-1} \bs{\phi}_{\eta_j}$ if $\epsilon \le 1/2.$ This gives:
	\begin{align*}
	s_\lambda(\bv{G}^* \bv{G})  \le \frac{2}{m} \sum_{j=1}^m \frac{1}{q(\eta_j)} \bs{\phi}_{\eta_j}^* (\bv{K} +\lambda \bv{I})^{-1} \bs{\phi}_{\eta_j} = \frac{2}{m}  \sum_{j=1}^m \frac{\tau_{\lambda,\bv{K}}(\eta_j)}{q(\eta_j)},
	\end{align*}
	where we use that $\tau_{\lambda,\bv{K}}(\eta_j) = \bs{\phi}_{\eta_j}^* (\bv{K} +\lambda \bv{I})^{-1} \bs{\phi}_{\eta_j}$. This is well known in the finite-dimensional setting, and was proven in \cite{AvronKapralovMusco:2017} in the kernel setting.
	Let $S = \frac{2}{m}  \sum_{j=1}^m \frac{\tau_{\lambda,\bv{K}}(\eta_j)}{q(\eta_j)}$. From above with probability $\ge 1-\delta$,  we have $s_\lambda(\bv{G}^* \bv{G})  \le S$. Further:
	$$\E[S] = 2 \E \left [ \frac{\tau_{\lambda,\bv{K}}(\eta_j)}{q(\eta_j)} \right ] = 2 \int_{\eta \in \R} \tau_{\lambda,\bv{K}}(\eta) d\eta = 2s_\lambda(\bv{K}).$$
	Additionally, by design we have chosen $q(\eta) = \frac{\bar \tau_{\lambda,\bv{K}} (\eta)}{T}$ for $T \eqdef \int_{\eta \in \R} \bar \tau_{\lambda,\bv{K}} (\eta)  d\eta $ and $\bar \tau_{\lambda,\bv{K}} (\eta) \ge \tau_{\lambda,\bv{K}} (\eta)$. Thus $\frac{\tau_{\lambda,\bv{K}}(\eta_j)}{q(\eta_j)} \le T$. So by a standard Hoeffding bound, 
	\begin{align*}
	\Pr[ S > 4 s_\lambda(\bv{K})] \le e^{-2 m s_\lambda(\bv{K})^2/T^2} \le e^{-2 m},
	\end{align*}
	since $T =  \int_{\eta \in \R} \bar \tau_{\lambda,\bv{K}} (\eta)  d\eta  \ge  \int_{\eta \in \R} \tau_{\lambda,\bv{K}} (\eta)  d\eta = s_\lambda(\bv{K})$. Finally, since $m = \Omega(\log(1/\delta))$, the bound holds with probability at least $1-\delta$. Overall, via a union bound, we have with probability $1-2\delta$, $s_\lambda(\bv{G}^* \bv{G}) \le S \le 4 s_\lambda(\bv{K})$, completing the proof after adjusting constants on $\delta$.

	%\bs \phi_\eta^* (\bv{K}+\lambda \bv I)^{-1} \bs \phi_\eta
	%where $\bs \phi_\eta \in \C^{n}$ has $\bs \phi_\eta (j) = e^{-2 \pi i \eta x_j} \cdot \sqrt{p_k(\eta) }$
\end{proof}

Combining Theorem \ref{thm:rp} and \ref{thm:rffSD} with Corollaries \ref{cor:gaussian} and \ref{cor:cauchy} gives:
\begin{corollary}[Oblivious Embedding Full Result]
	\label{cor:full_kernel_result}
	Consider any set of points $x_1,\ldots, x_n \in \R$ and an associated Gaussian kernel matrix $\bv{K} \in \R^{n \times n}$. Let $s_\lambda$ be the $\lambda$-statistical dimension of $\bv{K}$, $\bv{G} \in \R^{n \times m'}$ be the modified RFF embedding of Corollary \ref{cor:gaussian}, and $\bs{\Pi} \in \R^{m' \times m}$ have independent sub-Gaussian entries with variance $1/m$. Then for $m' = O\left (\frac{s_\lambda^{5/2} \cdot \log(s_\lambda/\delta)}{\epsilon^2}\right )$ and $m = O \left (\frac{s_\lambda + \log(1/\delta)}{\epsilon^2} \right )$, letting $\bv{Z} = \bv{G}^* \bs{\Pi}$, with probability $\ge 1-\delta$, $\bv{ZZ}^*$ is an $(\epsilon,\lambda)$-spectral approximation of $\bv{K}$.
	%$$(1-\epsilon)(\bv{K} + \lambda\bv{I}) \preceq \bv{ZZ}^* + \lambda\bv{I} \preceq (1+\epsilon)(\bv{K} + \lambda\bv{I}).$$
	The embedding $\bv{z}(x_i) \in \C^m$ can be computed obliviously in $O(m' \cdot m) = \poly(s_\lambda, \log(1/\delta),1/\epsilon)$ time.  
	
	The same bound holds for the Cauchy kernel using the RFF embedding of Corollary \ref{cor:cauchy} with the $m' = O\left (\frac{s_\lambda^{2} \log(s_\lambda) \cdot \log(s_\lambda/\delta)}{\epsilon^2}\right )$.
\end{corollary}

\section{Active Learning -- Omitted Proofs}
\label{app:active}

As discussed in Section \ref{sec:active}, our main active function fitting problem of interest (Problem \ref{prob:unformal_interp}) can be solved via an infinite kernel ridge regression problem (Claim \ref{claim:regression_reduction}.) \cite{AvronKapralovMusco:2019} shows that this problem can in turn be solved approximately with essentially optimal sample complexity by sampling query points according to the kernel operator ridge leverage scores (Definition \ref{def:ridgeScores}). In particular: %the kernel operator statistical dimension $\smu$ essentially characterizes the sample complexity of Problem \ref{prob:unformal_interp}. Under very mild assumptions (see Section 6 of \cite{AvronKapralovMusco:2019} for details), they show that any algorithm solving Problem \ref{prob:unformal_interp} must use $\Omega(\smu)$ samples. Conversely, 
%by sampling data points according to the kernel operator ridge leverage score function (Def.  \ref{def:ridgeScores}), one can achieve a sample complexity nearly matching this lower bound:
\begin{theorem}[Approximate regression via leverage function sampling -- Theorem 6  of \cite{AvronKapralovMusco:2019}]\label{thm:baseSampling}
	Assume that $\lambda \leq \opnorm{\Kmu}$.\footnote{We define the operator norm as $\opnorm{\Kmu} \eqdef \sup_{f \in L_2(p): \norm{f}_p=1} \norm{\Kmu f}_p$.
If $\lambda > \opnorm{\Kmu}$ then \eqref{eq:approx_regress_1} is  solved to a constant approximation factor by the trivial solution $\tilde g = 0$.} Consider a function $\ttmu$ with $\ttmu(x)  \ge \tmu(x)$ for all $x \in \R$, where $\tmu$ is the ridge leverage function of Def. \ref{def:ridgeScores}. Let $T  = \int_{x \in \R} \ttmu(x) dx$ and $m = c \cdot T \cdot \left(\log T + 1/\delta\right)$  for sufficiently large fixed constant $c$. Let $x_1,\ldots,x_m$ be time points sampled independently according to density $h(x) \eqdef \frac{\ttmu(x)}{T}.$ For $j \in 1,\ldots,m$, let $w_j = \sqrt{\frac{p(x_j)}{m \cdot h(x_j)}}$. Let $\bv{F}: \C^m \rightarrow L_2(q)$ be the operator:
	\begin{align*}
		\left[\bv{F} \,\bv{g} \right](\eta) = \sum_{j=1}^m w_j \cdot \bv{g}_j \cdot e^{- 2\pi i \eta x_j}
	\end{align*}
	and $\bv{y},\bv{n} \in \R^m$ be the vectors with $\bv{y}_j = w_j \cdot y(x_j)$ and $\bv{n}_j = w_j \cdot n(x_j)$.
	Let:
	\begin{align}
		\tilde{g} = \argmin_{w \in L_2(q)}\left[  \|\bv{F}^* w - (\bv y+\bv n)\|_2^2 + \lambda \|w\|_q^2\right].\label{eq:approxkrr}
	\end{align}
	With probability $\ge 1- \delta$:
	\begin{align}
	\label{12again2}
		\|\Fmu^* \tilde{g} - (y+n)\|_p^2 + \lambda \|\tilde{g}\|_q^2 \leq 3\min_{w \in L_2(q)}\left[  \|\Fmu^* w - (y+n)\|_p^2 + \lambda\|w\|_q^2\right].
	\end{align}
	\end{theorem}
		Note that via Claim \ref{claim:regression_reduction}, $\Fmu^* \tilde g$ of Theorem \ref{thm:baseSampling} thus solves Problem \ref{prob:unformal_interp} with probability $\ge 1-\delta$ and with error parameters $\lambda' = 6\lambda$ and $C' = 8$. If $\ttmu(x)$ is a tight upper bound on the leverage scores, the sample complexity is near linear in $\smu = \int_{x \in \R} \tmu(x) dx$.
		Also note that the subsampled optimization problem of \eqref{eq:approxkrr} is just a standard kernel ridge regression problem, and thus efficiently solvable. Specifically:
	\begin{claim}\label{clm:krr}
	Consider the set up of Theorem \ref{thm:baseSampling}. Let $k_q: \R \times \R \rightarrow \R$ be the shift-invariant kernel with Fourier transform $q$. Let $\bv{K} \in \R^{m \times m}$ have $\bv{K}_{i,j} = w_i \cdot w_j \cdot k_q(x_i,x_j)$. Then
	$\tilde f = \Fmu^* \tilde{g} $ is given by $\tilde f(x) = \bv{k}(x)^T \bv{z}$ where $\bv{z} = (\bv{K} + \lambda \bv{I})^{-1} (\bv{y}+ \bv{n})$ and $\bv{k}(x) = [w_1 \cdot k_q(x_1,x),\ldots, w_m \cdot k_q(x_m,x)]$.
	\end{claim}
	
\subsection{Kernel operator leverage score bound via Fourier sparse approximation}

We now prove Theorem \ref{thm:css2}, which bounds the kernel operator leverage function $\tmu$ of Definition \ref{def:ridgeScores}, in terms of the Fourier sparse leverage scores for the class $\mathcal T_s$. Combined with Theorem \ref{thm:baseSampling}, Claim \ref{clm:krr}, and Claim \ref{claim:regression_reduction}, this bound will yield our main sample complexity results for Problem \ref{prob:unformal_interp}, which is stated in Corollary \ref{cor:gaussActive2}. % let us solve the active function fitting problem (Problem \ref{prob:unformal_interp}) with near optimal sample complexity, if we can find a sampling distribution $\ttmu$ that tightly upper bounds the true kernel operator leverage function $\tmu$. In this section we show how to do this using the Fourier sparse leverage score bounds of Theorem \ref{thm:gaussian} and \ref{thm:exp}.
%We give a bound based on approximating any function $\Fmu^* w$ via a Fourier sparse function, with sparsity linear in the statistical dimension $\smu$. In particular, we prove the following analog to Theorem \ref{thm:sparseReduction}:
\begin{reptheorem}{thm:css2}[Kernel operator leverage function bound] Let $s = \lceil 36 \cdot \smu \rceil +1$. 
For all $x \in \R$:
	\begin{align*}
	\tmu(x) \le (2+8 \smu) \cdot \tau_{ s, p}(x).
	\end{align*}
\end{reptheorem}
As discussed, Theorem \ref{thm:css2} is analogous to Theorem \ref{thm:sparseReduction} and is proved similarly, by approximating  $\Fmu^* w$ via a Fourier sparse function, with sparsity linear in the statistical dimension $\smu$. In giving this approximation, we use the following continuous analog of Theorem \ref{thm:rss}:
 \begin{theorem}[Frequency subset selection -- Theorem 9 of \cite{AvronKapralovMusco:2019} ]\label{thm:css}
	%Let $\Fmu$ be as in \eqref{eq:informal}.
	For some $s \le \lceil 36 \cdot \smu\rceil$ there exists a set of frequencies $\eta_1,\ldots,\eta_s \in \C$ such that, letting $\bv{C}_s: \C^{s} \rightarrow L_2(p)$ be the operator $[\bv{C}_s \bv{w}](x) = \sum_{j=1}^s \bv{w}_j e^{-2 \pi i \eta_j x} $ and $\bv{Z}:  L_2(q) \rightarrow \C^s$ be the operator $\bv{Z} = (\bv{C}_s^* \bv{C}_s)^{-1} \bv{C}_s^* \Fmu^*$,
	\begin{align}\label{eq:frobNormBound}
	\tr(\Kmu - \bv{C}_s\bv{Z}\bv{Z}^*\bv{C}_s^*)  \le 4\lambda \cdot \smu \text{ and }\bv{Z}^*\bv{C}_s^*  \bv{C}_s\bv{Z} \preceq \Fmu \Fmu^*.
	\end{align}
	Letting $f_x \in L_2(q)$ be given by $f_x(\eta) = e^{2 \pi i x \eta}$ and $\bv{c}_x \in \C^s$ have $j^{th}$ entry $[\bv{c}_x]_j = e^{-2 \pi i \eta_j x}$ we can write: $\tr(\Kmu - \bv{C}_s\bv{Z}\bv{Z}^*\bv{C}_s^*) = \int_{x \in \R} \norm{f_x - \bv{Z}^* \bv{c}_x}_q^2 \cdot p(x) dx$.
\end{theorem}
	 \begin{proof}[Proof of Theorem \ref{thm:css2}]
	 The proof closely follows that of Theorem \ref{thm:sparseReduction}. We can bound the ridge leverage function of Definition \ref{def:ridgeScores} by:
\begin{align}%\label{eq:split1C}
\tmu(x) &= \sup_{w \in L_2(q),\norm{w}_q > 0} \frac {p(x) \cdot |[\Fmu^* w](x)|^2}{\norm{\Fmu^* w}_p^2 + \lambda \norm{w}_q^2}\\
&\le \frac{2p(x) \cdot  |[\bv{C}_s \bv{Z } w](x)|^2}{\norm{\Fmu^* w}_p^2} + \frac{2p(x) \cdot |[\Fmu^* w](x) - [\bv{C}_s \bv{Z } w](x)|^2}{\lambda \norm{w}_q^2}.
 \end{align}
 Since by Theorem \ref{thm:css2}, $\bv{Z}^*\bv{C}_s^*  \bv{C}_s\bv{Z} \preceq \Fmu \Fmu^*$ we have 
 $$\|\Fmu^* w\|_{p}^2 = \langle \Fmu^* w, \Fmu^* w \rangle_p = \langle \Fmu \Fmu^* w, w \rangle_q \ge \langle \bv{Z}^*\bv{C}_s^*  \bv{C}_s\bv{Z}  w, w \rangle_q = \norm{\bv{C}_s\bv{Z}  w}_p^2,$$ which combined with \eqref{eq:split1} gives: 
 \begin{align*}%\label{eq:inter1C2}
\tmu(x) \le \frac{2p(x) \cdot  |[\bv{C}_s \bv{Z } w](x)|^2}{\norm{\bv{C}_s \bv{Z} w}_p^2} + \frac{2p(x) \cdot |[\Fmu^* w](x) - [\bv{C}_s \bv{Z } w](x)|^2}{\lambda \norm{w}_q^2}.
 \end{align*}
 We can observe that $\bv{C}_s \bv{Z} w$ is an $\lceil 36 \cdot \smu \rceil = s-1$ sparse Fourier function in $\mathcal{T}_{s_1}$, giving:
  \begin{align}\label{eq:inter1C}
\tmu(x) \le 2 \tau_{s,p}(x) + \frac{2p(x) \cdot |[\Fmu^* w](x) - [\bv{C}_s \bv{Z } w](x)|^2}{\lambda \norm{w}_q^2}.
 \end{align}
 It thus remains to bound the second term. Let $\bv{c}_x \in \C^s$ have $j^{th}$ entry $[\bv{c}_x]_j = e^{-2 \pi i \eta_j x}$. $\bv{c}_x$ is the `row' of the operator $\bv{C}$ corresponding to $x$ and we have $[\bv{C}_s \bv{Z } w](x) = \bv{c}_x^T \bv{Z } w$. Similarly, let $f_x \in L_2(q)$ be given by $f_x(\eta) = e^{2 \pi i \eta x}$. We can write:
 \begin{align*}
|[ \Fmu^* w](x) - [\bv{C}_s \bv{Z } w](x)|^2 = |\langle f_x - \bv{Z}^* \bv{c}_x, w \rangle_q|^2 \le \norm{f_x - \bv{Z}^* \bv{c}_x}_q^2 \cdot \norm{w}_q^2
 \end{align*}
 via Cauchy-Schwarz. Plugging back into \eqref{eq:inter1C} gives 
 \begin{align}\label{eq:beforeClaim}
 \tmu(x) \le 2 \tau_{s,p}(x) + \frac{2p(x) \cdot  \norm{f_x - \bv{Z}^* \bv{c}_x}_q^2}{\lambda}.
 \end{align}
 The theorem then follows from \eqref{eq:beforeClaim} combined with the following claim:
 \begin{claim}
 $p(x) \cdot  \norm{f_x - \bv{Z}^* \bv{c}_x}_q^2\le \tau_{s,p}(x) \cdot 4 \lambda \smu.$
 \end{claim}
 \begin{proof}
 Let $f_\eta \in L_2(p)$ be given by $f_\eta(x) = e^{2\pi i \eta x}$ and 
let $\bv{z}_\eta \in \C^s$ be the `column' of $\bv{Z}$ corresponding to $\eta$. Formally, as $\bv{Z} =  (\bv{C}_s^* \bv{C}_s)^{-1} \bv{C}_s^* \Fmu^*$, $\bv{z}_\eta = (\bv{C}_s^* \bv{C}_s)^{-1} \bv{C}_s^* f_\eta$. We have $f_\eta - \bv{C}_s \bv{z}_\eta \in \mathcal{T}_s$ and thus:
\begin{align*}
\frac{p(x) \cdot |f_\eta(x) -[ \bv{C}_s \bv{z}_\eta](x)|^2}{\norm{f_\eta - \bv{C}_s \bv{z}}_p^2} \le \tau_{s,p}(x).
\end{align*}
This gives:
\begin{align*}
 p(x) \cdot \int_{\eta \in \R} |f_\eta(x) -[ \bv{C}_s \bv{z}_\eta](x)|^2 q(\eta) d \eta \le \tau_{s,p}(x) \cdot \int_{\eta \in \R}\norm{f_\eta - \bv{C}_s \bv{z}_\eta}_p^2 q(\eta) d \eta.
\end{align*}
Note that $f_\eta(x) = f_x(\eta)$ and $\bv{C}_s \bv{z}_\eta(x) = [\bv{Z}^* \bv{c}_x](\eta)$. Thus we can simplify to:
\begin{align*}
 p(x) \cdot \norm{f_x - \bv{Z}^* \bv{c}_x}_q^2 &\le \tau_{s,p}(x) \cdot \int_{\eta \in \R} \int_{x \in \R} |f_\eta(x)-\bv{C}_s \bv{z}_\eta(x)|^2 p(x) q(\eta) dx d \eta\\
 &= \tau_{s,p}(x) \cdot  \int_{x \in \R} \norm{f_x-\bv{Z}^* \bv{c}_x}_q^2 p(x) dx\\
 &= \tau_{s,p}(x) \cdot \tr(\Kmu - \bv{C}_s\bv{Z}\bv{Z}^*\bv{C}_s^*)\\
 &\le \tau_{s,p}(x) \cdot 4\lambda \smu,
\end{align*}
where the last two bounds follow from Theorem \ref{thm:css}.
 \end{proof}
 \end{proof}

\subsection{Active regression bounds}\label{app:cor}

We conclude by combining the leverage score sampling result of Theorem \ref{thm:baseSampling}, and Claim \ref{clm:krr} with the kernel operator leverage score upper bound of Theorem \ref{thm:css2} to solve Problem \ref{prob:unformal_interp} with sample complexity depending polynomially on the statistical dimension $\smu$. Our main result is summarized in Corollary \ref{cor:gaussActive} of Section \ref{sec:active}, and stated in full detail below.

\begin{corollary}[Active Function Fitting -- Gaussian or Exponential Density]\label{cor:gaussActive2}
Consider the active regression set up of Problem \ref{prob:unformal_interp}.
Let $p$ be the Gaussian density $p(x) = \frac{1}{\sigma  \sqrt{2\pi}} e^{-x^2/(2\sigma^2)}$.

 For any frequency density $q$ and $0 < \lambda < \opnorm{\Kmu}$, let $\smu$ be the $\lambda$-statistical dimension of $\Kmu$. Let $s = \lceil 36 \smu \rceil +1$ and let $\bar \tau_{s,p}(x)$ be the leverage score bound of Theorem \ref{thm:gaussian}.
	Let 
	$m = c \cdot \smu^{5/2} \cdot \left(\log \smu  + 1/\delta\right)$  for a sufficiently large constant $c$. Let $x_1,\ldots,x_m$ be time points sampled independently according to the density proportional to $\bar \tau_{s,p}(x)$ and let $\tilde y$ be computed from these points using kernel ridge regression according to the procedure of Theorem \ref{thm:baseSampling} and Claim \ref{clm:krr}.
	 Then with probability $\ge 1- \delta$:
	\begin{align}
	\label{12again2}
		\|y - \tilde y\|_p^2 \le  8 \norm{n}_p^2 + 6 \lambda \|{g}\|_q^2.
	\end{align}
	
	An identical bound holds when $p$ is the Laplacian density $p(x) = \frac{1}{\sqrt{2}\sigma} e^{-|x| \sqrt{2}/\sigma}$, $\bar \tau_{s,p}(x)$ is the leverage score bound of Theorem \ref{thm:exp}, and $m = c \cdot \smu^{2} \cdot \left(\log \smu  + 1/\delta\right)$.
\end{corollary}

As discussed in Section \ref{sec:klsbv}, the sample complexity bounds of Corollary \ref{cor:gaussActive2} can be improved to near linear in $\smu$ by simply applying a second sampling step to the final kernel ridge regression problem of Claim \ref{clm:krr}, using the ridge leverage scores of the finite kernel matrix $\bv{K}$ \cite{Sarlos:2006,DrineasMahoneyMuthukrishnan:2006}. This is analogous to the final finite-dimensional random projection discussed in Section \ref{sec:rffApproach}. A full proof requires an extension of Theorem \ref{thm:baseSampling}, which applies to an approximate solution of the finite ridge regression problem. This extension was shown in \cite{AvronKapralovMusco:2019}.

\section{Empirically Estimating the Leverage Scores}
\label{app:plot_info}
The main technical challenge of this paper is to prove rigorous upper bounds on the leverage scores of a function class $\mathcal{F}$, under a distribution $p$. To do so, it is useful to have a way of empirically estimating the \emph{true} leverage function $\tau_{\mathcal{F}, p}$. Such an estimate may not be accurate for all $x$, and it may not have a closed-form. However, a good enough estimate can serve as guidance in proving theoretically sound bounds.

For some function classes (e.g., low-degree polynomials) establishing an empirical estimate for $\tau_{\mathcal{F}, p}(x)$ is straight-forward. The class of sparse Fourier functions, $\mathcal{T}_s$, studied in this paper presents a somewhat greater challenge, but we are able to obtain relatively good estimates, including those used to plot Figure \ref{fig:upper_bound}. In this section we briefly discuss our approach, which might be useful for future work, for example on other distributions beyond Gaussian and Laplace. MATLAB code for reproducing Figure \ref{fig:upper_bound} can be found in \texttt{empirical\_upper\_bounds.m} of the supplemental. 

The key observation is that the function class $\mathcal{T}_k$ is a union of linear subspaces, and for each subspace, it is possible to relatively easily approximate the true leverage scores. In particular, for any \emph{fixed} choice of frequencies $\lambda_1, \ldots, \lambda_k \in \R$, consider the function class:
\begin{align*}
\mathcal{T}_{\lambda_1, \ldots, \lambda_k} = \left \{f: f(x) = \sum_{j=1}^k a_j e^{i \lambda_j x}, a_j \in \C \right \}.
\end{align*}
For any fixed set of frequencies, $\mathcal{T}_{\lambda_1, \ldots, \lambda_k}$ is a subset of $\mathcal{T}_{k}$ and 
\begin{align*}
\mathcal{T}_{k} = \bigcup_{\lambda_1, \ldots, \lambda_k \in \R} \mathcal{T}_{\lambda_1, \ldots, \lambda_k}.
\end{align*}
So, if we let $\tau_{\lambda_1, \ldots, \lambda_k,p}(x)$ denote the leverage score of $\mathcal{T}_{\lambda_1, \ldots, \lambda_k}$, then the leverage scores of $\mathcal{T}_{k}$ equal:
\begin{align}
\label{eq:sup_calc}
\tau_{k,p}(x) = \sup_{\lambda_1, \ldots, \lambda_k \in \R}\tau_{\lambda_1, \ldots, \lambda_k,p}(x).
\end{align}
This equation is useful because, for any fixed $\lambda_1, \ldots, \lambda_2$, the right hand side is actually relatively easy to approximate. In particular, any function $f$ in $\mathcal{T}_{\lambda_1, \ldots, \lambda_k}$ can be written as $\mathcal{A}\alpha$ where $\alpha \in \C^k$ and $\mathcal{A}$ is an infinite dimensional linear operator with $k$ columns, the $j^\text{th}$ being equal to $e^{i\lambda_j x}$. I.e., $\mathcal{T}_{\lambda_1, \ldots, \lambda_k}$ is a $k$ dimensional linear subspace. If we are estimating the leverage scores with respect to  distribution $p$, let $\bar{\mathcal{A}}_p$ be the rescaled linear operator with $j^\text{th}$ column equal to $e^{i\lambda_j x}\sqrt{p}$. We have
\begin{align}
\label{eq:lin_score}
\tau_{\lambda_1, \ldots, \lambda_k,p}(x) = \sup_{\alpha\in \C^k} \frac{|\bar{\mathcal{A}}_p\alpha(x)|^2}{\norm{\bar{\mathcal{A}}_p\alpha}_2^2}.
\end{align}
It is well know that the optimal $\alpha$ for maximizing \eqref{eq:lin_score} can be obtain by setting $\alpha = (\bar{\mathcal{A}}_p^*\bar{\mathcal{A}}_p)^{-1}\bar{\mathcal{A}}_p(x)$ where $\bar{\mathcal{A}}_p^*$ is the adjoint operator of $\bar{\mathcal{A}}_p$ \cite{AvronKapralovMusco:2017,Bach:2017,AvronKapralovMusco:2019}. This leads to a leverage score of $\tau_{\lambda_1, \ldots, \lambda_k,p}(x) = \bar{\mathcal{A}}_p(x)^*(\bar{\mathcal{A}}_p^*\bar{\mathcal{A}}_p)^{-1}\bar{\mathcal{A}}_p(x)$, where $\bar{\mathcal{A}}_p(x)^*$ is the conjugate transpose of the $k$ length vector $\bar{\mathcal{A}}_p(x)$.
While these expression involves infinite dimensional operators indexed by values in $\R$, they can be very well approximated for any $x$ discretizing $\bar{\mathcal{A}}_p$ to a finite number of rows. Specifically, $\bar{\mathcal{A}}_p$ is replaced with a matrix $\bar{{A}}_p$ with rows indexed $t\in \{-R, -R + \Delta,-R + 2\Delta, \ldots, R - \Delta, R\}$, each equal to $\begin{bmatrix}e^{i\lambda_1t}\sqrt{p(t)/\Delta} & \ldots & e^{i\lambda_kt}\sqrt{p(t)/\Delta}\end{bmatrix}$ and we can approximate $\alpha \approx  (\bar{{A}}_p^*\bar{{A}}_p)^{-1}\bar{{A}}_p(x)$ for any given $x$. The leverage score is approximated as $\tau_{\lambda_1, \ldots, \lambda_k,p}(x) \approx \bar{{A}}_p(x)^* (\bar{{A}}_p^*\bar{{A}}_p)^{-1}\bar{{A}}_p(x)$

With these equations in hand, our full approach for estimating $\tau_{k,p}(x)$ for a given $x$ is:
\begin{itemize}
	\item Set $\tau_{k,p}(x) = 0$.
	\item For $iter = 1,\ldots, N$
	\begin{itemize}
	\item Randomly select $k$ frequencies $\lambda_1, \ldots, \lambda_k \in \R$. 
	\item Approximately compute $\tau_{\lambda_1, \ldots, \lambda_k,p}(x)$ via discretization.
	\item Set $\tau_{k,p}(x) = \max(\tau_{k,p}(x),\tau_{\lambda_1, \ldots, \lambda_k,p}(x))$.
	\end{itemize}
\end{itemize}
To ensure this approach obtains a good approximation, it is important that the method for randomly selecting subsets of $k$ frequencies provides good ``coverage'', as different frequency subsets can lead to very different values of $\tau_{\lambda_1, \ldots, \lambda_k,p}(x)$. One point to note is that, as frequencies become far apart, the columns of $\bar{\mathcal{A}}_p$ become close to mutual orthogonal, and the leverage scores converge to the squared $\ell_2$ norms of the rows of $\bar{\mathcal{A}}_p$, which equal $k\cdot p(x)$ for any $x$. This means that the benefit of considering subsets involving distance frequencies is marginal, as such subsets always lead to approximately the same scores. So, we can focus on sampling values of $\lambda_1, \ldots, \lambda_k$ that are relatively close together. 

To generate the plots of Figure \ref{fig:upper_bound}, we do so via independent sampling. At each iteration, a random order of magnitude $h$ was chosen on a geometric grid between $.01$ and $10$ and $\lambda_1, \ldots, \lambda_k$ where chosen as random Gaussians with variance $h$. A large number of iterations (10 million) was run, and the range of $h$ was increased until doing so had no noticeable effect on the estimate for $\tau_{k,p}(x)$. This leaves us reasonable confident that the curves of Figure \ref{fig:upper_bound} accurately reflect the true leverage scores, although we of course can not be sure, as the method is only heuristic. 

\end{document}

%% file: bib_macros.tex
  \usepackage{nth}
  \usepackage{intcalc}